\documentclass[11pt]{article}

\usepackage{url} 
\usepackage[hidelinks]{hyperref}
\usepackage[utf8]{inputenc}
\usepackage[small]{caption}
\usepackage{helvet}
\usepackage{courier}
\usepackage{amsfonts}
\usepackage{amsmath}
\usepackage[square,numbers,sort&compress]{natbib}
\usepackage{amsthm}
\usepackage{graphicx}
\usepackage{mathdots}
\usepackage[margin=1in]{geometry}
\usepackage{kbordermatrix}
\usepackage[capitalize]{cleveref}
\usepackage{comment}

\newtheorem{theorem}{Theorem}
\newtheorem{proposition}[theorem]{Proposition}
\newtheorem{observation}[theorem]{Observation}

\newtheorem{definition}[theorem]{Definition}
\newtheorem{lemma}[theorem]{Lemma}
\newtheorem{example}[theorem]{Example}
\usepackage{color}
\usepackage{amssymb}
\usepackage{nicefrac}
\usepackage{booktabs}
\newcommand{\Ret}{\mathcal{R}}

\makeatletter
\def\moverlay{\mathpalette\mov@rlay}
\def\mov@rlay#1#2{\leavevmode\vtop{%
   \baselineskip\z@skip \lineskiplimit-\maxdimen
   \ialign{\hfil$\m@th#1##$\hfil\cr#2\crcr}}}
\newcommand{\charfusion}[3][\mathord]{
    #1{\ifx#1\mathop\vphantom{#2}\fi
        \mathpalette\mov@rlay{#2\cr#3}
      }
    \ifx#1\mathop\expandafter\displaylimits\fi}
\makeatother

\newcommand{\cupdot}{\charfusion[\mathbin]{\cup}{\cdot}}

\usepackage{caption}
\usepackage{subcaption}

\newcommand{\unknownn}{\textsf{unknown}^{2n}}

\usepackage{todonotes}
\newcommand{\ssnote}[1]{}
\newcommand{\ssalert}[1]{}
\newcommand{\task}[1]{}

\newcommand{\ID}{\text{ID}}
\newcommand{\IDn}{\ensuremath{\ID^{2n}}}

\newcommand{\sym}{\text{MA}}
\newcommand{\symn}{\ensuremath{\sym^{2n}}}

\newcommand{\asym}{\text{MD}}
\newcommand{\asymn}{\ensuremath{\asym^{2n}}}

\newcommand{\assymn}{\ensuremath{\text{A-S-SYM}^{2n}}}

\usepackage{titlesec}

\newcommand{\chaos}{\text{CH}}
\newcommand{\unknown}{\chaos}
\newcommand{\unn}{\ensuremath{\unknown^{2n}}}

\DeclareMathOperator{\dist}{d}
\DeclareMathOperator{\nordist}{nd}

\newcommand{\retrodist}{\dist_{\operatorname{MAD}}}

\newcommand{\norretrodist}{\nordist_{\operatorname{MAD}}}

\newcommand{\normphi}{{{\mathrm{norm}\hbox{-}\phi}}}

\newcommand{\swap}{{{\mathrm{swap}}}}
\newcommand{\spear}{{{\mathrm{spear}}}}
\newcommand{\pos}{{\mathrm{pos}}}
\newcommand{\ag}{{\mathrm{ag}}}
\newcommand{\MA}{\mathcal{MA}}

\newcommand{\std}{\sigma}
    
\usepackage{wrapfig}

\usepackage{thmtools} 
\usepackage{thm-restate}
\usepackage{amssymb}

\usepackage{mathtools}
\DeclarePairedDelimiter\floor{\lfloor}{\rfloor}
\newcommand{\speardist}{d_{\spear}}
\usepackage{authblk}

\title{A Map of Diverse Synthetic Stable Roommates Instances}
\author[1]{Niclas Boehmer}
\author[1]{Klaus Heeger}
\author[2]{Stanisław Szufa}
\affil[1]{\small
  Technische Universit\"at Berlin, Algorithmics and Computational 
  Complexity\protect\\
  \{niclas.boehmer,heeger\}@tu-berlin.de}
  \affil[2]{\small
  AGH University, Kraków, Poland
 szufa@agh.edu.pl}
\date{\today}

\begin{document}

\maketitle

\begin{abstract}
Focusing on \textsc{Stable Roommates} (SR) instances, we contribute to the toolbox for conducting experiments for stable matching problems. We introduce a polynomial-time computable pseudometric to measure the similarity of SR instances, analyze its properties, and use it to create a map of SR instances. This map visualizes $460$ synthetic SR instances (each sampled from one of ten different statistical cultures) as follows: Each instance is a point in the plane, and two points are close on the map if the corresponding SR instances are similar to each other. Subsequently, we conduct several exemplary experiments and depict their results on the map, illustrating the map's usefulness as a non-aggregate visualization tool, the diversity of our generated dataset, and the need to use instances sampled from different statistical cultures.
Lastly, to demonstrate that our framework can also be used for other matching problems under preference, we create and analyze a map of \textsc{Stable Marriage} instances.
\end{abstract}

\section{Introduction}

Since their introduction by \citet{DBLP:journals/tamm/GaleS13}, stable matching problems have been extensively studied, both from a theoretical and a practical viewpoint. 
Numerous practical applications have been identified, and theoretical research has influenced the design of real-world matching systems~\cite{DBLP:books/daglib/0066875,Knuth76,DBLP:books/ws/Manlove13}. 
In addition to the rich theoretical literature, there are also several works containing empirical investigations of stable matching problems (see \cite{DBLP:journals/cor/PetterssonDGGKM21,DBLP:journals/eor/DelormeGGKMP19,DBLP:journals/jea/IrvingM09,DBLP:conf/or/KwanashieM13,Filho2016Automatic,DBLP:journals/tplp/0001FMP20,DBLP:conf/wea/CooperM20,DBLP:journals/corr/abs-1905-06626,DBLP:journals/corr/abs-2201-12484,DBLP:conf/aaai/TziavelisGJDKK20,DBLP:journals/dam/ManloveMO22,mertens2015stable,DBLP:journals/eor/TeoS00,DBLP:conf/cpaior/Genc0SO19,DBLP:conf/ijcai/Genc0OS17,DBLP:conf/cp/0002O17,DBLP:conf/siamcsc/ManneNLH16,DBLP:journals/corr/abs-2112-05777,DBLP:journals/corr/abs-2007-04948,DBLP:journals/corr/abs-2204-04162} as a certainly incomplete list). 
Although these examples indicate that experimental works regularly occur, many papers on stable matchings do not include an experimental part and instead solely focus on the computational or axiomatic aspects of some mechanism or problem. 
However, to understand the properties of problems and mechanisms in practice, experiments are vital. 

One reason for the lack of experimental work might be the rarity of real-world data (exceptions can be found in \cite{DBLP:journals/jea/IrvingM09,DBLP:journals/eor/DelormeGGKMP19,DBLP:journals/dam/ManloveMO22}). 
Consequently, researchers typically resort to some synthetic distribution, refereed to as a statistical culture, for generating synthetic data.
Remarkably, the vast majority of works simply use random preferences where all possible valid preferences are sampled with the same probability 
(out of the twenty works listed above, fourteen use this model, most of them as a single data source).
However, as we will see later, instances with random preferences  have very similar properties. 
Accordingly, conclusions drawn from experiments using only such instances (or, generally speaking, only instances sampled from one model) should be treated with caution, as it is unclear whether their results generalize.

With our work, we want to lay the foundation for more experimental work around stable matchings by introducing a measure for the similarity of instances and by creating a diverse synthetic dataset for testing together with a convenient framework to visualize and analyze it as a map (see \Cref{fig:mainMap} for an example).
We focus on instances of the \textsc{Stable Roommates} (SR)  problem, where we have a set of agents, and each agent has strict preferences over all other agents. 
We selected the SR problem for this first, exemplary study because it is the mathematically most natural stable matching problem (agents' preferences do not contain ties and are complete, and there are no different ``types" of agents). 
Consequently, statistical cultures for SR instances are relatively simple and do not need to distinguish between different types of agents.
Nevertheless, our general approach and several of our ideas and techniques can also be used to carry out similar studies for other stable matching problems, as demonstrated in \Cref{sec:map-of-SM}. 

As part of our agenda to empower experimental work on stable matchings, we carry out the following steps: 

\paragraph{Distances Between SR Instances (\Cref{sec:SR-dist}).} 
To judge the diversity of a dataset for testing and to compare different statistical cultures to each other, a similarity measure is needed. 
We introduce a notion of isomorphism between SR instances and show how distances between preference orders naturally extend to distances between SR instances. 
Most importantly, we propose the polynomial-time computable mutual attraction distance\footnote{Note that we use the terms ``distance (measure)'' in an informal sense to refer to some function mapping pairs of instances to a positive real number; in particular, all our distance measures are pseudometrics but not all are metrics.}, which we use in the following. 

\paragraph{Understanding the Space of SR Instances (\Cref{sec:understanding}).}
To better understand the space of SR instances induced by our mutual attraction distance, we introduce four canonical ``extreme" instances, which are far away from each other. 
Moreover, we prove that two of them form a diameter of our space, i.e., they are at the maximum possible distance.

\paragraph{A Map of Synthetic SR Instances (\Cref{sec:map}).}
We define a variety of statistical cultures to generate SR instances. 
From them, we generate a diverse test set for experimental work and picture it as a map of SR instances, a convenient framework to visualize non-aggregate experimental results. 
Moreover, we give intuitive interpretations of the different areas on the map. 
In addition, we analyze where different statistical cultures land on the map and how they relate to each other.

\paragraph{Using the Map of SR Instances (\Cref{sec:using}).}
To demonstrate possible use cases for the map, we perform exemplary experimental studies.
We analyze different quality measures for stable matchings, the number of blocking pairs for random/minimum-weight matchings, and the running time to compute an ``optimal" stable matching using an ILP. 
In sum, the instance-based view on experimental results provided by the map allows us to identify several interesting phenomena, for example, that instances sampled from the same culture all behave very similarly in our experiments. 
Moreover, we observe that instances from the same area of the map exhibit a similar behavior.  

\paragraph{Outlook: A Map of \textsc{Stable Marriage} Instances (\Cref{sec:map-of-SM}).}
To demonstrate the general applicability of our framework to draw maps of instances of other stable matching problems, we create a map of \textsc{Stable Marriage} (SM) instances---SM is the bipartite analogue of SR.
For this, we describe how to transfer the mutual attraction distance, extreme instances, and statistical cultures from the SR to the SM setting. 
Notably, the resulting drawn map of SM instances looks quite similar to the one for SR instances.
Finally, we illustrate the usefulness of the map of SM instances and verify that instances that are close to each other on the map have similar properties by conducting some exemplary experiments. 

\bigskip

From a methodological perspective, our work follows a series of recent papers on (ordinal) elections \cite{DBLP:conf/aaai/FaliszewskiSSST19,DBLP:conf/atal/SzufaFSST20,DBLP:conf/ijcai/BoehmerBFNS21}:
\citet{DBLP:conf/aaai/FaliszewskiSSST19} introduced the problem of computing the distance between elections, focusing on isomorphic distances. 
Following up on this, \citet{DBLP:conf/atal/SzufaFSST20} created a dataset of synthetic elections sampled from a variety of different cultures and visualized them as a map of elections. 
Subsequently, \citet{DBLP:conf/ijcai/BoehmerBFNS21} added several canonical elections to the map to give absolute positions a clearer meaning, and added some real-world elections.
Recently, \citet{DBLP:map-approval} created and analyzed a map of approval elections.
The usefulness of the maps has already been demonstrated in different contexts.
For example, \citet{DBLP:conf/atal/SzufaFSST20} identified  
that for elections from a certain region of the map, election winners are particularly hard to compute,
\citet{DBLP:conf/ijcai/BoehmerBFNS21} and \citet{DBLP:journals/corr/abs-2204-03589} analyzed the nature and relationship of real-world elections by placing them on the map, and
\citet{DBLP:conf/ijcai/BoehmerBFN21} evaluated the robustness of election winners using the map.
Although our general agenda and approach are similar to the works of \citet{DBLP:conf/aaai/FaliszewskiSSST19,DBLP:conf/atal/SzufaFSST20} and \citet{DBLP:conf/ijcai/BoehmerBFNS21}, the intermediate steps, used distance measures, cultures, experiments, and technical details are naturally quite different. 

\medskip 
The code for generating the map and conducting our experiments is available at \url{https://github.com/szufix/mapel}.
The generated datsets of \textsc{SR} and \textsc{SM} instances is available at \url{https://github.com/szufix/mapel_data}. 

\section{Preliminaries} \label{sec:prel}

We define some concepts and notation here and some in the corresponding sections.
For a positive integer $n\in \mathbb{N}$, let $[n]:=\{1,\dots,n\}$. 
For two real-valued vectors $x = (x_1, \ldots, x_n)$ and $y = (y_1, \ldots, y_n)$ and some $p\in \mathbb{R}$, their $\ell_p$-distance is 
$\ell_p(x,y) := (|x_1-y_1|^p + \cdots + |x_n-y_n|^p)^{\frac{1}{p}}$.

\paragraph{Preference Orders.}
Let $A$ be a set of agents. 
We denote by $\mathcal{L}(A)$ the set of all total orders over $A$ to which we refer to as preference orders.
We usually denote elements of $\mathcal{L}(A)$ as $\succ$ and for three agents $a$, $b$, and $c$, we say that $a$ is preferred to $b$ is preferred to $c$ if $a\succ b \succ c$.
Moreover, for a preference order~${\succ}\in \mathcal{L}(A)$ and an agent $a\in A$, let $\pos_{\succ}(a)$ denote the position of $a$ in $\succ$, i.e., the number of agents that are preferred to $a$ in $\succ$ plus one. 
Furthermore, for $i\in [|A|]$, let $\ag_{\succ}(i)$ be the agent ranked in $i$-th position in $\succ$, i.e., the agent $b\in A$ such that $i=\pos_{\succ}(b)$.

\paragraph{Distances Between Preference Orders.}
For two preference orders $\succ,\succ'\in \mathcal{L}(A)$, their swap distance $\swap(\succ,\succ')$ is the number of agent pairs on whose ordering $\succ$ and $\succ'$ disagree. 
Alternatively, the swap distance can also be interpreted as the minimum number of swaps of adjacent agents that are necessary to transform $\succ$ into $\succ'$. 
For two preference orders $\succ,{\succ'} \in \mathcal{L(A)}$, their Spearman distance $\spear(\succ,\succ')$ is $\sum_{a\in A} | \pos_{\succ}(a)- \pos_{\succ'}(a)|$.
As proven by \citet{diaconis1977spearman}, it holds that $\swap(\succ,\succ')\leq \spear(\succ,\succ')\leq 2\cdot \swap(\succ,\succ')$.

\paragraph{Stable Roommates Instances.}
A \textsc{Stable Roommates} (SR) instance $\mathcal{I}$ consists of a set $A$ of agents, with each agent $a\in A$ having a preference order ${\succ_a}\in \mathcal{L}(A\setminus \{a\})$ over all other agents. 
For the sake of simplicity, we will focus on instances with an even number of agents.

\paragraph{Stable Matchings.}
A matching of agents $A$ is a subset of agent pairs $\{a,a'\}$ with $a\neq a'\in A$ where each agent appears in at most one pair.
We say that an agent is unmatched in a matching $M$ if $a$ does not appear in any pair from $M$;
otherwise, we say that $a$ is matched.
For a matched agent $a\in A$ and a matching $M$, we write $M(a)$ to denote the partner of $a$ in $M$, i.e., $M(a)=a'$ if $\{a,a'\}\in M$.
A pair $\{a,a'\}$ of agents blocks a matching $M$ if $a$ is unmatched or prefers $a'$ to $M(a)$ and $a'$ is unmatched or prefers $a$ to~$M(a')$. 
A matching that is not blocked by any agent pair is called a stable matching. 

\paragraph{Mappings between SR Instances.}
For two sets $X$ and $Y$ with $|X|=|Y|$, we denote by~$\Pi(X,Y)$ the set of all bijections $\sigma: X\to Y$ between $X$ and $Y$.
Let $A$ and $A'$ be two sets of agents with $|A|=|A'|$ and let $\sigma\in \Pi(A,A')$. 
Then, for an agent $a\in A$ and a preference order~${\succ_a}\in \mathcal{L}(A\setminus \{a\})$, we write $\sigma(\succ_a)$ to denote the preference order over $A'\setminus \{\sigma(a)\}$ arising from~$\succ_a$ by replacing each agent $b\in A\setminus \{a\}$ by~$\sigma(b)\in A' \setminus \{\sigma(a)\}$.

\paragraph{Pseudometrics.} 
We call a function $d: X\times X\mapsto \mathbb{R}$ a pseudometric if for each three elements~$x,y,z\in X$, we have that $d(x,y)=d(y,x)\geq 0$, $d(x,x)=0$, and $d(x,z)\leq d(x,y)+d(y,z)$. 

\section{Distance Measures} \label{sec:SR-dist}
This section is devoted to measuring the distance between two SR instances, a key ingredient of our map. 
Other use cases include meaningfully selecting test instances, comparing different statistical cultures, and analyzing real-world instances. 
Specifically, in \Cref{sub:iso-dis}, we define an isomorphism between two SR instances, show how distance measures over preferences orders can be generalized to distance measures over SR instances, and prove that computing the Spearman distance between SR instances is computationally intractable. 
In \Cref{sub:MAD}, we introduce our mutual attraction distance and make some observations concerning its properties and the associated mutual attraction matrices.

\subsection{Isomorphism and Isomorphic Distances} \label{sub:iso-dis}

Two SR instances are isomorphic if renaming the agents in one instance can produce the other instance. 
For this, as each agent is associated with a preference order defined over other agents, a single mapping suffices. 
Accordingly, we define an isomorphism on SR instances: 
\begin{definition}
	Two SR instances $(A,(\succ_a)_{a\in A})$ and $(A',(\succ_{a'})_{a'\in A'})$ with $|A|=|A'|$ are isomorphic if there is a bijection $\sigma:A\rightarrow A'$ such that $\succ_{\sigma(a)}=\sigma(\succ_a)$ for all $a\in A$. 
\end{definition}
\begin{example}\label{ex:iso}
	Let $\mathcal{I}$ with agents $a, b, c$, and $d$ and $\mathcal{I}'$ with agents $x,y,z$, and $w$ be two SR instances with the following preferences:
	\begin{align*}
	a&: b\succ c \succ d,  &b&: c\succ a \succ d, &c&: b\succ d \succ a, &d&: a\succ c \succ b,\\
	x &:  y\succ w \succ z, &y&:  z\succ w \succ x, &z&:  w\succ y \succ x, &w&: z\succ x \succ y.
	\end{align*}
	
	\noindent $\mathcal{I}$ and $\mathcal{I}'$ are isomorphic as witnessed by the mapping $\sigma(a)=y$, $\sigma(b)=z$, $\sigma(c)=w$, and $\sigma(d)=x$. 
\end{example}

One can easily check whether two SR instances  $(A,(\succ_a)_{a\in A})$ and  $(A',(\succ_{a'})_{a'\in A'})$ are isomorphic:
Assuming that an isomorphism~$\sigma^{a'}$ maps~$a \in A$ to $a' \in A'$, then this already completely characterizes~$\sigma^{a'}$, as for any~$b \in A \setminus \{a\}$ with $\pos_{\succ_a}(b) = i$, we must have $\sigma^{a'} (b) = \ag_{\succ'_{a'}} (i)$.
Thus, it suffices to fix an arbitrary agent~$a \in A$ and then check for each~$a' \in A'$ whether $\sigma^{a'}$ is an isomorphism.
\begin{observation}
    Deciding whether two SR instances with $2n$ agents are isomorphic can be done in~$\mathcal{O}(n^3)$~time.
\end{observation}

For each distance measure $p$ between preference orders, our notion of isomorphism can be easily used to extend~$p$ to a distance measure over SR instances:
The resulting distance between two SR instances is the minimum (over all bijections~$\sigma$ between the agent sets) sum (over all agents) of the distance between the preferences of~$a \in A$ and the preferences of~$\sigma (a)$ (measured by~$p$):
\begin{definition} \label{def:iso-dis}
	Let $p$ be a distance measure between preference orders. 
 Let $\mathcal{I}=(A,(\succ_a)_{a\in A})$ and  $\mathcal{I}'=(A',(\succ_{a'})_{a'\in A'})$ be two SR instances with $|A|=|A'|$. 
 Their $d_p$ distance  is: 
 $ d_{p}(\mathcal{I},\mathcal{I}'):=\min_{\sigma\in \Pi(A,A')} \sum_{a\in A} p(\sigma(\succ_a),\succ_{\sigma(a)}).$
\end{definition}
In particular, for all distance measures $p$ between preference orders where $p(x,y)=0$ if and only if~${x=y}$, the distance $d_p$ induces a metric on the equivalence classes defined by isomorphisms (i.e., when two matrices are equivalent if and only if they are isomorphic).
In other words, this means that for any two SR instances $\mathcal{I}$ and $\mathcal{I}'$ it holds that $d_p(\mathcal{I},\mathcal{I}')=0$ if and only if $\mathcal{I}$ and $\mathcal{I}'$ are isomorphic.
We will call such a distance also an \emph{isomorphic distance}.
\begin{example} \label{ex:2}
	Applying \Cref{def:iso-dis}, the Spearman distance $\spear(\cdot, \cdot)$  and the swap distance~$\swap(\cdot, \cdot)$ between preference orders (as defined in \Cref{sec:prel}) can be lifted to distance measures~$d_{\spear}$ and $d_{\swap}$ between SR instances. Let $\mathcal{I}$ with agents $a, b, c$, and $d$ and $\mathcal{I}'$ with agents~$x,y,z$, and $w$ be two SR instances with the following preferences: 
	\begin{align*}
	a&: b\succ c \succ d,  &b&: a\succ c \succ d, &c&: a\succ b \succ d, &d&: a\succ b \succ c,\\
	x &:  y\succ z \succ w, &y&:  x\succ z \succ w, &z&:  w\succ y \succ x, &w&: z\succ y \succ x.
	\end{align*}
	
	\noindent Then, for the mapping $\sigma(a)=x$, $\sigma(b)=y$, $\sigma(c)=z$, and $\sigma(d)=w$, the Spearman distance of~$\mathcal{I}$ and~$\mathcal{I}'$ is $8$ and the swap distance is $6$. While for the Spearman distance this is the optimal mapping (so $d_{\spear}(\mathcal{I},\mathcal{I}')=8$) for the swap distance the mapping $\sigma(a)=y$, $\sigma(b)=x$, $\sigma(c)=z$, and $\sigma(d)=w$ results in a smaller distance of~$4$. Indeed, we have $d_{\swap}(\mathcal{I},\mathcal{I}')=4$.
	
\end{example}

We consider the Spearman distance $d_{\spear}$ and the swap distance $d_{\swap}$ as ``ideal'' distances, as they are quite fine-grained and \emph{isomorphic} distances.
Unfortunately, both are hard to compute. 
For $d_{\swap}$ this follows from the NP-hardness of computing the Kemeny score of an election~\cite{DBLP:conf/www/DworkKNS01}. 
As we focus in the following only on the in some sense simpler
Spearman distance~$d_{\spear}$, here we only present that computing the Spearman distance between two SR instances is at least as hard as deciding whether two graphs are isomorphic, which is a famous candidate for the complexity class NP-intermediate.  

\begin{restatable}{proposition}{spearrr} \label{pr:spear}
	There is no polynomial-time algorithm to compute $d_{\spear}$, unless the \textsc{Graph Isomorphism} problem is in P. 
\end{restatable}
\begin{proof}
	For a graph $G=(V,E)$ and a vertex $v\in V$, let $N_G(v)$ be the set of vertices adjacent to $v$ in $G$. 
	In the \textsc{Graph Isomorphism} problem we are given two graphs $G=(V,E)$ and $G'=(V',E')$ with $|V|=|V'|$ and the question is whether there is a bijection $\mu: V \to V'$ such that $\{v,v'\}\in E$ if and only if $\{\mu(v),\mu(v')\}\in E'$. 
	We will now reduce \textsc{Graph Isomorphism} to the problem of computing~$d_\spear$.
	
	\paragraph{Construction.}
	Given an instance $(G=(V,E),G'=(V',E'))$ of \textsc{Graph Isomorphism}, we construct two SR instances as follows.  
	Without loss of generality, we assume that there are no isolated vertices in $G$ and $G'$ and that $\nu=|V|=|V'|>2$. 
	From $G$, we construct an SR instance~$\mathcal{I}$ with agent set~$A$ as follows:
	First, we add each vertex~$v \in V$ as an agent to~$A$. 
	Moreover, we add a set $D$ of $\nu^4$ dummy agents. 
	We now describe the preferences of the agents.
    In order to do so, we denote for a set $B$ of agents by~$[B]$ an arbitrary but fixed total order of agents from $B$. 
    For an agent~$b \in B$, we denote by~$[B]\setminus \{b\}$ the order arising from~$[B]$ through the deletion of~$b$.
	The preferences of the agents are as follows:
	\begin{align*}
	v & : [N_G(v)]\succ [D] \succ [V\setminus ( N_G(v)\cup \{v\})] \qquad & \forall v\in V \\
	d & : [D] \setminus \{d\}\succ [V] \qquad & \forall d\in D
	\end{align*}
	
	From $G'=(\{v'_1,\dots, v'_{\nu}\},E')$, we construct a second SR instance $\mathcal{I}'$ with agent set $A'$.
	We add each vertex $v'\in V'$ as an agent to $A'$ and for each $i\in [\nu]$ a set $D'_i$ of $\nu^3$ dummy agents. 
	We set  $D':=\bigcup_{i\in [\nu]} D'_i$. 
	The preferences of the agents are as follows: 
	\begin{align*}
	v' & : [N_{G'}(v')]\succ [D'] \succ [V'\setminus ( N_{G'}(v')\cup \{v'\})] \qquad & \forall v'\in V' \\
	d' & : [D'] \setminus \{d'\} \succ v'_i \succ v'_{i+1} \succ v'_{i+2} \succ \dots \succ v'_{\nu} \succ v'_1 \succ \dots \succ v'_{i+1} \qquad & \forall d'\in D'_i,  \forall i\in [\nu]
	\end{align*}
    
	We now prove that the given \textsc{Graph Isomorphism} instance is a yes-instance if and only if $d_{\spear}(\mathcal{I},\mathcal{I}')\leq \nu^3\cdot \sum_{j\in [\nu]}\sum_{i\in [\nu]} |i-j|+\nu^3 = \frac{1}{3}\nu^4 (\nu^2 -1) + \nu^3$.
	
	\paragraph{Proof of Correctness.}
	$(\Rightarrow)$ 
	Let $\pi:D\to D'$ be the mapping that maps for $i\in [\nu^4]$ the dummy agent ranked in position $i$ in $[D]$ to the dummy agent ranked in position $i$ in $[D']$.
	Assume that $G$ and~$G'$ are isomorphic witnessed by the bijection $\mu:V\to V'$. 
	Then, we construct a bijection~$\sigma: A\to A'$ by mapping $v$ to $\mu(v)$ for all $v\in V$ and $d$ to $\pi(d)$ for all $d\in D$.
	We start by upper-bounding the distance between $\sigma(\succ_v)$ and $\succ_{\sigma(v)}$ for~$v\in V$. 
	As $\mu$ is an isomorphism between $G$ and $G'$, we have that $\{\sigma(w)\mid w\in N_G(v)\}=N_{G'}(\sigma(v))$. Moreover, we have $\{\sigma(w)\mid w\in V\setminus ( N_G(v)\cup \{v\})\}=V'\setminus ( N_G'(\sigma(v))\cup \{\sigma(v)\})$. 
 	Thus, in $\sigma(\succ_v)$ the same agents from $V'$ appear before the first dummy agent as in~$\succ_{\sigma(v)}$ and the same agents from $V'$ appear after the last dummy agent. 
	Moreover, note that all dummy agents are ranked in the same position in the two preference orders.
	Thus, we can upper bound $d_{\spear}(\sigma(\succ_v),\succ_{\sigma(v)})\leq \nu^2$: 
	For each of the $\nu-1$ agents from $V'\setminus \{\sigma(v)\}$ their position in the two preference orders can differ by at most $\nu$, since in both preference orders the same at most $\nu$ agents appear before the first dummy agent and the same at most $\nu$ agents after the last dummy agent.
	Consequently, we have 
	\begin{equation}\label{eq:1}
	    \sum_{v\in V} d_{\spear}(\sigma(\succ_v),\succ_{\sigma(v)}) \leq \nu^3\,.
	\end{equation}
	
	Turning to the dummy agents, note that for each $d\in D$, $\sigma(\succ_d)$ ranks all dummy agents in the same position as $\succ_{\sigma(d)}$. 
	Thus, only the different ordering of the agents from $V'$ in $\sigma(\succ_d)$ and~$\succ_{\sigma(d)}$ contribute to the Spearman distance between the two. 
	Observe that for each two $b,d\in D$, agents~$\sigma(\succ_b)$ and $\sigma(\succ_d)$ rank each agent from $V'$ in the same position. 
	Moreover, observe that considering the preference orders of agents from $D'$, each agent $v'\in V'$ appears exactly $\nu^3$ times in position~$(\nu^4-1)+i$ for each $i\in [\nu]$.
	Let us now focus on agent  $v' := \sigma (v)\in V'$ where $v$ is ranked in position~$(\nu^4-1)+j$ for $j\in [\nu]$ by $\succ_d$ for each $d\in D$. 
	Then $v'$ contributes $|j-i|$ to $d_{\spear}(\sigma(\succ_d),\succ_{\sigma(d)})$ for each~$d\in D$ where $v'$ is ranked in position $(\nu^4-1)+i$ in~$\succ_{\sigma(d)}$. 
	Together with our previous observation that each vertex agent appears $\nu^3$ times in position $(\nu^4-1)+i$ for each $i\in [\nu]$ in the preferences of agents from $D'$ this implies that agent $v'$ overall contributes $\nu^3\cdot (\sum_{i\in [\nu]} |i-j|)$ to the Spearman distance between the mapped preference orders of dummy agents. 
	Summing up over all~$j\in [\nu]$, we get that the total Spearman distance between the mapped preference orders of dummy agents is $\nu^3\cdot \sum_{j\in [\nu]}\sum_{i\in [\nu]} |i-j|$. 
    Combining this with \Cref{eq:1}, we get that  $d_{\spear}(\mathcal{I},\mathcal{I}')\leq \nu^3\cdot \sum_{j\in [\nu]}\sum_{i\in [\nu]} |i-j|+\nu^3$. 
	
	$(\Leftarrow)$ 
    Let~$\sigma : A \rightarrow A'$ witness $d_{\spear} (\mathcal{I}, \mathcal{I'}) \le \frac{1}{3}\nu^4 (\nu^2 -1) + \nu^3$.
    
    We first show that $\sigma$ does not map any agent from $D$ to an agent from $V'$. 
    To show this, let $X\subseteq D$ be the subset of agents from $D$ which are mapped to agents from $V'$ in $\sigma$ and $Y' := \sigma (X) \subseteq V'$ be the subset of agents from $V'$ to which an agent from $D$ is mapped in $\sigma$.
    Assume for the sake of contradiction that $x:=|X|=|Y'|>0$.
    Now we compute the summed distance between the preferences of the agents from $D\setminus X$ and the preferences of the agent they are mapped to in $\sigma$ and show that this distance already exceeds the given budget. 
    In particular, we give a lower bound on
    \begin{align}
	    \sum_{d\in D\setminus X} d_{\spear}(\sigma(\succ_d),\succ_{\sigma(d)}) \ge &\sum_{d \in D\setminus X} \sum_{v' \in V' \setminus Y'} |\pos_{\succ_d} (\sigma^{-1} (v')) - \pos_{\sigma (\succ_d)} (v')| \notag\\
	    & + \sum_{d \in D\setminus X} \sum_{y' \in Y'} |\pos_{\succ_d} (\sigma^{-1} (y')) - \pos_{\sigma (\succ_d)} (y')|\label{eq:2}\,.
	\end{align}
    We first give a lower bound on the first summand, i.e., $\sum_{d \in D\setminus X} \sum_{v' \in V' \setminus Y'} |\pos_{\succ_d} (\sigma^{-1} (v')) - \pos_{\sigma (\succ_d)} (v')|$.
    Note that for each two agents $b,d\in D\setminus X$ and each $v'\in V'\setminus Y'$ it holds that $\sigma(\succ_b)$ and $\sigma(\succ_d)$ rank $v'$ in the same position.
    Let us now focus on agent~$v'\in V'\setminus Y'$ where $v'$ is ranked in position $(\nu^4-1)+j$ for some $j\in [\nu]$ by $\sigma(\succ_d)$ for each $d\in D\setminus X$. 
    Then $v'$ contributes $|j-i|$ to $d_{\spear}(\sigma(\succ_d),\succ_{\sigma(d)})$ for each $d\in D\setminus X$ where $v'$ is ranked in position~$(\nu^4-1)+i$ in~$\succ_{\sigma(d)}$.
    This together with the facts that there are $\nu^3$ agents from $D'$ ranking $v'$ in position $i$ for each~$i\in [n]$, $|X|=x$, and $|j-i|\leq \nu$, we get that (where $j = \pos_{\sigma(\succ_d)} (v') - (\nu^4-1)$ for some $d\in D$)
    \begin{align}
        \sum_{d \in D\setminus X} |\pos_{\succ_d} (\sigma^{-1} (v')) - \pos_{\sigma (\succ_d)} (v')| \ge \nu^3\cdot \bigl(\sum_{i\in [\nu]} |i-j|\bigr)-x\nu \label{eq:lb1}\,.
    \end{align}
    
    We now turn to the second summand of \Cref{eq:2}, i.e., $\sum_{d \in D\setminus X} \sum_{y' \in  Y'} |\pos_{\succ_d} (\sigma^{-1} (y')) - \pos_{\sigma (\succ_d)} (y')|$. 
    Note that for each agent $y'\in Y'$, we have that it is placed in position $\nu^4-i \le \nu^4 -1$ for some $i\in [\nu^4-1]$ in $\sigma(\succ_d)$ for all $d\in D\setminus X$, as a dummy agent is mapped to $y'$ and dummy agents appear only in the first $\nu^4-1$ positions in $\sigma (\succ_d)$. 
    Thus, as there are $\nu^3$ agents in $D'$ that rank $y'$ in position $(\nu^4-1)+j$ for~$j\in [\nu]$ and as $|X|=x$, we get that 
    \begin{align}
        \sum_{d\in D\setminus X} |\pos_{\succ_d} (\sigma^{-1} (y')) &- \pos_{\sigma (\succ_d)} (y')| \notag \ge - x\nu+\nu^3\sum_{j=1}^{\nu} j= \nu^3\cdot \frac{\nu\cdot (\nu + 1)}{2} - x \nu\\
        & = \nu^4 + \nu^3\frac{\nu^2 - \nu}{2}-x\nu \,.\label{eq:lb2}
    \end{align}
    
    Summing \Cref{eq:lb1,eq:lb2} over all $v' \in V'$ and using (for the second inequality) that $\sum_{i\in [\nu]} |i-j|\leq \frac{\nu^2 - \nu}{2}$ for all $j\in [\nu]$, we get
    \begin{align*}
        \sum_{d\in D\setminus X} d_{\spear}&(\sigma(\succ_d),\succ_{\sigma(d)}) \\
        & \ge
        \sum_{v' \in V' \setminus Y'} \Bigl(\nu^3\cdot \bigl(\sum_{i\in [\nu]} |i-\pos_{\sigma(\succ_b)} (v') + (\nu^4-1)|\bigr)-x\nu \Bigr) + \sum_{y' \in Y'} \Bigl( \nu^4 + \nu^3\frac{\nu^2 - \nu}{2}-x\nu \Bigr)\\
        & \ge |Y'| \cdot \nu^4 + \sum_{v' \in V'} \Bigl(\nu^3 \cdot \bigl(\sum_{i\in [\nu]} |i-\pos_{\sigma(\succ_b)} (v') + (\nu^4-1)|\bigr) - x \nu\Bigr) \\
        & = x \nu^4 + \nu^4 \cdot \bigl(\sum_{j \in [\nu]}\sum_{i\in [\nu]} |i-j|\bigr) - x \nu^2\\
        & > \nu^4 \cdot \bigl(\sum_{j\in [\nu]}\sum_{i\in [\nu]} |i-j|\bigr) + \nu^3
    \end{align*}
    where we used our assumption~$\nu>2$ as well as $x > 0$ for the last inequality. 
    Thus, we have reached a contradiction to $\sigma$ witnessing a solution, implying that $|X| = 0$. 
    Consequently, we may assume in the following without loss of generality that $\sigma$ matches dummy agents from $D$ to dummy agents from $D'$ and vertex agents from $V$ to vertex agents from $V'$ in $\sigma$. 
    
    Observe that the arguments given in the forward direction of the proof imply that independent of how $\sigma$ maps vertex agents to vertex agents and dummy agents to dummy agents we have that $\sum_{d\in D} d_{\spear}(\sigma(\succ_d), \succ_{\sigma(d)})\geq \frac{1}{3}\nu^4\cdot(\nu^2-1)$. 
    Thus, it needs to hold that $\sum_{v\in V} d_{\spear}({\sigma(\succ_v)}, {\succ_{\sigma(v)})} \leq\nu^3<\nu^4$.
    As we have $\nu^4$ dummy agents, this implies that for each~$v\in V$, we need to have that $\sigma(\succ_v)$ and $\succ_{\sigma(v)}$ rank the same agents before the first dummy agent: The position difference of an agent that appears in one preference order before the dummy agents and in the other after the dummy agents would be at least $\nu^4$, which is not possible. 
    Thus, we have that $\{\sigma(w)\mid w\in N_G(v)\}=N_{G'}(\sigma(v))$. 
    Thus, restricting the mapping $\sigma$ to the agents from $V$ leads to a mapping~$\mu:V\to V'$ that induces an isomorphism from $G$ to $G'$.
\end{proof}

\subsection{Mutual Attraction Distance} \label{sub:MAD}

In this section, we introduce and discuss our main distance measure, which we call mutual attraction distance. 

\paragraph{Intuition.} One characteristic of SR instances, which distinguishes them from classical elections studied by Szufa et al.~\cite{DBLP:conf/atal/SzufaFSST20}, is that each agent is associated with a preference order and also appears in the preference order of other agents. 
Thus, when considering, for instance, stable matchings, for an agent $a$ it is not only important which agents $a$ likes, but also whether they like $a$ as well.
Accordingly, our mutual attraction distance focuses on how pairs of agents rank each other. 
In particular, each agent~$a$ is characterized by a mutual attraction vector whose $i$-th entry contains the position in which $a$ appears in the preferences of the agent who $a$ ranks in $i$-th position.
In the mutual attraction distance (see \Cref{def:MAD} for a formal definition), we match the agents from two different instances such that the $\ell_1$ distance between the mutual attraction vectors of matched agents is minimized. 

\paragraph{Notation.}
For a matrix $M\in \mathbb{R}^{p\times q}$ and some $i\in [p]$, let $M_{i}$ denote the $i$-th row of $M$.
For an SR instance $\mathcal{I}=(A=\{a_1,\dots a_{2n}\}, (\succ_a)_{a\in A})$, an agent $a\in A$, and some $i\in [2n-1]$, let $\MA_{\mathcal{I}}(a,i)$ be the position of $a$ in the preference order of the agent~$a'$ which is ranked in position $i$ by $a$, i.e., $\MA_{\mathcal{I}}(a,i) :=\pos_{\succ_{a'}}(a)$ where~$a' := \ag_{\succ_a}(i)$.
Then, the mutual attraction vector of agent $a$ is $\MA_{\mathcal{I}}(a)=\bigl(\MA_{\mathcal{I}}(a,1),\dots, \MA_{\mathcal{I}}(a,2n-1)\bigr)$. 
Lastly, the mutual attraction matrix~$\MA_{\mathcal{I}}$ of $\mathcal{I}$ is the matrix whose $i$-th row is the vector~$\MA (a_i)$.

 \begin{definition}\label{def:MAD}
  The mutual attraction distance between two SR instances $\mathcal{I}$ with agents $A$ and $\mathcal{I}'$ with agents $A'$ with $|A|=|A'|$ is defined by their mutual attraction matrices as 
$$\retrodist(\mathcal{I},\mathcal{I}'):=\retrodist(\MA_\mathcal{I},\MA_{\mathcal{I}'}):=\min_{\sigma\in \Pi([|A|],[|A'|])} \sum_{i\in [|A|]} \ell_1\big((\MA_{\mathcal{I}})_i, (\MA_{\mathcal{I}'})_{\sigma(i)}\big).$$ 
 \end{definition}
As the mutual attraction distance is defined over mutual attraction matrices, we sometimes speak about mutual attraction matrices without specifying the underlying SR instance. 

\begin{example}

Consider the two SR instances $\mathcal{I}$ and $\mathcal{I}'$ defined in \Cref{ex:2}. 
Their mutual attraction matrices are: 
 $$
	   \MA_\mathcal{I} = 
	\kbordermatrix{ & 1 & 2 & 3  \\
		a &                1 & 1 & 1  \\
		b &                1 & 2 & 2  \\
		c &                2 & 2 & 3   \\
		d &                3 & 3 & 3   
	}
	\text{, }
	                              \MA_\mathcal{I'} = 
	\kbordermatrix{ & 1 & 2 & 3  \\
		x &                1  & 3  & 3 \\
		y &                1  & 2  & 2 \\
		z &              1  & 2 & 2 \\
		w &              1  & 3  & 3  
	}
	$$
Their mutual attraction distance is $2+0+2+2=6$ as witnessed by the mapping $\sigma(a)=z$, $\sigma(b)=y$, $\sigma(c)=x$, and $\sigma(d)=w$.  

\end{example}

\paragraph{Computation.}
Given two SR instances $\mathcal{I}$ over agents $A$ and $\mathcal{I}'$ over agents $A'$, computing their mutual attraction distance reduces to finding an minimum-weight perfect matching in a complete bipartite graph~$G = (A\cupdot A', E)$ with the following weights:
Edge $\{a,a'\}\in E$ has weight~$\ell_1(\MA_{\mathcal{I}}(a), \MA_{\mathcal{I}'}(a'))$.

\begin{observation}
 Given two SR instances $\mathcal{I}$ and $\mathcal{I}'$ with $2n$ agents each,  $\retrodist(\mathcal{I}, \mathcal{I}')$ can be computed in $\mathcal{O}(n^3)$ time.
\end{observation}

\paragraph{Realizable Mutual Attraction Matrices.}
Not every $(2n)\times (2n-1)$ matrix is the mutual attraction matrix of some SR instance. 
Accordingly, we call a matrix $M$ \emph{realizable} if there is an SR instance $\mathcal{I}$ with $\MA_{\mathcal{I}}=M$. 
Realizable matrices exhibit certain characteristics.
For example, as each agent ranks exactly one agent at position~$j$ for every~$j\in [2n-1]$, every realizable matrix~$M\in \mathbb{N}^{(2n)\times (2n-1)}$ contains each number from~$[2n-1]$ exactly $2n$ times. 
Unfortunately, we prove that unless P=NP there cannot be a list of sufficient polynomial-time checkable conditions for when a matrix is realizable:
\begin{restatable}{theorem}{realiz} \label{pr:real}
Given a $(2n)\times (2n-1)$ matrix $M$, deciding whether there is an SR instance $\mathcal{I}$ with $\MA_{\mathcal{I}}=M$ is NP-complete.
\end{restatable}
\begin{proof}
We reduce from the NP-complete problem of deciding whether the edge set of a $3$-regular graph can be partitioned into three disjoint perfect matchings \citep{DBLP:journals/siamcomp/Holyer81a}. 

\paragraph{Construction.}
Given a $3$-regular graph $G=(V=\{v_1,\dots, v_{\nu}\},E)$ , we construct matrix $M$ as follows. 
First, we need to introduce some notation. 
Let $\{v_{p_1},v_{q_1}\},\{v_{p_2},v_{q_2}\},\dots, \{v_{p_z},v_{q_z}\}$ be a list of all vertex pairs that are not adjacent in $G$ (this means that $z={{\nu}\choose{2}}-\frac{3\nu}{2}$).

To construct matrix $M$, we first construct a dummy SR instance $\mathcal{J}$ consisting of dummy and vertex agents:
We introduce one \emph{vertex agent} $a_v$ for each vertex $v\in V$. 
Moreover, we introduce one \emph{dummy agent} $d_{i,j}$ for~$i\in [z]$ and $j\in [\nu]\setminus \{p_i,q_i\}$.
Concerning the agent's preferences, we start by constructing the preferences of some vertex agent $a_{v_{\ell}}$ for some $\ell\in [\nu]$. 
Vertex agent $a_{v_{\ell}}$ ranks the three agents corresponding to the three vertices adjacent to it in $G$ in the first three position in arbitrary order. 
For the subsequent positions for~$i\in [z]$, if $v_{\ell}=p_{i}$ or $v_{\ell}=q_{i}$, then $a_{v_{\ell}}$ ranks $a_{q_{i}}$, respectively, $a_{p_{i}}$ in position $i+3$; otherwise $a_{v_{\ell}}$ ranks $d_{i,\ell}$ in position $i+3$.
All remaining agents are appended to the preferences in some arbitrary order. 
Concerning the dummy agents, agent $d_{i,j}$ for $i\in [z]$ and $j\in [\nu]\setminus \{p_i,q_i\}$ ranks agent $a_{v_j}$ in first position.
Moreover, the dummy agents rank all other dummy agents in an arbitrary order in the subsequent positions such that no two dummy agents rank each other on the same position (this can be achieved by performing cyclic shifts). 
Subsequently, they rank all vertex agents in some arbitrary ordering. 

Let $M':=\Ret_{\mathcal{I}}$. 
To obtain $M$ we modify $M'$:
For each vertex agent, we set the first entry of its vector to one, the second entry to two, and the third entry to three. 

\paragraph{Proof of Correctness.}
$(\Rightarrow)$ Given a partitioning of $E$ into three perfect matching $M_1$, $M_2$, and~$M_3$, for each $\ell\in [\nu]$ and $i\in [3]$, let $v_{\ell,i}$ be the vertex adjacent to $v_{\ell}$ in $M_i$.
To construct an SR instance $\mathcal{I}$ realizing $M$, we start with the above constructed instance $\mathcal{J}$ and modify the first three positions of each vertex agent as follows: 
For $\ell\in [\nu]$ and $i\in [3]$, agent~$a_{v_{\ell}}$ ranks $a_{v_{\ell,i}}$ in position $i$. 
Note that as $M_1$, $M_2$, and $M_3$ are perfect disjoint matchings, in the resulting instance each agent ranks all other agents in its preferences.
We now claim that $\Ret_{\mathcal{I}}=M$. 
Note that $\Ret_{\mathcal{I}}$ and $\Ret_{\mathcal{J}}$ are identical up to the first three columns in rows corresponding to vertex agents. 
As for each edge~$\{v,w\}\in M_i$ for $i\in [3]$ agent~$v$ ranks $w$ in position $i$ and $w$ ranks $v$ in position $i$, in $\Ret_{\mathcal{I}}$ the vector of each vertex agent starts with $1,2,3$. 
Thus, $\Ret_{\mathcal{I}}=M$. 

$(\Leftarrow)$
Assume that there is an SR instance $\mathcal{I}$ with $\Ret_{\mathcal{I}}=M$.
First observe that in $M$ in each column~$i\in [4,z+3]$ there are exactly two rows which contain an $i$ at position $i$, that are, the two rows corresponding to vertex agents $a_{v_{p_{i-3}}}$ and $a_{v_{q_{i-3}}}$: All other vertex agents rank a dummy agent in this position, which in turn ranks the vertex agent in first position. Moreover, by construction we have that dummy agents rank only other dummy agents in position $4$ to $z+3$ and that no two dummy agents rank each other in the same position. 
Thus, it follows that $a_{v_{p_{i-3}}}$ and $a_{v_{q_{i-3}}}$ rank each other in position $i$. 
This implies that a vertex agent~$a_v$ ranks all agents corresponding to vertices that are not adjacent to $v$
in $G$ between positions~$4$ to~$z+3$. 
We claim that this further implies that $a_v$ ranks the agents corresponding to adjacent vertices at $G$ in the first three positions in $\mathcal{I}$. 
By construction, for no dummy agent does its mutual attraction vector contain an $i$ at position~$i$ for $i\in [3]$.
Thus, $a_v$ needs to rank vertex agents in the first three positions, and the vertex agents for adjacent vertices are the only remaining ones. 
Furthermore, observe that if $a_v$ ranks~$a_w$ in position~$i$ for $i\in [3]$, then by the construction of $M$, agent~$a_w$ ranks $a_v$ in position~$i$. 
For~$i\in [3]$, let $M_i:=\{\{a_v,a_w\}\mid a_v \text{ and } a_w \text{ rank each other in position } i \text{ in } \mathcal{I}\}$. 
Note that $M_i$ is clearly a matching as each agent can only rank one other agent in each position. 
Moreover, by our above observations, $M_i$ is perfect. 
Furthermore, $M_1$, $M_2$, and $M_3$ need to be disjoint again because each agent can rank only one agent on each position. 
Thus, we found a solution to the given instance. 
\end{proof}

\paragraph{Properties of Mutual Attraction Distance.}
As the $\ell_1$-distance between vectors satisfies the triangle inequality, the mutual attraction distance also satisfies the triangle inequality (as mappings between agent sets can be simply applied on top of each other). 
From this, we easily get the following: 
\begin{observation}
    The mutual attraction distance is a pseudometric. 
\end{observation}

However, note that the mutual attraction distance is not isomorphic, i.e., there exist multiple non-isomorphic SR instances having the same mutual attraction matrix:

\begin{observation}\label{lem:non-iso}
 The mutual attraction distance is not an isomorphic distance.
\end{observation}
\begin{proof}
    Let $\mathcal{I}$ with agents $a, b, c$, and $d$ and $\mathcal{I}'$ with agents $a',b',c'$, and $d'$ be two SR instances with the following preferences:
	\begin{align*}
		a&: b\succ c \succ d,  &b&: a\succ c \succ d, &c&: a\succ d \succ b, &d&: a\succ b \succ c,\\
	a' &:  b'\succ c' \succ d', &b'&:  a'\succ d' \succ c', &c'&:  a'\succ b' \succ d', &d'&: a'\succ c' \succ b'.
	\end{align*}
	The mutual attraction matrices of the two instances are: 
	$$
	   \MA_\mathcal{I} = 
	\kbordermatrix{ & 1 & 2 & 3  \\
		a &                1 & 1 & 1  \\
		b &                1 & 3 & 2  \\
		c &                2 & 3 & 2   \\
		d &                3 & 3 & 2  
	}
	\text{, }
	                              \MA_\mathcal{I'} = 
	\kbordermatrix{ & 1 & 2 & 3  \\
		a' &                1 & 1 & 1 \\
		b' &                1 & 3 & 2 \\
		c' &              2 & 3 & 2 \\
		d' &              3 & 3 & 2  
	}.
	$$
    So we have $\retrodist(\mathcal{I},\mathcal{I}')=0$, yet $\mathcal{I}$ and $\mathcal{I}'$ are not isomorphic. 
\end{proof}

Thus, we say that a matrix has a unique realization if any two SR instances realizing the matrix are~isomorphic.
Unfortunately, there even exist mutual attraction matrices realized by two non-isomorphic SR instances $\mathcal{I}_1$ and $\mathcal{I}_2$ where $\mathcal{I}_1$ admits a stable matching but $\mathcal{I}_2$ does not.
This indicates that the mutual attraction distance between two instances has only a limited predictive value for their relationship in terms of their (distance to) stability, which in turn is not too surprising given that stability is dependent on local configurations. 
\begin{observation}
    There are two non-isomorphic instances~$\mathcal{I}_1$ and $\mathcal{I}_2$ such that $\retrodist (\mathcal{I}_1, \mathcal{I}_2) = 0$, $\mathcal{I}_1$ admits a stable matching, and $\mathcal{I}_2$ does not admit a stable matching.
\end{observation}

\begin{proof}
    Consider the following two instances:
    \begin{align*}
        a_1 &: a_2 \succ a_3 \succ a_4 \succ a_5 \succ a_6 & b_1 & :b_2 \succ b_6 \succ b_4 \succ b_5 \succ b_3\\
        a_2 &: a_3 \succ a_1 \succ a_6 \succ a_4 \succ a_5 & b_2 & : b_3 \succ b_1 \succ b_6 \succ b_4 \succ b_5\\
        a_3 &: a_1 \succ a_2 \succ a_5 \succ a_6 \succ a_4 & b_3 & :b_4 \succ b_2 \succ b_5 \succ b_6 \succ b_1\\
        a_4 &: a_5 \succ a_6 \succ a_1 \succ a_1 \succ a_3 & b_4 & :b_5 \succ b_3 \succ b_2 \succ b_1 \succ b_6\\
        a_5 &: a_6 \succ a_4 \succ a_3 \succ a_3 \succ a_2 & b_5 & :b_6 \succ b_4 \succ b_1 \succ b_3 \succ b_2\\
        a_6 &: a_4 \succ a_5 \succ a_2 \succ a_2 \succ a_1 & b_6 & :b_1 \succ b_5 \succ b_3 \succ b_2 \succ b_4
    \end{align*}
    For both instances, the mutual attraction matrix is the following:
    \[
    \begin{bmatrix}
     2 & 1 & 4 & 3 & 5\\
     2 & 1 & 4 & 3 & 5\\
     2 & 1 & 4 & 3 & 5\\
     2 & 1 & 4 & 3 & 5\\
     2 & 1 & 4 & 3 & 5\\
     2 & 1 & 4 & 3 & 5\\
    \end{bmatrix}
    \]
    The left instance does not admit a stable matching, while the right instance admits the two stable matchings~$M_1 = \{\{b_1, b_2\}, \{b_3, b_4\}, \{b_5, b_6\}\}$ and~$M_2 = \{\{b_2, b_3\}, \{b_4, b_5\}, \{b_6, b_1\}\}$.
\end{proof}
This is in partial contrast to the Spearman distance, where instances at distance zero are isomorphic and thus either both or neither of them admits a stable matching. 
Concerning instances which are at a non-zero distance, for Spearman there also exist SR instances at the minimum distance of $2$ where one admits a stable matching and the other does not. 
However, for the Spearman distance it holds that if a matching $M$ is stable in instance $\mathcal{I}$, then $M$ admits at most $\speardist (\mathcal{I},\mathcal{I}')$ blocking pairs in instance $\mathcal{I}$. 
Thus, under Spearman, if two instances are close to each other and one of them admits a stable matching, then the other instance is also guaranteed to contain a matching which is almost stable. 
While this is not the case for the mutual attraction distance, in \Cref{sec:min-BPS}, we will demonstrate that all our synthetically generated instances are anyway close to admitting a stable matching in the sense that in all instances there is a matching blocked by only few pairs. 

To better understand the general properties of the mutual attraction distance, we continue by proving upper and lower bounds on the distance of two SR instances.  
\begin{restatable}{proposition}{maxdist} \label{lem:max-dist}
	For any two SR instances $\mathcal{I}_1$ and $\mathcal{I}_2$ with $2n$ agents each and ${\retrodist(\mathcal{I}_1, \mathcal{I}_2) > 0}$, we have $2\leq \retrodist (\mathcal{I}_1, \mathcal{I}_2) \le  4\cdot (n-1)\cdot n^2$.
\end{restatable}
\begin{proof}
	Let $M_1:=\MA_{\mathcal{I}_1}$ and $M_2:=\MA_{\mathcal{I}_2}$.
	As every realizable matrix~$M\in \mathbb{N}^{(2n)\times (2n-1)}$ contains each number from~$[2n-1]$ exactly $2n$ times, both $M_1$ and $M_2$ contain each number from~$[2n-1]$ exactly $2n$ time.
	
	For the upper bound, note that from this it follows that each number from~$[2n-1]$ appears exactly $4n$ times in $M_1$ and $M_2$ together.
	Since $|x-y|= \max\{x, y\} - \min \{x, y\}$ holds for all~$x, y\in \mathbb{R}$, we can upper bound $\retrodist (M_1, M_2)$ by summing up the $2n \cdot (2n-1)$ largest numbers appearing in~$M_1$ and~$M_2$ and subtracting the $2n \cdot (2n-1)$ smallest numbers appearing in~$M_1$ and~$M_2$.
	Consequently, we have
	\begin{align*}
	\retrodist (M_1, M_2) &\le 4n \cdot (\sum_{j= n + 1}^{2n-1} j - \sum_{j= 1}^{n-1}j) = 4n \cdot (\sum_{j= 1}^{n-1} (n+j) - \sum_{j= 1}^{n-1}j)\\
	& = 4n \cdot (n \cdot (n-1) + \sum_{j = 1}^{n-1} j - \sum_{j=1}^{n-1} j) = 4n^2 \cdot (n-1) 
	\end{align*}
	
	For the lower bound, note that if $\retrodist (M_1, M_2)\neq 0$, then for each $\sigma\in \Pi([2n],[2n])$ there is some $i\in [2n]$ and $j\in [2n-1]$ with ${M_1}_{i,j}\neq {M_2}_{\sigma(i),j}$. 
	As each number appears in $M_1$ and $M_2$ the same number of times from this it follows that there also needs to be at least one other pair $i'\in [2n]$ and $j'\in [2n-1]$ with ${M_1}_{i',j'}\neq {M_2}_{\sigma(i'),j'}$ and $(i,j)\neq (i',j')$.
	From this it follows that $\retrodist (M_1, M_2)\geq 2$.
\end{proof}

In fact, it is easy to see that the lower bound is tight. 
Later in \Cref{pr:calc}, we will also establish the tightness of the upper bound. 
\begin{observation} \label{ob:MAD2}
	There are two SR instances $\mathcal{I}$ and $\mathcal{I}'$ with $\retrodist(\mathcal{I},\mathcal{I}')=2$. 
\end{observation}
\begin{proof}
	Let $\mathcal{I}$ with agents $a, b, c$, and $d$ and $\mathcal{I}'$ with agents $a',b',c'$, and $d'$ be two SR instances with the following preferences:
	\begin{align*}
	a&: b\succ c \succ d,  &b&: a\succ c \succ d, &c&: a\succ b \succ d, &d&: a\succ b \succ c,\\
	a' &:  c'\succ b' \succ d', &b'&:  a'\succ c' \succ d', &c'&:  a'\succ b' \succ d', &d'&: a'\succ b' \succ c'.
	\end{align*}
	The mutual attraction matrices of the two instances are: 
	$$
	\MA_\mathcal{I} = 
	\kbordermatrix{ & 1 & 2 & 3 \\
		a &                1 & 1 & 1  \\
		b &                1 & 2 & 2  \\
		c &                2 & 2 & 3   \\
		d &                3 & 3 & 3  
	}
	\text{, }
	\MA_\mathcal{I'} = 
	\kbordermatrix{ & 1 & 2 & 3  \\
		a' &                1 & 1 & 1 \\
		b' &                2 & 2 & 2 \\
		c' &              1 & 2 & 3 \\
		d' &              3 & 3 & 3
	}.
	$$
	It clearly holds that $\retrodist(\MA_\mathcal{I}, \MA_\mathcal{I'})=2$.
\end{proof}

\paragraph{Correlation of Mutual Attraction and Spearman Distance.}
As the Spearman distance~$d_{\spear}$ is a very natural and intuitively appealing distance measure, we checked the correlation between the mutual attraction and Spearman distance.
For this we used the test dataset of $460$ instances described in \Cref{sub:creat} for twelve agents.\footnote{We computed the Spearman distance by iterating over all possible agent mappings $\sigma$ and twelve was the largest number of agents we could handle within weeks.} 
To evaluate the correlation, we use the Pearson Correlation Coefficient (PCC).\footnote{PCC is a measure for the linear correlation between two quantities and $0$ if there is no correlation and $1$ if there is a perfect positive linear correlation.} The PCC between mutual attraction and Spearman distances on our test dataset is $0.801$, which is typically regarded as a strong correlation~\cite{schober2018correlation}. 
In particular, for~$95\%$ of instance pairs $(\mathcal{I}, \mathcal{I}')$ we have that $0.82\cdot\retrodist(\mathcal{I},\mathcal{I}')\leq d_{\spear}(\mathcal{I},\mathcal{I}')\leq 1.48\cdot \retrodist(\mathcal{I},\mathcal{I}')$. 
\Cref{fig:MAD} depicts this correlation on the instance level.  
Unfortunately, the ratio between the mutual attraction and Spearman distance is unbounded.
First, as the Spearman distance is isomorphic but the mutual attraction distance is not, there are instances with mutual attraction distance zero but positive Spearman distance.
Second, we show that there are instances with mutual attraction distance zero but unbounded Spearman distance:
\begin{observation}\label{obs:mad-spear}
    For any~$n \ge 2$, there are SR instances~$\mathcal{I}_1$, $\mathcal{I}_2$ on $n$ agents with $\speardist (\mathcal{I}_1, \mathcal{I}_2) = 2$ but $\retrodist (\mathcal{I}_1, \mathcal{I}_2) \ge n-2$.
\end{observation}

\begin{proof}
    Consider the following SR instance~$\mathcal{I}_1$ with~$n$ agents.
    \begin{align*}
        a_1 & : a_2 \succ a_n \succ a_{n-1} \succ \dots \succ a_3\\
        a_i & : a_2 \succ a_3 \succ \dots a_{i-1} \succ a_1 \succ a_{i+1} \succ a_{i+2} \succ \dots \succ a_n \succ a_1
    \end{align*}
    Let $\mathcal{I}_2$ be the instance arising from this SR instance by swapping~$a_2 $ and $a_n$ in the preferences of~$a_1$.
    We will denote the agents of the instance belonging to~$\mathcal{I}_2$ by $a_1'$, $a_2'$, $a_3'$, \dots, $a_n'$.
    Then $\speardist (\mathcal{I}_1, \mathcal{I}_2) = 2$.
    
    It remains to show that $\retrodist (\mathcal{I}_1, \mathcal{I}_2) \ge n-2$.
    Let $\sigma$ be a bijection from~$\{a_1, a_2, \dots, a_n\}$ to~$\{a_1', a_2', \dots, a_n'\}$ corresponding to the minimum distance~$\retrodist (\mathcal{I}_1, \mathcal{I}_2)$.
    Note that for each~$i \in[2, n]$, we have $\MA_{\mathcal{I}_1}(a_i, j) = i-1 = \MA_{\mathcal{I}_2} (a_i, j)$ for all $j \neq i$.
    Thus, if $\sigma (a_i) \neq a_i'$ for some~$i \in [2, n]$, then $\retrodist (\mathcal{I}_1, \mathcal{I}_2) \ge n-2$.
    Otherwise, we have $\sigma (a_i) = a_i'$ for all $i \in [n]$.
    However, we have $\MA_{\mathcal{I}_1} (a_1) = (1, n-1, n-2, n-3, \dots, 2)$ and $\MA_{\mathcal{I}_2} (a_1) = (n-1, 1, n-2, n-3, \dots, 2)$.
    Thus, in this case we would have $\retrodist (\mathcal{I}_1, \mathcal{I}_2) \ge 2 \cdot (n-2)$.
\end{proof}
Note that we do not believe that the bound from \Cref{obs:mad-spear} is tight.

\begin{figure*}
	\centering
	\begin{subfigure}[b]{0.49\textwidth}
		\centering
		\includegraphics[width=7cm]{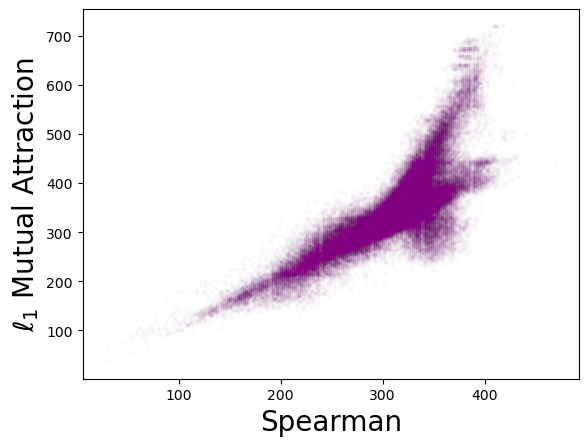}
		\caption{mutual attraction distance (PCC = 0.801)\label{fig:MAD}}
	\end{subfigure}
	\begin{subfigure}[b]{0.49\textwidth}
		\centering
		\includegraphics[width=7cm]{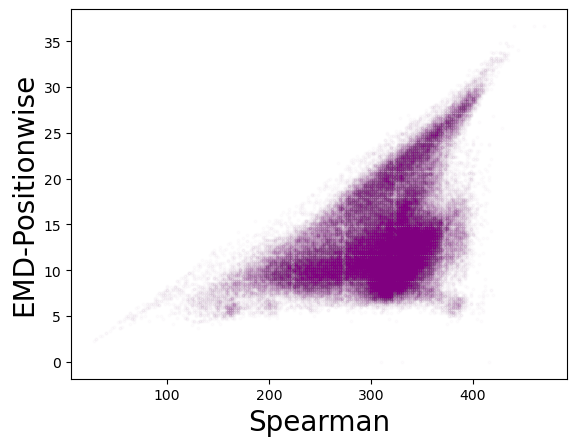}
		\caption{positionwise distance (PCC = 0.457)\label{fig:POS}}
	\end{subfigure}
	\caption{Correlation between the Spearman distance and our mutual attraction distance or the positionwise distance of \citet{DBLP:conf/atal/SzufaFSST20} on the dataset described in \Cref{sub:creat} for twelve agents. Each pair of instances is represented by a point with its $x$-axis representing their distance according to one of the measures and its $y$-axis representing their distances according to the other.}
	\label{fig:correlation}
\end{figure*}

\paragraph{Positionwise Distance.}
The papers of \citet{DBLP:conf/atal/SzufaFSST20} and \citet{DBLP:conf/ijcai/BoehmerBFNS21} on the map of elections used a different distance measure defined over the so-called position matrices. 
In a position matrix of an election, we have one row for each candidate and one column for each position, and an entry contains the fraction of voters that rank the respective candidate in the respective position. 
This distance naturally extends to SR instances by introducing a row for each agent capturing in which positions the agent is ranked by the other agents. 
Intuitively, this representation might appear appealing, as it captures the general popularity/quality of agents in the instance. 
However, the position matrix completely ignores that each agent is not only ranked by other agents, but also associated with a preference order itself.  
Consequently, the positionwise distance completely disregards mutual opinions, i.e., what agents think of each other, which are essential for stability related considerations. 
The unsuitably of the positionwise distance for SR instances is also illustrated in a Pearson correlation coefficient of only $0.457$ with the Spearman distance (see \Cref{fig:POS} for a visualization of the correlation). 
Lastly, note that the mutual attraction matrix of an SR instance captures all information contained in the position matrix, as $\MA_{\mathcal{I}}(a)$ for some agent $a$ contains the positions in which $a$ is ranked by the other agents in $\mathcal{I}$.

\section{Navigating the Space of SR Instances} \label{sec:understanding}
Interpreting and using a map of SR instances, it will be useful to give different regions on the map an intuitive meaning.  
This is why we now identify four somewhat ``canonical'' extreme mutual attraction matrices, which are far away from each other and thus fall into four very different parts of the map. 

\paragraph{Identity.} Our first extreme case is that all agents have the same preferences, i.e., there exists a central order called master list of the agents $A$ and the preferences of an agent~$a\in A$ are derived from the master list by deleting $a$.
Stable matching instances with master lists have already attracted significant attention in the past \cite{DBLP:conf/wine/BredereckHKN20,DBLP:journals/cn/CuiJ13,DBLP:journals/dam/IrvingMS08,DBLP:conf/atal/Kamiyama19}. 
For $n\in \mathbb{N}$, the identity matrix is defined by $\IDn[i,j] := \begin{cases} i &  j \ge i\\
i - 1 & j < i
\end{cases}$ for each $i\in [2n]$ and $j\in [2n-1]$.
We prove that, in fact, only SR instances where all preferences are derived from a master list realize the identity matrix:
\begin{restatable}{proposition}{IDreal} \label{lem:IDreal}
	For every $n \in \mathbb{N}$, an SR instance $\mathcal{I}$ is a realization of \IDn~if and only if the preferences in $\mathcal{I}$ are derived from a master list.
 In particular, the realization of \IDn~is unique.
\end{restatable}

\paragraph{Mutual Agreement.} Our second extreme case is mutual agreement: For each pair $a$ and $a'$ of agents, $a$ and $a'$ evaluate each other identically, i.e., $a$ ranks $a'$ on the $i$-th position if and only if $a'$ ranks $a$ on the $i$-th position.\footnote{Notably, in the \textsc{Stable Marriage with Symmetric Preferences} \cite{malleythesis,DBLP:journals/im/AbrahamLMO08} problem we are given a set of men and women with preferences over each other, where a woman $w$ ranks a man $m$ in position $i$ if and only if $m$ ranks $w$ in position $i$, an idea very similar to mutual agreement (see \citet[Chapter 6]{malleythesis} for additional motivation). \citet{malleythesis,DBLP:journals/im/AbrahamLMO08} study the computational complexity of various traditional questions on such instances (in the presence of ties).}
For $n\in \mathbb{N}$, this is captured in the mutual agreement matrix $\symn$ where we have $\symn[i,j] = j$ for each $i\in [2n]$ and $j\in [2n-1]$.
At first glance, it is unclear whether the mutual agreement matrix is realizable. 
It turns out that the realizations of $\symn$ correspond to Round-Robin tournaments: In a Round-Robin tournament of $2n$ agents there are $2n-1$ days with each agent competing exactly once each day and exactly once against each other agent \cite{harary1966theory}.
The intuition here is that an agent in the SR instance corresponding to a Round-Robin tournament ranks in the $i$-th position the agent against whom it competes on the $i$-th day. 
Formally, we have: 
\begin{restatable}{proposition}{MAreal} \label{lem:MAreal}
	 For every~$n \in \mathbb{N}$, there is a bijection between realizations of \symn and the set of Round-Robin tournaments.
 In particular, there are several non-isomorphic realization of \symn\ for~$n = 4 $.
\end{restatable}

\paragraph{Mutual Disagreement.} Our third extreme case is mutual disagreement. For each pair $a$ and $a'$ of agents, their evaluations for each other are diametrical, i.e., $a$ ranks $a'$ in the $i$-th position if and only if $a'$ ranks $a$ in the $(2n-i)$-th position. 
For $n\in \mathbb{N}$, this is captured in the mutual disagreement matrix $\asymn$ where we have $\asymn[i,j] = 2n - j$ for each $i\in [2n]$ and $j\in [2n-1]$.
There exists a straightforward realization of $\asymn$ with $2n$ agents $a_1,\dots, a_{2n}$ where the preferences of agent $a_i$ are derived from the preferences of agent $a_{i-1}$ by performing a cyclic shift, i.e.,  $a_i : a_{i+1} \succ a_{i+2} \succ \dots \succ a_n \succ a_1 \succ a_2 \succ \dots \succ a_{i - 1}$. 
However, this realization is not unique:
\begin{restatable}{proposition}{MDreal} \label{lem:MDreal}
	 For every $n \in \mathbb{N}$, matrix~$\asymn$ is realizable. For~$n = 3$, matrix~$\asymn $ has multiple non-isomorphic realizations.
\end{restatable}

\paragraph{Chaos.} Our fourth extreme mutual attraction matrix is the chaos matrix $\unn$, which is defined for each $i\in [2n]$ and $j\in [2n-1]$ as  $
    \unn[i,j] = \begin{cases}
        j, & \text{for } i=1\\
       i + nj - n - 1 \mod 2n - 1, & \text{otherwise}
        \end{cases}$.
Unlike the other three matrices, we have no natural interpretation of the chaos matrix. 
We added this matrix to the other three because it is far away from each of them and thus falls into an otherwise vacant part of the map. 
Its name ``chaos'' stems from the fact that this matrix is close on the map to instances with uniformly at random sampled preferences.
We prove that for infinitely many~$n\in \mathbb{N}$, $\unn$ is realizable:
\begin{restatable}{proposition}{CHreal} \label{lem:CHreal}
	 For every~$n \in \mathbb{N}$ such that $2n-1 $ is not divisible by~3, matrix~$\unn$ is realizable, and the realization is unique.
\end{restatable}

\paragraph{Distances Between  Matrices.} The distances between our extreme matrices are as~follows: 
\begin{restatable}{proposition}{disss}
	\label{pr:calc}
	For each $n\in \mathbb{N}$, we have
	
	{\footnotesize \begin{align*}
	    & \retrodist (\IDn , \symn)= \retrodist (\symn, \unn) = \frac{8}{3} n^3 - 4n^2 + \frac{4}{3} n, \quad   \retrodist (\symn, \asymn ) = 4\cdot (n-1)\cdot n^2,\\ 
	    & \retrodist (\IDn , \asymn) = \retrodist (\asymn, \unn) =\frac{8}{3} n^3 - 2n^2 - \frac{2}{3} n, \quad \retrodist (\IDn, \unn) =\frac{8}{3} n^3 \pm O(n^2).
	\end{align*} }
\end{restatable}

\noindent
As proven in \Cref{lem:max-dist},  $D(2n):=4\cdot (n-1)\cdot n^2$ is the maximum possible distance between two mutual attraction matrices of SR instances with $2n$ agents. 
Thus, the mutual agreement matrix and the mutual disagreement matrix are at the maximum possible distance and therefore form a diameter of our space. 
For each two matrices $X$ and $Y$ among $\ID$, $\sym$, $\chaos$, and $\asym$, we define their asymptotic normalized distance as:
$\norretrodist(X,Y) := \lim_{n \rightarrow
	\infty}\nicefrac{\retrodist(X^{2n},Y^{2n})}{D(2n)}$. 
It turns out that for all pairs of matrices $X,Y\in \{\ID,\sym,\asym,\chaos\}$ with $\{X,Y\}\neq \{\sym,\asym\}$ we have $\norretrodist(X,Y)=\frac{2}{3}$ , while $\norretrodist(\sym,\asym)=1$. 
This implies that our extreme matrices are indeed far from each other.

\section{A Map of Synthetic SR Instances} \label{sec:map}
In this section, we present a map of synthetic SR instances. 
In \Cref{sub:creat}, we describe how we create the map and how we generate the instances. 
To do this, we recall several statistical cultures from the literature but also introduce several new ones.
In \Cref{sub:eval}, we explain the map by giving the horizontal and vertical axis a natural interpretation and by analyzing where different statistical cultures land. 

Note that in the following, all discussed values of the mutual attraction distance are normalized values, i.e., they are divided by $D(2n)=4\cdot (n-1)+n^2$. 

\subsection{Creating the Map} \label{sub:creat}
We first describe our dataset of $460$ SR instances and then explain how we visualize it as a map.

\paragraph{Points on the Map -- Statistical Cultures.}
We use the following statistical cultures to generate SR instances. To the best of our knowledge, only the Impartial Culture, Attributes, Mallows, and Euclidean models have been previously considered. For all cultures we start by initializing a set $A$ of $n$ agents.
\begin{description}
\itemsep0em 
    \item[Impartial Culture (IC)] Agent $a\in A$ draws its preferences uniformly at random from $\mathcal{L}(A\setminus \{a\})$.
    \item[2-IC] Given parameter $p\in [0,0.5]$, we partition $A$ into two sets $A_1\cupdot A_2$ with $|A_1|=\floor{p\cdot |A|}$.  
    Each agent~$a\in A$ samples a preference order $\succ$ from $\mathcal{L}(A_1\setminus \{a\})$ and one order $\succ'$ from $\mathcal{L}(A_2\setminus \{a\})$. 
    If $a\in A_1$, then $a$'s preferences start with all agents from $A_1$ ordered according to $\succ$ and then all agents from $A_2$ ordered according to $\succ'$. 
    If $a\in A_2$, then it is the other way around, i.e., the preferences start with~$\succ'$ and end with~$\succ$.
    The intuition here is that there are two groups of different sizes (e.g., representing demographic groups), and each agent prefers all agents from its group to all agents from the other group, but preferences within the group are random. 
    \item[Mallows \cite{mal:j:mallows,DBLP:conf/ijcai/BoehmerBFNS21}]  In the original Mallows model, for a parameter $\phi\in [0,1]$ and a preference order $\succ^* \in \mathcal{L}(A)$, the Mallows distribution $\mathcal{D}_{\text{Mallows}}^{\succ^*, \phi}$ assigns preference order ${\succ}\in \mathcal{L}(A)$ a probability proportional to~$\phi^{\swap(\succ^*,\succ)}$.
    We use a normalized variant of Mallows model $\mathcal{D}_{\text{Mallows}}^{\succ^*, \normphi}$ proposed by \citet{DBLP:conf/ijcai/BoehmerBFNS21}  parameterized by a normalized dispersion parameter $\normphi$. 
    Sampling from $\mathcal{D}_{\text{Mallows}}^{\succ^*, \normphi}$, $\normphi$ is internally converted to a dispersion parameter $\phi$ such that the expected swap distance between $\succ^*$ and a sampled preference order from $\mathcal{D}_{\text{Mallows}}^{\succ^*, \phi}$ is $\normphi$ times $\frac{n(n-1)}{4}$. 
    Subsequently a preference order from $\mathcal{D}_{\text{Mallows}}^{\succ^*, \phi}$ is drawn.
    Then $\normphi=1$ corresponds to IC, $\normphi=0$ results in only $\succ^*$ being sampled and $\normphi=0.5$ results in preferences orders that lie in some sense exactly between the two.  
    Now, given a normalized dispersion parameter $\normphi \in [0,1]$, to generate an SR instance, we draw~$\succ^*$ uniformly at random from $\mathcal{L}(A)$.
    Afterwards, for each agent $a\in A$, we obtain its preferences by drawing a preference order from $\mathcal{D}_{\text{Mallows}}^{\succ^*, \normphi}$ and deleting $a$. 
    The intuition here is that there is a ground truth and agents have a given likelihood to deviate from the ground truth. 
    \item[Euclidean \cite{DBLP:journals/ipl/ArkinBEOMP09}] Given parameter $d\in \mathbb{N}$, for each agent $a\in A$, we uniformly at random sample a point $\mathbf{p}^a$ from $[0,1]^d$. Agent $a$ ranks other agents increasingly by the Euclidean distance between their points, i.e., by $\ell_2(\mathbf{p}^a, \mathbf{p}^b)$ for $b\in A\setminus \{a\}$. The intuition here is that each dimension represents some continuous property of the agents, and agents prefer similar agents. 
    \item[Reverse-Euclidean] Given parameters $p\in [0,1]$ and $d\in \mathbb{N}$, we partition $A$ into two sets $A_1\cupdot A_2$ with $|A_1|=\floor{p\cdot |A|}$.    
    As in the Euclidean model, each agent corresponds to some uniformly at random sampled point $\mathbf{p}^a$ from $[0,1]^p$ and ranks other agents according to their Euclidean distance. 
    However, here an agent $a\in A_1$ ranks agents decreasingly by their Euclidean distance to $\mathbf{p}^a$ and an agent $a\in A_2$ ranks agents increasingly by their Euclidean distance to $\mathbf{p}^a$. The intuition is similar to Euclidean, but a $p$-fraction of agents prefer agents that are different from them.
    \item[Mallows-Euclidean] Given a normalized dispersion parameter $\normphi\in [0,1]$ and some $d\in \mathbb{N}$, we start by generating agents' intermediate preferences $(\succ_a)_{a\in A}$ according to the Euclidean model with $d$ dimensions. 
    Subsequently, for each $a\in A$, we obtain its final preferences by sampling a preference order from $\mathcal{D}_{\text{Mallows}}^{\succ_a, \normphi}$. The resulting instances are perturbed Euclidean instances.
    \item[Expectations-Euclidean] Given parameters $d\in \mathbb{N}$ and $\std\in \mathbb{R}^+$, for each agent $a\in A$, we sample one point~$\mathbf{p}^a$ uniformly at random from $[0,1]^d$. 
    Subsequently, we sample a second point $\mathbf{q}^a$ from~$[0,1]^d$ using a $d$-dimensional Gaussian function with mean $\mathbf{p}^a$ and standard deviation $\std$.
    Agent~$a$ ranks the agents increasingly according to $\ell_2(\mathbf{p}^a, \mathbf{q}^b)$ for $b\in A\setminus \{a\}$. Again, agents are characterized by continuous attributes; however, their ``ideal'' points are not necessarily where they are. 
    \item[Fame-Euclidean] Given parameter $d\in \mathbb{N}$ and $f\in [0,1]$, 
    we sample for each agent $a\in A$ uniformly at random a point $\mathbf{p}^a\in [0,1]^d$ and a $f^a\in [0,f]$. 
    Agent $a$ ranks the other agents increasingly by~${\ell_2(\mathbf{p}^a, \mathbf{p}^b)-f^b}$ for $b\in A\setminus \{a\}$.   
    The intuition is similar as for Euclidean but certain agents have a generally higher quality/fame and are thus attractive partners independent of their location. 
    \item[Attributes \cite{DBLP:conf/soda/BhatnagarGR08}] Given parameter $d\in \mathbb{N}$, for each agent $a\in A$ we uniformly at random sample $\mathbf{p}^a\in [0,1]^d$ and $\mathbf{w}^a\in [0,1]^d$. 
    Agent $a$ ranks the other agents decreasingly by the inner product of $\mathbf{w}^a$ and~$\mathbf{p}^b$, i.e., by $\sum_{i\in [d]} \mathbf{w}^a_i \cdot {\mathbf{p}^b}_i$.
    The intuition is that there are different objective evaluation criteria, but agents assign a different importance to them.
    \item[Mallows-MD] Given a normalized dispersion parameter $\normphi\in [0,1]$, we start with the instance realizing the mutual disagreement matrix described in \Cref{sec:understanding}, i.e., for each $i\in [2n]$ the intermediate preferences $\succ_{a_i}$ of agent~$a_i$ are $ a_{i+1} \succ_{a_i} a_{i+2} \succ_{a_i} \dots \succ_{a_i} a_n \succ_{a_i} a_1 \succ_{a_i} a_2 \succ_{a_i} \dots \succ_{a_i} a_{i - 1}$.
     Subsequently, for each~$a_i\in A$, we obtain its final preferences by sampling a preference order from $\mathcal{D}_{\text{Mallows}}^{\succ_{a_i}, \normphi}$.
    The reason we consider this model is that it covers a part of the map that would otherwise remain uncovered. 
\end{description}

Our dataset consists of $460$ instances sampled from the above described statistical cultures. 
That is, we sampled $20$ instances for each of the following cultures: Impartial Culture, 2-IC with~$p \in \{0.25, 0.5\}$, Mallows with $\normphi \in \{0.2, 0.4, 0.6, 0.8\}$, Euclidean with $d\in \{1,2\}$,  Reverse-Euclidean with $d=2$ and $p \in \{0.05, 0.15, 0.25\}$,  Mallows-Euclidean with $d=2$ and $\normphi \in \{0.2, 0.4\}$, Expectations-Euclidean with $d=2$ and $\std \in \{0.2, 0.4\}$, Fame-Euclidean with $d=2$ and $f \in \{0.2, 0.4\}$, Attributes with $d\in \{2,5\}$, and Mallows-MD with $\normphi \in \{0.2, 0.4, 0.6\}$.
In addition, on our maps, we include the four extreme matrices described in \Cref{sec:understanding}. 

\paragraph{Drawing the Map.}
To draw a map of our dataset, we first compute for each pair of instances from our dataset their mutual attraction distance. 
Subsequently, we embed the instances into the two-dimensional Euclidean space. 
Our goal here is that each instance is represented by a point, and the Euclidean distance of two points on the map should reflect the mutual attraction distance between the two respective SR instances.
 To obtain the embedding, we use a variant of the force-directed Kamada-Kawai algorithm \cite{mt:sapala, DBLP:journals/ipl/KamadaK89}.\footnote{\citet{DBLP:conf/atal/SzufaFSST20} and \citet{DBLP:conf/ijcai/BoehmerBFNS21} used the closely related Fruchterman–Reingold algorithm.} 
The general idea here is that we start with an arbitrary embedding of the instances, then we add an attractive force between each pair of instances whose strength reflects their mutual attraction distance and a repulsive force between each pair ensuring that there is a certain minimum distance between each two points. 
Subsequently, the instances move based on the applied forces until a minimal energy state is reached.
We depict the map visualizing our dataset of $460$ instances for $200$ agents in \Cref{fig:mainMap}.\footnote{We focus on $200$ agents. Maps for different numbers of agents are available in the appendix and look quite similar.} 

To correctly interpret the map, we stress that our embedding algorithm does not optimize some global objective function, e.g., some summed absolute difference between the Euclidean distance of two points on the map and the mutual attraction distance between their respective instances. 
Instead, the visualization algorithm works in a decentralized fashion also aiming at producing a visually pleasing image. 
Consequently, the position of instances on the map can be different in different runs and certainly depend on which other instances are part of the map. 
Thus, in the following, if we say that two instances are close to each other, then we refer to their mutual attraction distance which is typically but not necessarily reflected by them being close on the map.
To verify the quality of the embedding, in the appendix, we compute the embedding's distortion and find that while the embedding is certainly not perfect, most of the distances are represented adequately.
We want to remark that some error is to be expected here as the space of SR instances under the mutual attraction distance is highly complex; however, the general picture the map provides is indeed correct and helpful to get an intuitive interpretation of experimental~results. 

\begin{figure*}[t]
    \centering
    \includegraphics[width=10cm]{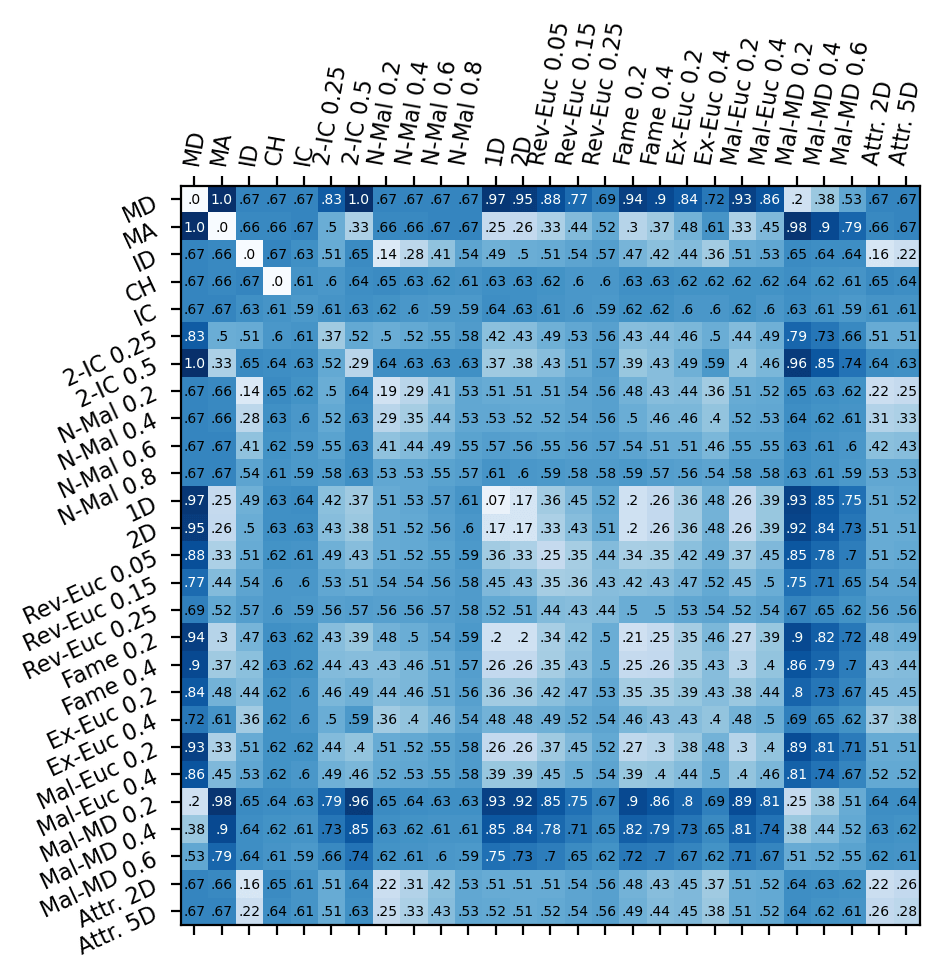}
    \caption{For each pair of statistical cultures, average mutual attraction distance between instances sampled from the two. The first four lines/columns contain, for each statistical culture, the average distance of instances sampled from that culture to our four extreme matrices. The diagonal contains the average distance of two instances sampled from the same statistical culture.}
    \label{fig:dis_matrix}
\end{figure*}

\begin{figure*}[t]
    \centering
         \begin{subfigure}[b]{\textwidth}
         \centering
        \includegraphics[trim={1.4cm 2.0cm 0.3cm 0.1cm},clip,width=8cm]{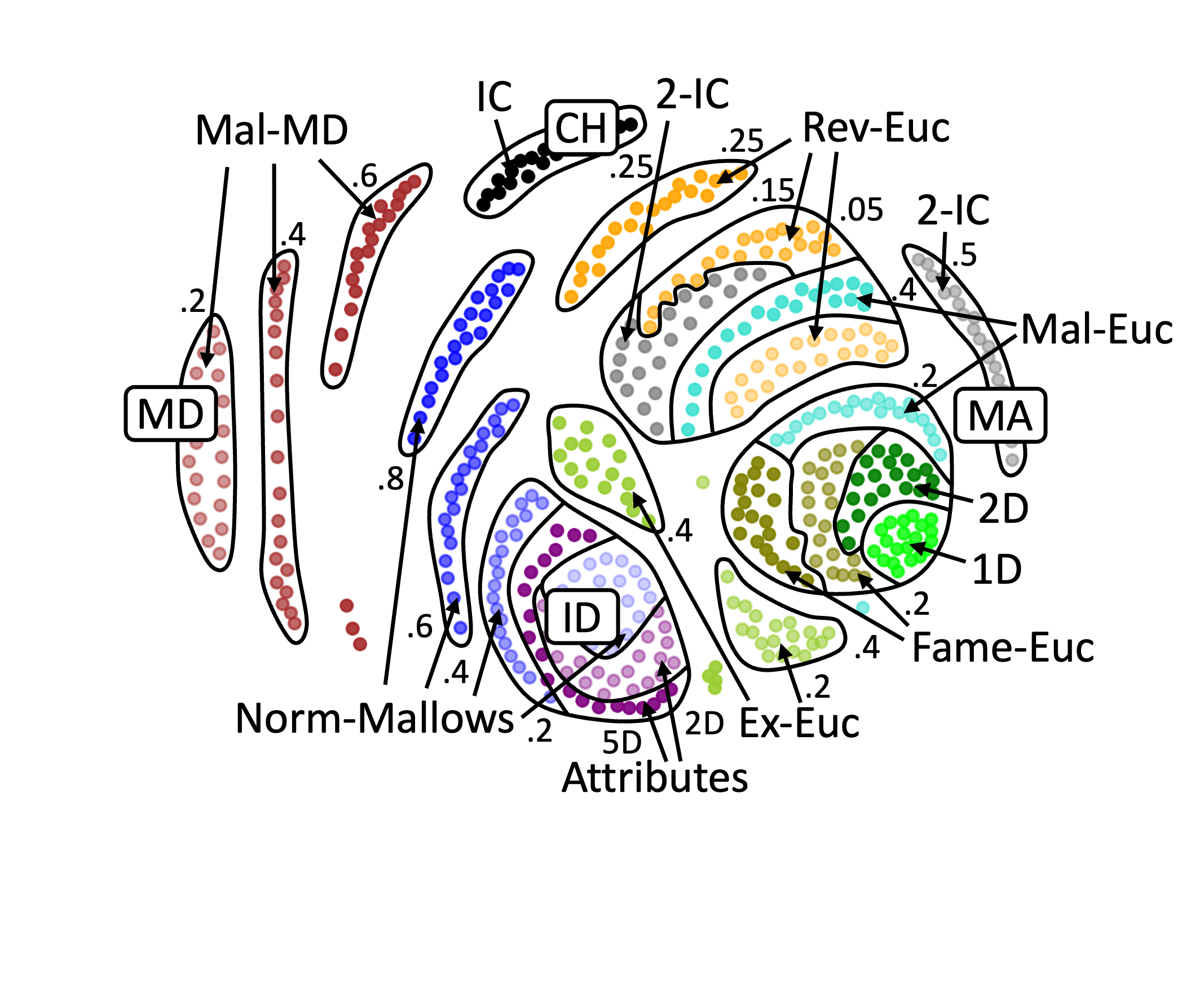}
         \caption{The color of a point indicates the statistical culture it was sampled from}
        \label{fig:mainMap}
     \end{subfigure} \\
     \begin{subfigure}[b]{0.47\textwidth}
         \centering
         \includegraphics[trim={1.4cm 0.1cm 0.3cm 0.1cm}, clip,width=7cm]{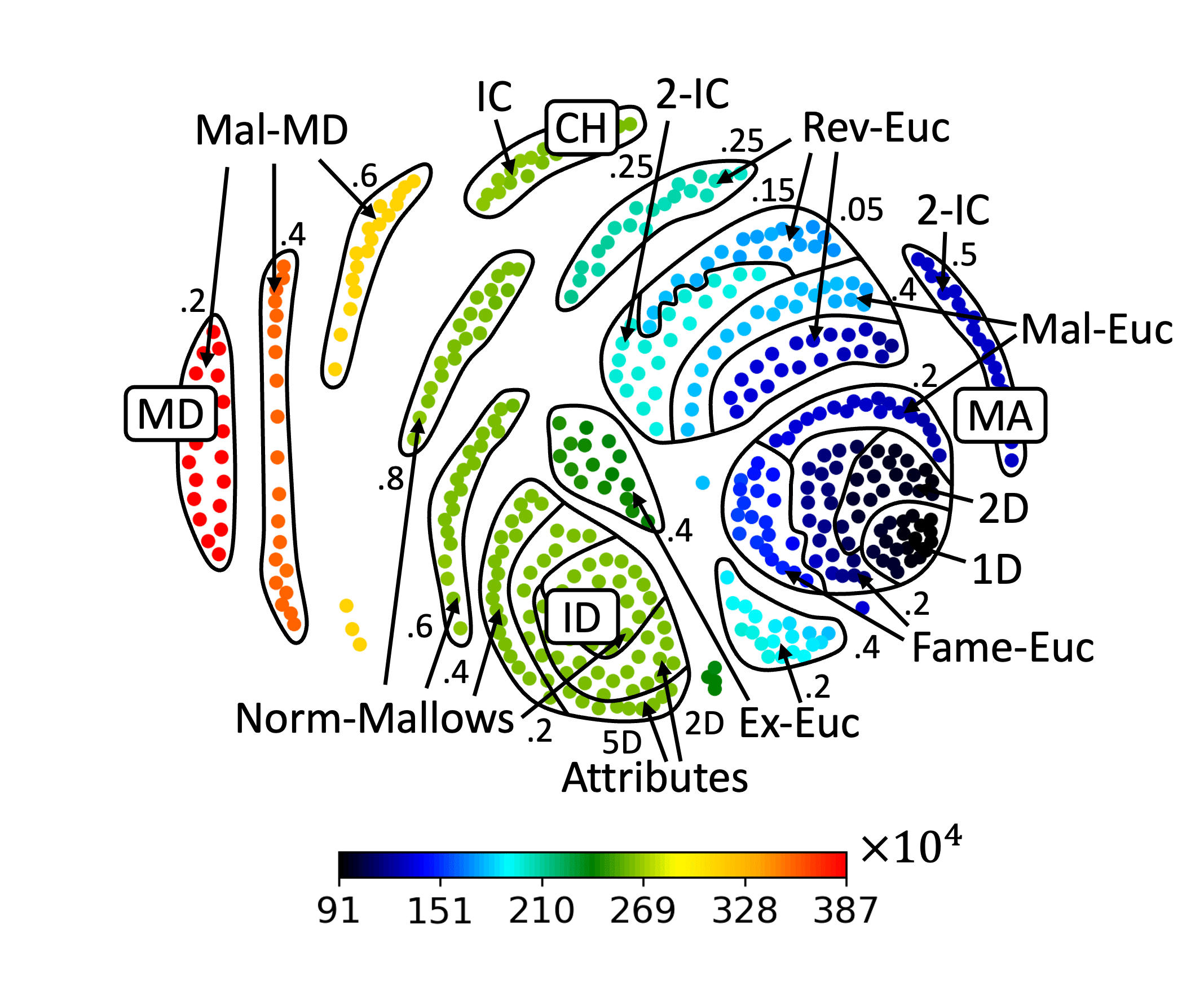}
         \caption{Mutuality}
         \label{fig:mutuality}
     \end{subfigure} \qquad
      \begin{subfigure}[b]{0.47\textwidth}
         \centering
         \includegraphics[trim={1.4cm 0.1cm 0.45cm 0.1cm},clip, width=7cm]{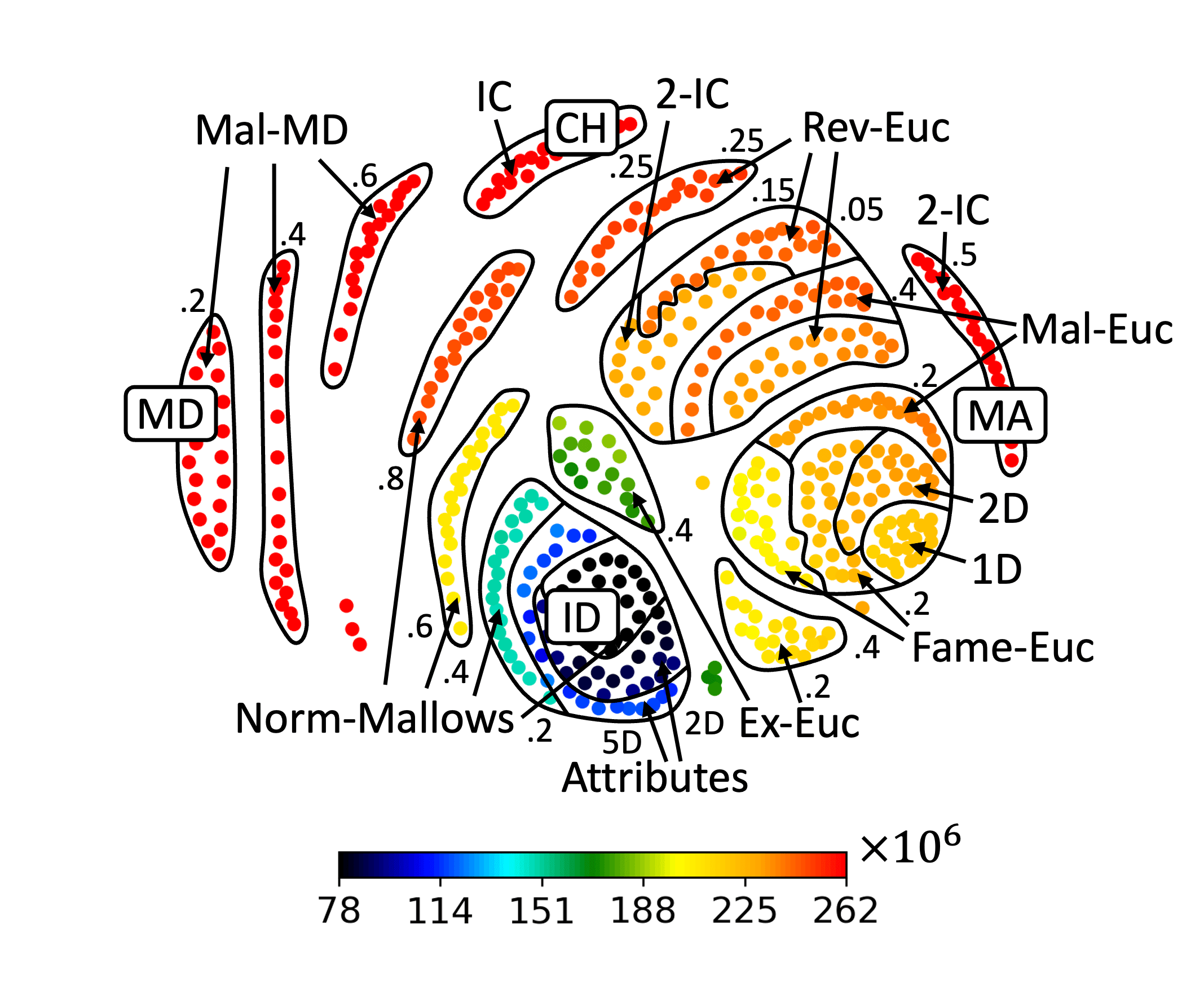}
         \caption{Rank distortion}
         \label{fig:rank_distortion}
     \end{subfigure}
      
    \caption{For \Cref{sec:map}, maps of $460$ SR instances for $200$ agents visualizing different quantities for each instance. Each instance is represented by a point. Roughly speaking, the closer two points are on the map, the more similar are the respective SR instances under the mutual attraction distance. }
    \label{fig:my_label_5}
    
\end{figure*}

\subsection{Understanding the Map} \label{sub:eval}
We now take a closer look at the map of SR instances shown in \Cref{fig:mainMap}.
Examining the map, what stands out is that for all cultures, instances sampled from this culture are placed close to each other on the map. 
This is also highlighted by the fact that for most cultures, we were able to draw a box around all instances from this culture and by the general island-like structure of the map. 
In fact, instances sampled from the same culture are usually similar to each other (or at least more similar to each other than to instances sampled from other cultures). 
While this is to be expected to a certain extent, this observation validates our approach in that the mutual attraction distance is seemingly able to identify the shared structure of instances sampled from the same statistical culture and in that our embedding algorithm is able to detect these clusters. 
Moreover, interestingly, the different statistical cultures have a different ``variation'', i.e., the average mutual attraction distance of two instances sampled from the culture substantially differ for the different cultures.
The Impartial Culture model has with $0.59$ the highest variation, while the Euclidean model for $d=1$ has with $0.07$ the lowest variation  (see \Cref{fig:dis_matrix}).
The value for Impartial Culture is quite remarkable, as it means that Impartial Culture instances are on average almost as far away from each other as, for instance, $\ID$ from the other extreme points.
Because of the limitations of two-dimensional Euclidean space, this is not adequately represented on the map, as Impartial Culture instances are still placed close to each other (the reason for this is that they are all even further away from instances sampled from the other cultures than from each other). 
Nevertheless, in the following experiments, we observe that Impartial Culture instances behave quite similarly. 

Taking a closer look at the map, we observe that our four extreme points indeed fall into four different parts of the map.
On the right, we have the mutual agreement matrix MA. 
Accordingly, models for which mutual agreement is likely to appear all land in the right part of the map, namely, Euclidean instances (where intuitively speaking agent $a$ likes agent $b$ if they are close to each other making it also likely that $b$ likes $a$), the Fame-Euclidean model for $f=0.2$, the Mallows-Euclidean model for $\normphi=0.2$, and the Reverse-Euclidean model for $p=0.05$ (these three are all basically differently perturbed variants of Euclidean models and consequently also on average slightly further away from $\sym$ than Euclidean instances), and the 2-IC model for $p=0.5$ (where we have some guaranteed level of mutual agreement because there are two groups of agents and agents from one group prefer each other to the agents from the other group). 
For all these cultures, the average normalized mutual attraction distance to MA is $0.33$ or smaller, with Euclidean for $d=1$ being the closest culture with an average of $0.25$ (we were slightly surprised by how small the difference between Euclidean for $d=1$ and 2-IC is here). 

On the left, we have the mutual disagreement matrix MD with only instances from the Mallows-MD model being close to it.
Note that, in general, it is to be expected that if we apply the Mallows model on top of some other model $\mathcal{X}$, then for small values of $\normphi$ the sampled instances are close to the ones from $\mathcal{X}$ but move further and further away as $\normphi$ grows (and naturally towards Impartial Culture instances as for $\normphi=1$ the models coincide).
Given that the mutual (dis)agreement matrices are at the two ends of the horizontal axis, this raises the question whether the horizontal axis can be indeed interpreted as an indicator for the degree of mutuality in SR instances. 
This hypothesis gets strongly confirmed in \Cref{fig:mutuality} where we color the points on the map according to their mutuality value, which we define as the total difference between the mutual evaluations of agent pairs, i.e., $\sum_{a\in A} \sum_{i\in [|A|-1]} |\MA(a,i)-i|$. 
The nicely continuous shading in \Cref{fig:mutuality} indicates a strong correlation between the mutuality value of an instance and its $x$-coordinate on the map. Moreover,  instances that are close on the map have indeed similar mutuality values.  
Moreover, the continuous coloring indicates that our dataset provides a good and almost uniform coverage of the space of SR instances (at least in terms of their mutuality value).

Turning to the middle part of the map, the identity matrix ID can be found at the bottom. 
Close to identity are instances from cultures where agent's quality is ``objective". Namely, Mallows model with $\normphi=0.2$ (where the preferences of agents are still expected to be close to the central order) and the Attributes model with $d=2$ (where each agent has two quality scores and the preferences of agents only differ in how they weight the quality scores).
The chaos matrix CH is placed in the top part of the map together with the ``chaotic'' Impartial Culture instances. 
Mallows instances naturally form a continuous spectrum between identity and chaos.
These observations give rise to the hypothesis that in instances placed at the bottom of the map most agents have similar preferences, while in instances placed at the top all agents have roughly the same quality and few agents are particularly (un)popular.
To quantify whether agents agree or disagree on the quality of the agents, we measure the rank distortion of an instance, i.e., for each agent we sum up the absolute difference between all pairs of entries in its mutual attraction vector $\sum_{a\in A} \sum_{i,j\in [|A|-1]} |\MA(a,i)-\MA(a,j)|$. 
Note that, for example, for an agent that is always ranked in the same position by all other agents this absolute difference is zero. 
We show in \Cref{fig:rank_distortion} a map colored by the rank distortion of instances.
The picture here is slightly different than for the horizontal axis in that instances with the same $y$ coordinate might still have a very different rank distortion. 
In fact, what we see here is that the further a point is from $\ID$ on the map, the larger is the rank distortion and thus the higher is the disagreement concerning agents quality (which is quite intuitive recalling that for both $\sym$ and $\asym$ the rank distortion is maximal). 

\section{Using the Map} \label{sec:using}
To illustrate the usefulness of the map to evaluate experiments and to check whether instances that are close to each other on the map have similar properties,  we perform multiple exemplary experiments.

\subsection{Blocking Pairs and Stable Matchings}\label{sec:BPS}
We start in this subsection by analyzing various properties related to the number of blocking pairs that block a matching. 
Specifically, we first compute for each SR instance the minimum number of blocking pairs for some matching, then the average number of blocking pairs for a random matching, and lastly the number of blocking pairs for a minimum-weight matching.
We visualize the results of our experiments in \Cref{fig:maps-bps}. 

\begin{figure*}[t]
	\centering
 \begin{subfigure}[t]{0.45\textwidth}
         \centering
         \includegraphics[trim={0.1cm 0.1cm 0.2cm 0.1cm}, clip,width=7cm]{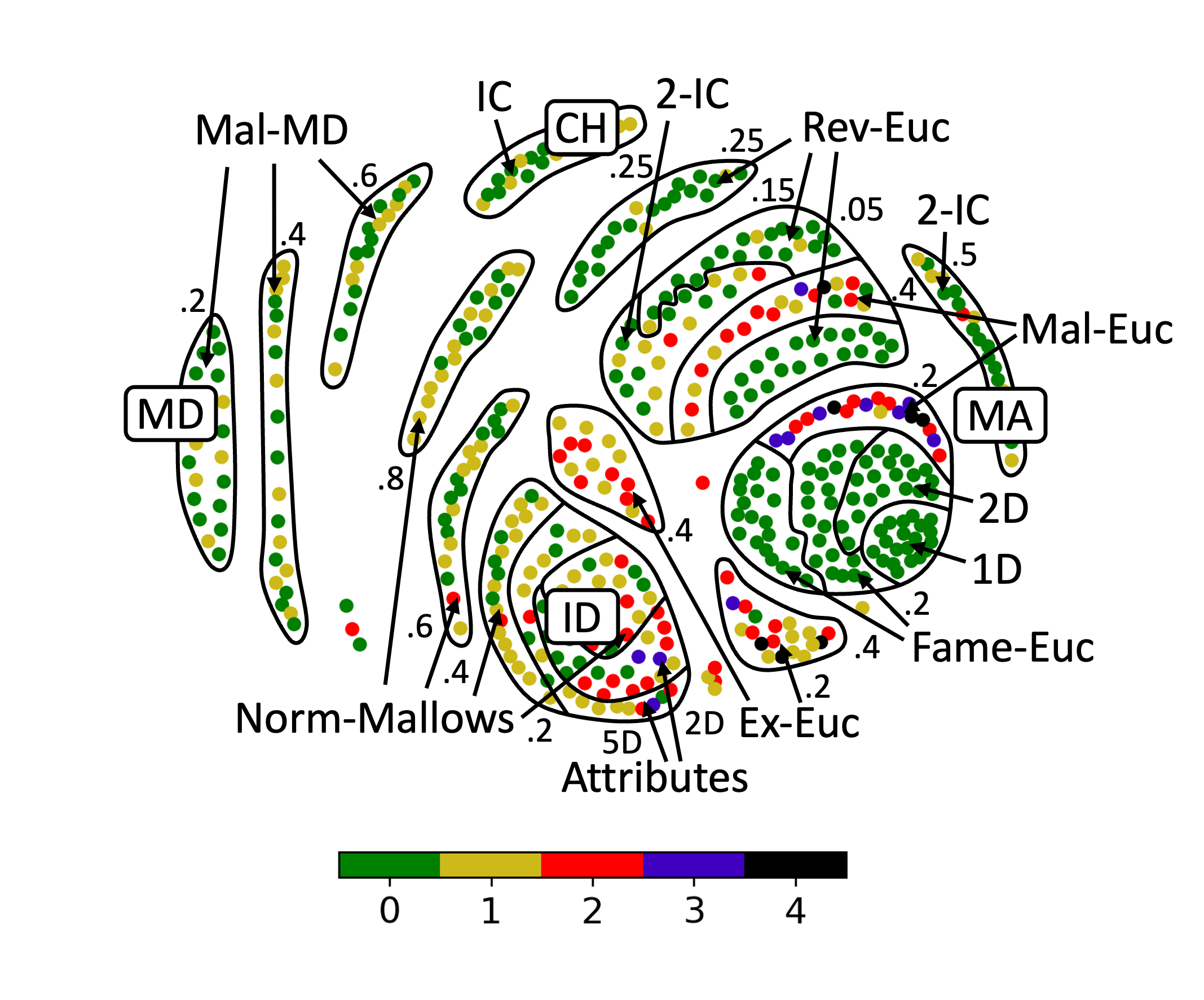}
        \caption{Minimum number of blocking pairs}\label{fig:min_num_bps}
     \end{subfigure} \qquad
  \begin{subfigure}[t]{0.45\textwidth}
         \centering
         \includegraphics[trim={0.3cm 0.1cm 0.16cm 0.1cm}, clip,width=7cm]{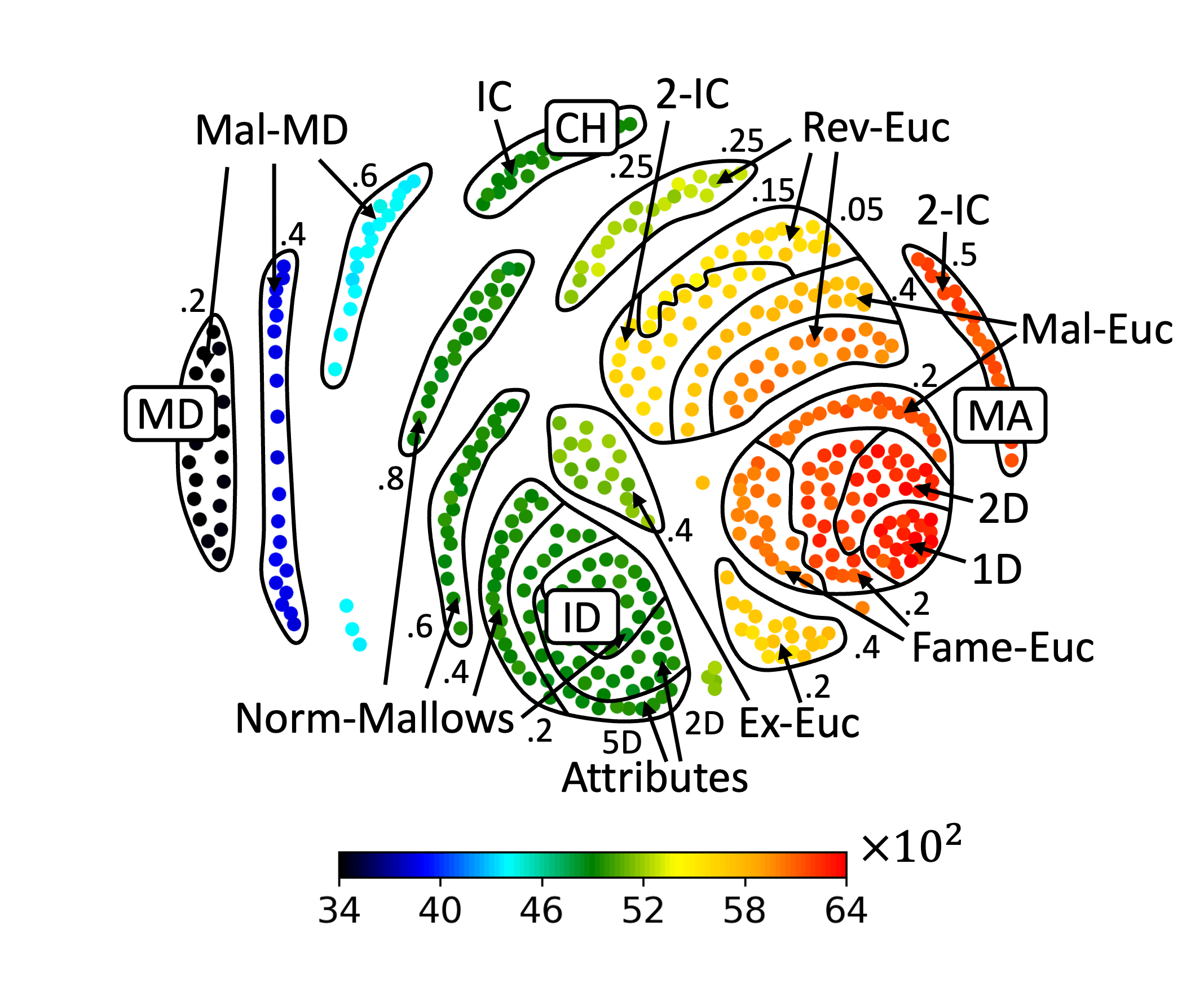}
         \caption{Average number of blocking pairs for a perfect matching}
         \label{fig:avg_num_of_bps}
     \end{subfigure}
      \begin{subfigure}[t]{1\textwidth}
         \centering
         \includegraphics[trim={1.82cm 0.1cm 0.1cm 0.1cm},clip,width=7cm]{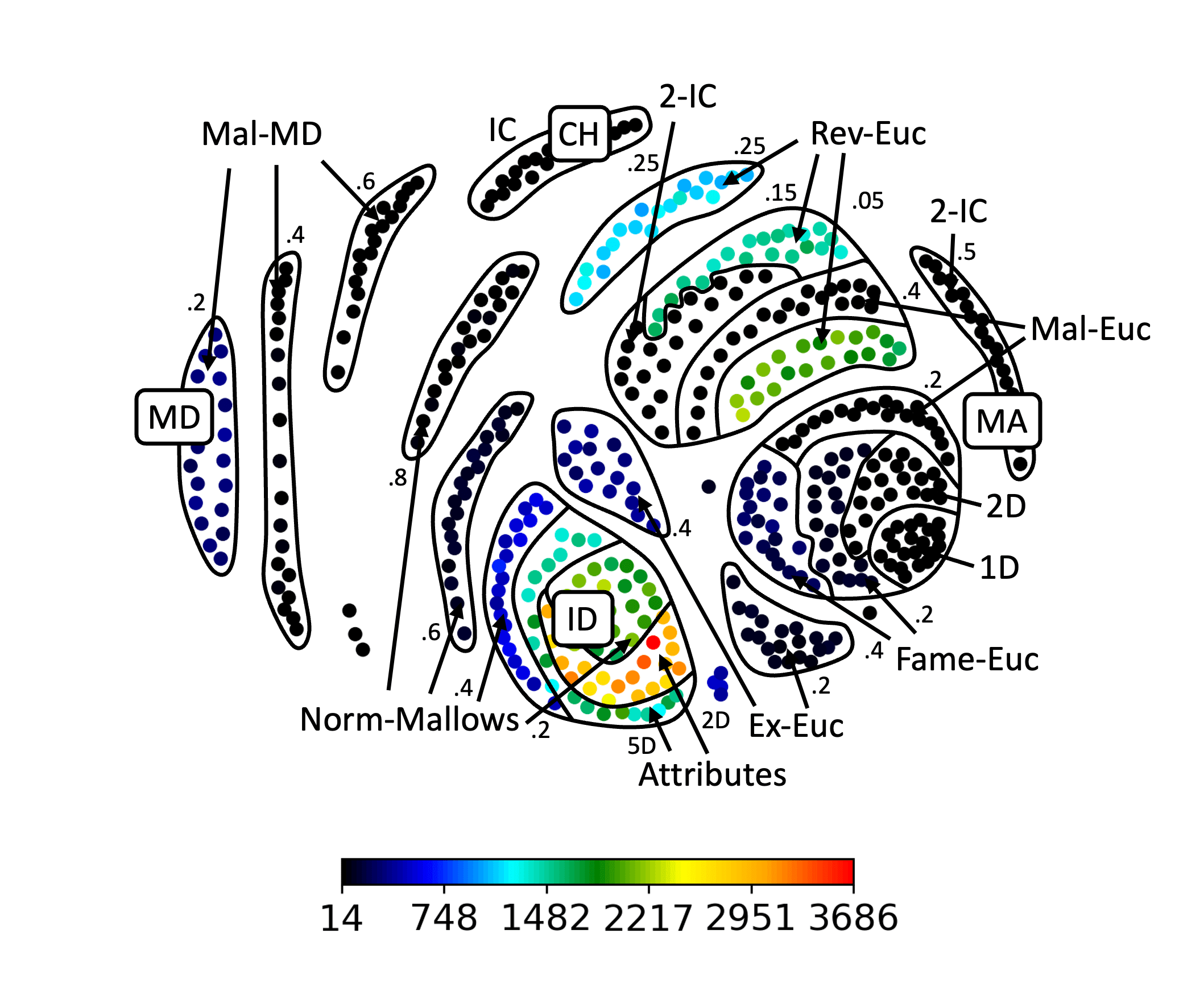}
        \caption{Number of blocking pair for a minimum-weight matching. The scale is logarithmic.}\label{fig:min-bp}
     \end{subfigure}
      
    \caption{For \Cref{sec:BPS}, maps of $460$ SR instances for $200$ agents visualizing different quantities for each instance.}
    \label{fig:maps-bps}
    
\end{figure*}

\subsubsection{Blocking Pair Minimizing Matching}\label{sec:min-BPS}
Naturally, the most important question related to an SR instance is whether the instance admits a stable matching or not. 
Slightly more nuanced, it is also possible to ask for a matching minimizing the number of blocking pairs.
As computing the minimum number of blocking pairs in an SR instance is NP-hard \cite{DBLP:conf/waoa/AbrahamBM05}, we solve this problem using an ILP.
We visualize the results of this experiment on the map in \Cref{fig:min_num_bps}.

First, considering which instances admit a stable matching (green points on the map), we do not see a clear correlation with the instance's position on the map. 
This is also quite intuitive, given that whether an instance admits a stable matching might depend on some local configuration. Such configurations can naturally not be fully captured in the mutual attraction matrix. 
However, what is clearly visible is that for different cultures the probability of admitting a stable matching is quite different: 
On the one hand, instances sampled from the Euclidean, Fame-Euclidean, and Reverse-Euclidean model almost always admit a stable matching (for the Euclidean model this is even guaranteed). 
On the other hand, instances sampled from the Mallows-Euclidean and Expectations-Euclidean model only very rarely admit a stable matching. 
The drastic contrast between the Euclidean model and the Mallows-Euclidean model with $\normphi=0.2$ and between the Reverse-Euclidean and Expectations-Euclidean model are quite remarkable, as they are conceptually quite similar. 

However, moving to the minimum number of blocking pairs, the picture becomes more uniform:.
A large majority of the map (and cultures) solely consist of SR instances where the minimum number of blocking pairs is at most one (recall that all our experiments here are for $n=200$ agents). 
However, in instances sampled from the Mallows-Euclidean and Expectations-Euclidean model (which only rarely admit a stable matching)  the minimum number of blocking pairs is often two or more. 
Some further such instances can be found close to identity. 

Overall, the minimum number of blocking pairs clearly depends on the model from which the relevant SR instance was sampled, leading to a clustering of (very close to) stable instances on the map. 
However, there are also regions on the map exhibiting a mixed picture, for instance the regions around $\MA$ and $\ID$; interestingly, it seems that while highly structured instances always admit a stable matching (like Euclidean instances where this is even guaranteed to be the case), slightly perturbing these instances leads to an increase in the number of blocking pairs.
Lastly, it is remarkable that in all our $460$ SR instances a matching admitting at most four blocking pairs exists, indicating that instances which are ``far away'' from stability are quite exceptional. 

\subsubsection{Expected Number of Blocking Pairs} \label{sub:expNumberBPs}
Motivated by the fundamental importance of blocking pairs for stable matchings, we measure the expected number of blocking pairs for an arbitrary perfect matching.
For this, for each instance, we sampled $100$ perfect matchings uniformly at random and for each counted the number of blocking pairs. 
The results are depicted in \Cref{fig:avg_num_of_bps}.
As for the mutuality value, we get a nicely continuous shading along the horizontal axis. 
This clear correlation with the mutuality value is quite intuitive as in case there is a high mutual agreement agents are also more likely to form blocking pairs (if an agent $a$ prefers an agent $b$ to its current partner, then because the mutuality is high $b$ also tends to like $a$ and tends to prefer $a$ to its current partner); if there is mutual disagreement, the picture is reversed (an agent prefers the agents to its current partner that tend to dislike it). 
This is also clearly visible in \Cref{fig:avg_num_of_bps}, as instances close to $\asym$ have a low expected number of blocking pairs, whereas for instances close to $\sym$ the expected number is much higher.  Moreover, \Cref{fig:avg_num_of_bps} again validates that instances that are close to each other on the map have similar properties and that our test dataset provides a good and uniform coverage of the space of SR instances. 

\subsubsection{Number of Blocking Pairs for Minimum-Weight Matching} \label{sub:min-weightBPs}
We define the minimum-weight matching $M$ in an instance as the perfect matching minimizing the summed rank that agents assign to their partner, i.e., $M$ minimizes $\sum_{a\in A} \pos_{\succ_a}\big(M(a)\big)$. 
If stability is not vital or if a stable matching might be too complicated to compute, a minimum-weight matching is a natural candidate matching to choose and might even serve as a heuristic for choosing a stable matching. 
Thus, it is an interesting question ``how'' stable such minimum-weight matchings are.
We depict in \Cref{fig:min-bp} the number of pairs that block a minimum-weight matching for all instances from our dataset.
Analyzing the results, we find that for almost all of our instances from the dataset, a minimum-weight matching is only blocked by few pairs. 
There are two exceptions: 
First, instances sampled from the Reverse-Euclidean model and instances close to identity. 
For both of these types of instances, the number of pairs blocking the minimum-weight matching can be quite substantial. 
For Reverse-Euclidean, we see that these instances behave very differently than instances close to them sampled from different models. 
Remarkably, for this model the higher $p$ gets, the lower gets the number of pairs that block a minimum-weight matching. 
For the identity region, we see that in all instances from this region is the minimum-weight matching blocked by many pairs. 
What stands out from the map again is that instances sampled from the same model exhibit a very uniform behavior. 
Overall, the outlier-behavior of Reverse-Euclidean (as in \Cref{fig:summed_rank_value}) underlines that the map is not perfect and our mutual attraction distance and matrix (naturally) do not capture all facets of similarity. 
Nevertheless, the other results highlight the usefulness of the map as a visualization tool and an intuition provider. 

\subsection{Different Types of Stable Matchings} \label{sec:typesOfSM}

\begin{figure*}[t]
	\centering
     \begin{subfigure}[t]{0.45\textwidth}
         \centering
        \includegraphics[trim={0.1cm 0.1cm 0.2cm 0.1cm}, clip,width=7cm]{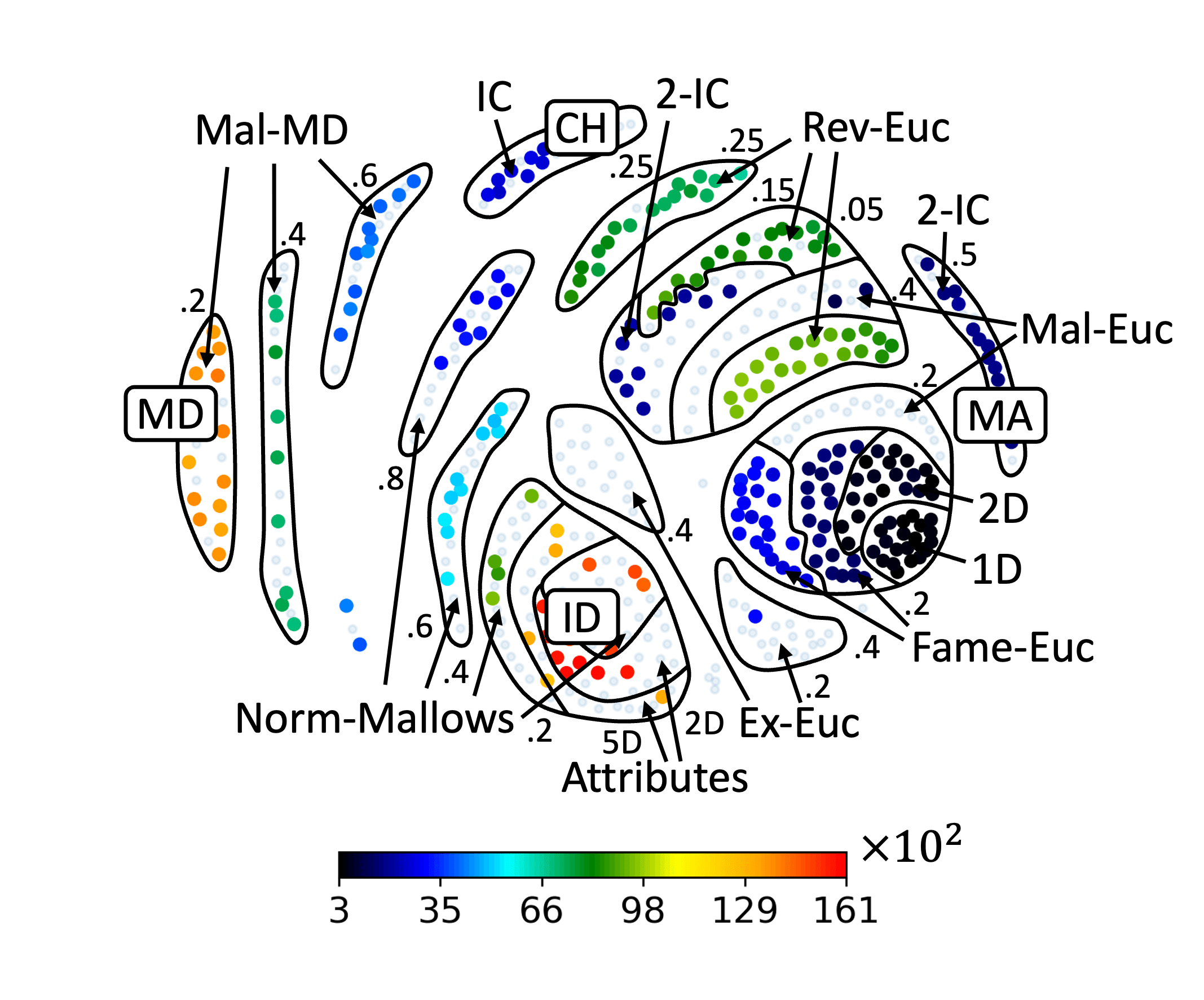}
         \caption{Minimum summed rank of agents for partner in any stable matching (transparent points have no stable matching)}
         \label{fig:summed_rank_value}
     \end{subfigure} \qquad
      \begin{subfigure}[t]{0.45\textwidth}
         \centering
        \includegraphics[trim={0.1cm 0.1cm 0.1cm 0.1cm}, clip,width=7cm]{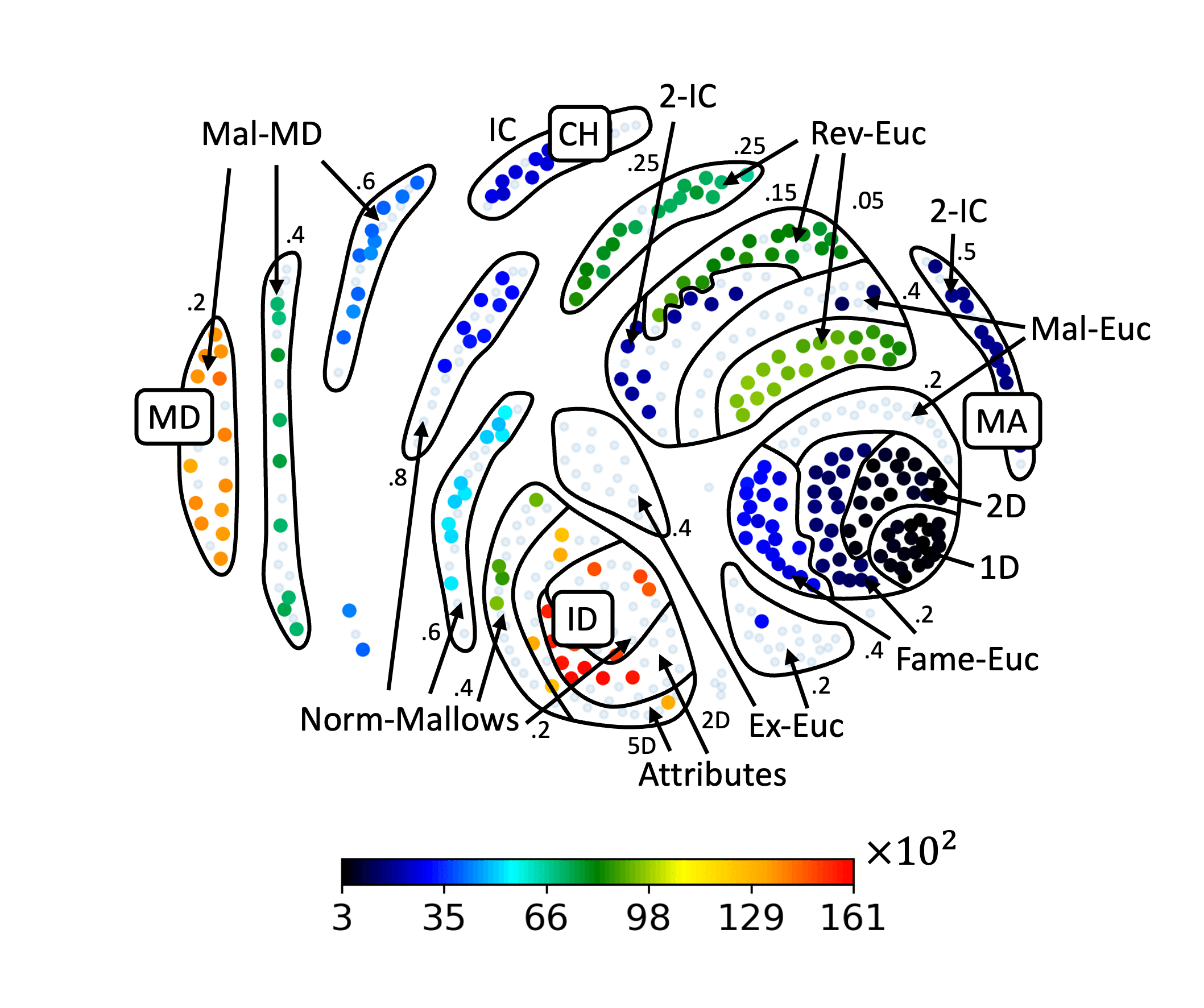}
    \caption{Maximum summed rank of agents for partner in any stable matching (transparent points do not admit a stable matching)}
  \label{fig:max-rank}
     \end{subfigure}
     \begin{subfigure}[t]{1\textwidth}
         \centering
        \includegraphics[trim={0.1cm 0.1cm 0.1cm 0.1cm},clip,width=7cm]{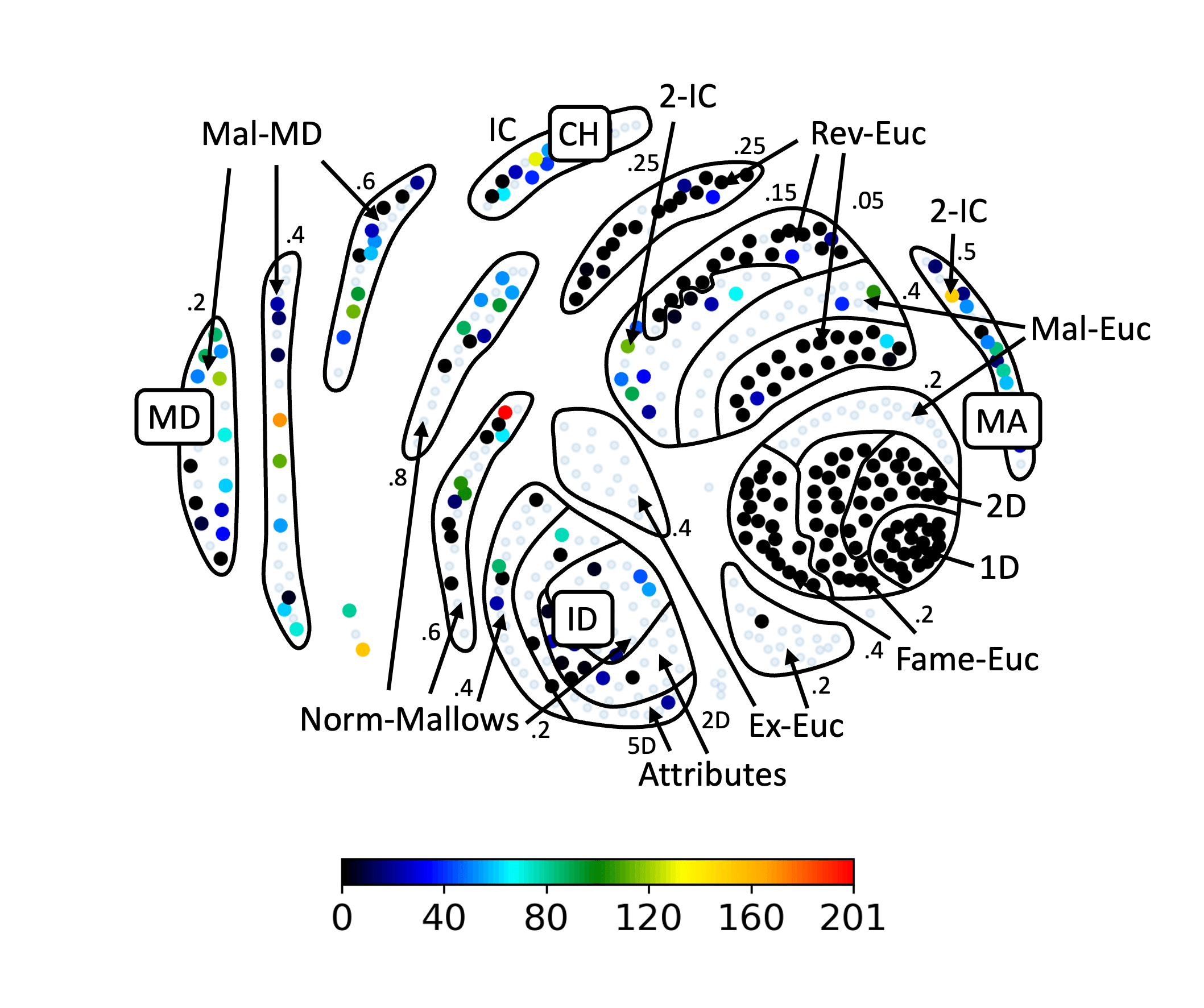}
  \caption{Difference between maximum and minimum summed rank of agents for partner in any stable matching (transparent points do not admit a stable matching)}
  \label{fig:diff-rank}
     \end{subfigure}
    \caption{For \Cref{sec:typesOfSM}, maps of $460$ SR instances for $200$ agents visualizing different quantities for each instance.}
    \label{fig:my_label_5}
    
\end{figure*}

In this section, we restrict our focus to instances that admit a stable matching. 
For them we compute different types of stable matchings maximizing certain objectives and compare their quality. 

\subsubsection{Summed Rank Stable Matchings}\label{sub:summed-rank}

\paragraph{Summed Rank Minimal Matching.}
We start by analyzing summed rank minimal stable matchings, i.e., stable matchings $M$ minimizing $\sum_{a\in A} \pos_{\succ_a}\big(M(a)\big)$ (these matchings are also sometimes called egalitarian matchings). 
Such a matching can also be interpreted as a matching maximizing the summed satisfaction of agents and is thus a natural candidate to pick if multiple stable matchings exist.
However, computing it is NP-hard \cite{DBLP:journals/jcss/Feder92} and thus we resorted to an ILP.
We visualize the quality of summed rank minimal stable matchings in \Cref{fig:summed_rank_value} and depict instances without a stable matching as transparent points.
Focusing on instances that admit a stable matching, first observe that instances sampled from one culture again behave remarkably similarly. 
In addition, there is some but certainly not a perfect correlation between the results and instances' position on the map: 
Ignoring Reverse-Euclidean instances which are a clear outlier here, if we move from chaos to mutual agreement the minimal summed rank decreases (as for perfect mutual agreement every agent can be matched to its top-choice); in contrast, 
if we move from chaos to mutual disagreement or from chaos to identity, then the minimal summed rank constantly increases. 
Remarkably, instances close to identity have a higher minimum summed rank than instances close to mutual disagreement, while for instances realizing the two matrices, the minimum summed rank for an instance with $2n$ agents is $2n^2$.

\paragraph{Summed Rank Maximal Matching.}
We now consider the in some sense opposite objective. 
That is, we analyze the stable matching that maximizes the rank that agents assign to their partner (this is the in some sense worst stable matching that minimizes agent's satisfaction). 
We visualize the results in \Cref{fig:max-rank}. 
While there is no simple pattern visible on the map, disregarding Reverse-Euclidean instances, there is a clear correlation of instance's behavior and their position on the map. 
The behavior can be nicely described by moving along the extreme matrices. 
Moving from mutual disagreement to chaos, the maximum summed rank constantly decreases and if we move further towards identity (ignoring Reverse-Euclidean instances) it decreases even further. 
If we move from chaos to identity, the maximum summed rank substantially increases, while moving from identity to mutual disagreement it first decreases and then increases again. 

\begin{figure}[t]
\centering
\includegraphics[trim={0.1cm 0.1cm 0.15cm 0.1cm}, clip,width=7cm]{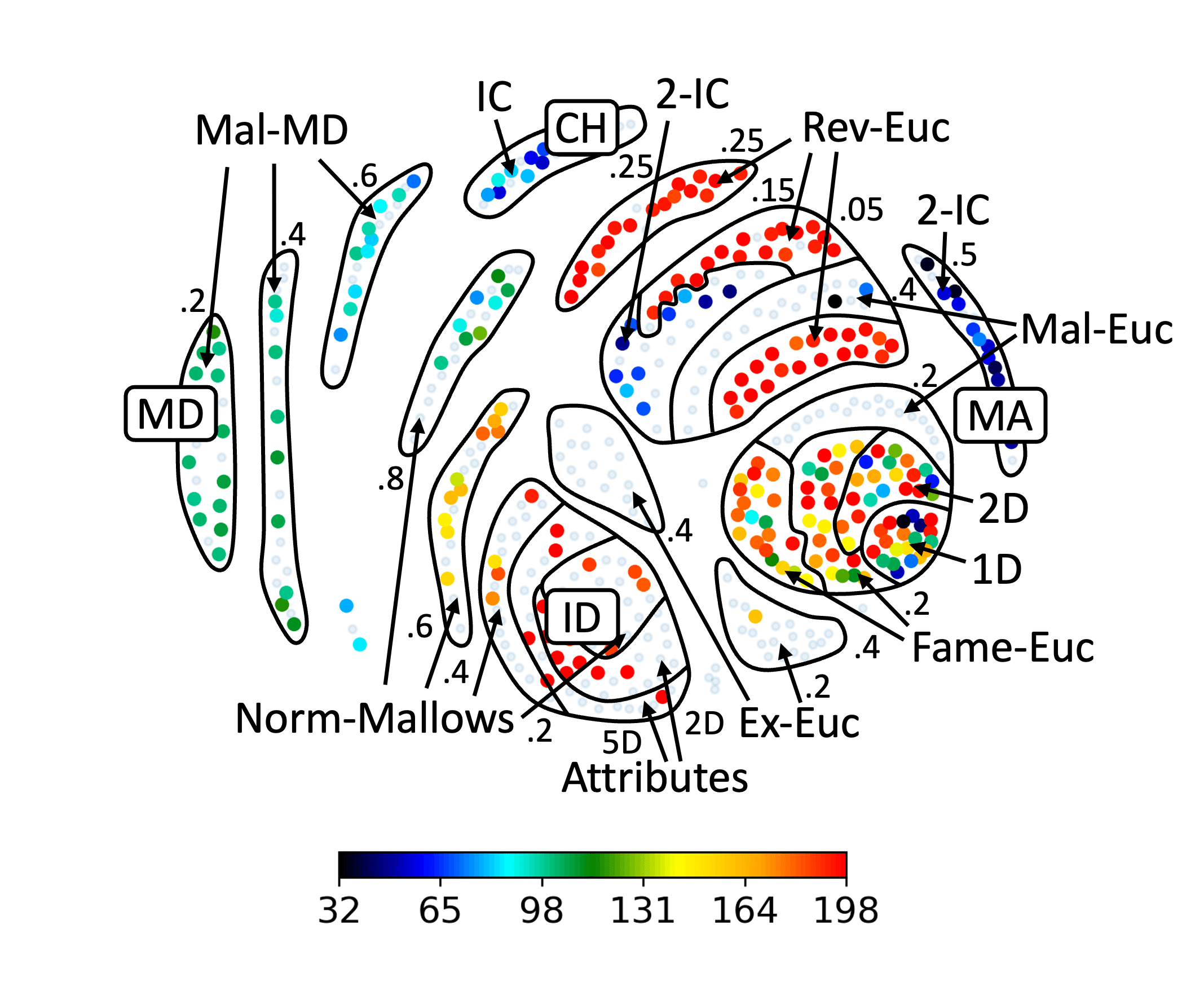}
\caption{Maximum rank an agent assigns to its partner in a stable matching minimizing this value (transparent points do not admit a stable matching)}\label{fig:maxrank}
\end{figure}

\paragraph{Difference Summed Rank Minimal and Maximal Stable Matching.} Lastly, we ask the question how large the influence of the selected stable matching is.
So how much does it matter which matching is selected?
We restrict our focus to the summed agent's satisfaction and quantify this influence as the difference between the maximum and minimum summed rank of agents for their partner in a stable matching.
This quantity might also serve as an indicator for the ``richness'' of the set of stable matchings. 
We present the results in \Cref{fig:diff-rank}. 
Remarkably, this is the first of our maps where we see a very different behavior of instances sampled from the same culture. 
However, this might be due to the fact that in terms of the overall satisfaction of agents it does not seem to make a substantial difference in most cases in which stable matching is selected. 
Notably, for numerous instances it makes nearly no difference, especially for instances sampled from the Euclidean model or similar models. 
The latter observation is quite intuitive, as in Euclidean instances there only exists a single stable matching. 

\subsubsection{Maximal Rank Minimizing Stable Matching} \label{sub:maxrank}

In our last experiment about different types of stable matchings, we consider the stable matching that maximizes the satisfaction of the agent worst off.
That is, we consider the maximum rank an agent assigns to its partner, i.e., $\max_{a\in A} \pos_{\succ_a}(M(a))$ in a stable matching $M$ minimizing this value. 
This matching is also know as a minimum regret stable matching and can be computed in linear time \cite{DBLP:books/daglib/0066875}. 
This matching is another naturally attractive special stable matching. 
We show the results in \Cref{fig:maxrank}.
Examining the map, the first remarkable observation is that instances which are close to Euclidean instances may behave very differently, even if they are sampled from the same statistical culture.
A possible explanation for this is that in those instances stable matchings are often unique leaving little flexibility to satisfy the agent worst-off. 
Moreover, again, clear patterns on the map can be identified. 
For impartial culture instances and 2-IC instances the satisfaction of the worst-off agent is quite high. 
For instances close to mutual agreement the picture is quite mixed, while for instances close to identity, it is not possible to satisfy all agents adequately (this is actually quite intuitive because someone needs to be matched to the agents that are collectively considered to have a low quality in such instances).
Moving from identity to mutual disagreement or chaos the situation of the worst-off agent constantly improves. 

\subsection{Running Time Analysis}

\begin{figure}
         \centering
         \includegraphics[trim={1.4cm 0.1cm 0.45cm 0.1cm}, clip,width=7cm]{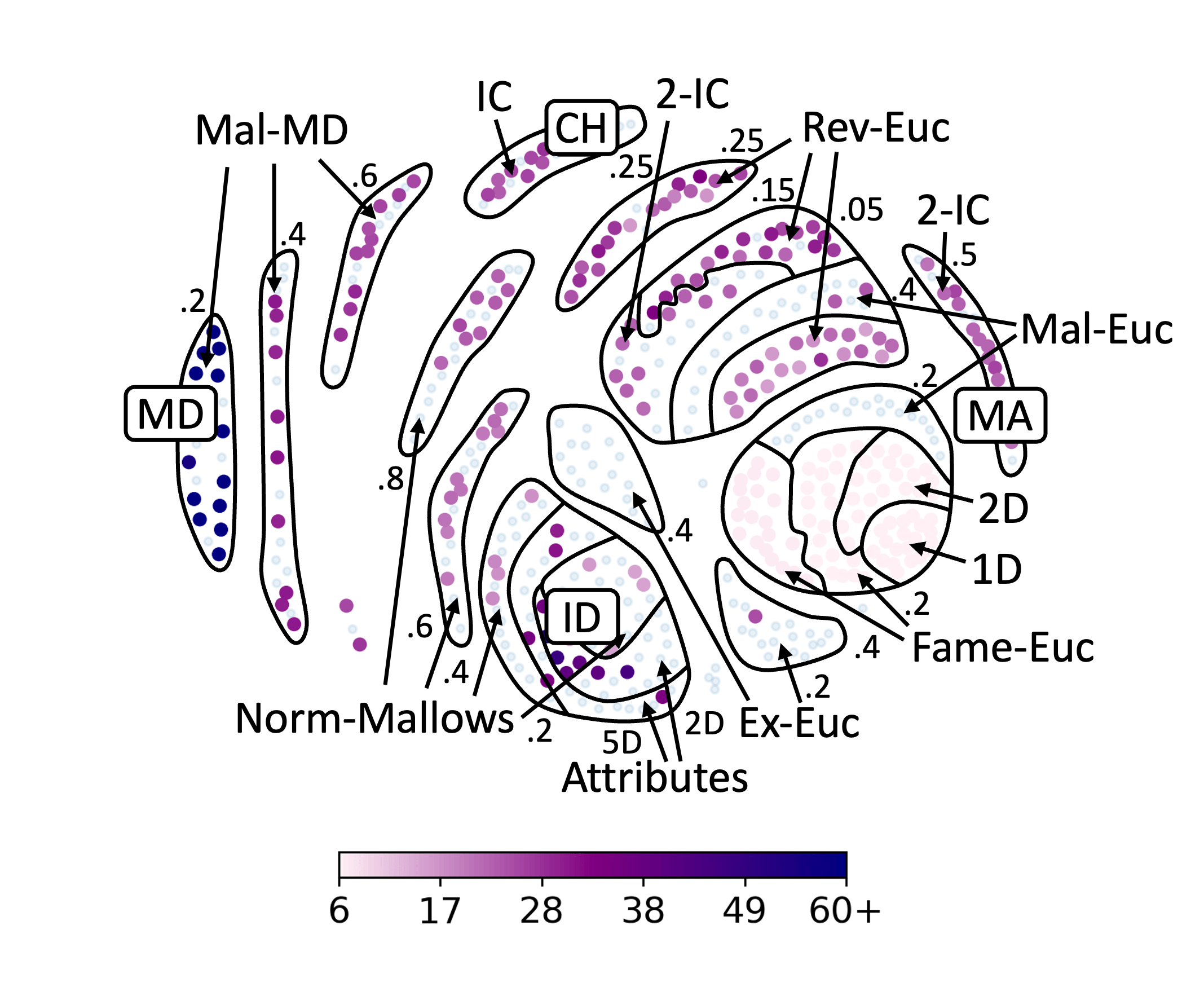}
         \caption{Seconds needed to compute summed rank minimal stable matching (transparent points have no stable matching)}
         \label{fig:summed_rank_time}
\end{figure}

Lastly, to illustrate another possible application of the map, in \Cref{fig:summed_rank_time} we visualize the time our ILP, which we solved using \citet{gurobi}, needed to find a summed rank minimal stable matching (from \Cref{sub:summed-rank}). 
Analyzing the results, again instances from the same culture behave quite similar to each other and the results are clearly connected to instances' position on the map. 
More specifically, instances from the Euclidean and Fame-Euclidean model seem to be particularly easy to solve, whereas instances close to $\ID$ and close to $\asym$ seem to be particularly challenging, maybe because here the achievable minimum summed rank is quite high.
Remarkably, for election-related problems typically Impartial Culture elections are most challenging and the more structure there is in an election, the easier it is to solve \cite{DBLP:conf/atal/SzufaFSST20}. In sharp contrast to this, we observe that instances close to $\ID$ and $\asym$, which are both heavily structured, are particularly challenging.
We remark that naturally our observations on which instances are easy and which are hard are limited to the specific problem and solution method.

\section{Outlook: A Map of Stable Marriage Instances} \label{sec:map-of-SM}
The framework developed in this paper to draw a map of synthetic \textsc{Stable Roommates} (\textsc{SR}) instances can also be applied to different types of matching under preferences problems. 
In this section, we demonstrate how this can be done for the \textsc{Stable Marriage} (\textsc{SM}) problem focusing on describing the adjustments necessary compared to the discussed \textsc{SR} setting. 
\paragraph{Stable Marriage Instances.}
A \textsc{Stable Marriage (SM)} instance $\mathcal{I}$ consists of a set~$A$ of agents partitioned into two sets $U$ and $W$, which are traditionally referred to as men and women, respectively.
We assume for simplicity that $|U|=|W|$.
Each man $u\in U$ has a preference order~$\succ_u \in \mathcal{L}(W)$ over all women and each woman $w\in W$ has a preference order $\succ_w \in \mathcal{L}(U)$ over all men. 
A matching is a subset $M$ of man-woman pairs where each agent appears in at most one pair and a matching is stable if no man-woman pair exists preferring each other to their assigned partner.
Note that a stable matching is guaranteed to exist in every \textsc{SM} instance. 

\paragraph{Mutual Attraction Distance between SM Instances.}
Let $\mathcal{I}=(U=\{u_1,\dots, u_n\}, W=\{w_1,\dots w_n\}, (\succ_u)_{u\in U}, (\succ_w)_{w\in W})$ be an SM instance. 
Recall that for some $i\in [2n-1]$ and $a\in U\cup W$, $\MA_{\mathcal{I}}(a,i)$ is the position of $a$ in the preference order of the agent~$a'$, which is ranked in position~$i$ by $a$, i.e., $\MA_{\mathcal{I}}(a,i)=\pos_{\succ_{a'}}(a)$ where~$a' := \ag_{\succ_a}(i)$.
The mutual attraction vector of an agent~$a\in U\cup W$ is $\MA_{\mathcal{I}}(a)=\big(\MA_{\mathcal{I}}(a,1),\dots, \MA_{\mathcal{I}}(a,n)\big)$.

For each instance $\mathcal{I}$, we define two mutual attraction matrices:  $\MA_{\mathcal{I},U}$ and $\MA_{\mathcal{I},W}$. 
$\MA_{\mathcal{I},U}$ is the mutual attraction matrix of men. 
The $i$-th row of $\MA_{\mathcal{I},U}$ is the vector $\MA_{\mathcal{I}}(u_i)$. 
$\MA_{\mathcal{I},W}$ is the mutual attraction matrix of women. 
The $i$-th row of $\MA_{\mathcal{I},W}$ is the vector  $\MA_{\mathcal{I}}(w_i)$.
Thus, notably, an SM instance does not correspond to a single mutual attraction matrix but a pair of mutual attraction matrices. 
Consequently, in the case of SM instances not a single mutual attraction matrix, but a pair of mutual attraction matrices $(A,B)$ is called \emph{realizable} if there is an SM instance $\mathcal{I}$ defined over a set of men $U$ and women $W$ with $A=\MA_{\mathcal{I},U}$ and $B=\MA_{\mathcal{I},W}$.

The mutual attraction distance between two SM instances $\mathcal{I}$ with agents $U\cupdot W$ and $\mathcal{I}'$ with agents $U'\cupdot W'$ with $|U|=|W|=|U'|=|W'|=n$ is defined as:
\begin{align*}
     & \min \Big(\retrodist(\MA_{\mathcal{I},U},\MA_{\mathcal{I}',U'})+\retrodist(\MA_{\mathcal{I},W},\MA_{\mathcal{I}',W'}),\\
    & \retrodist(\MA_{\mathcal{I},U},\MA_{\mathcal{I}',W'})+\retrodist(\MA_{\mathcal{I},W},\MA_{\mathcal{I}',U'})\Big)
\end{align*}
where for two $n\times n$ mutual attraction matrices $A$ and $B$ their mutual attraction distance $\retrodist(A,B)$ is still defined as: 
$\min_{\sigma\in \Pi([n],[n])} \sum_{i\in [n]} \ell_1\big(A_i, B_{\sigma(i)}\big)$.
In particular, we do not fix that the ``women'' in one instance should be matched to the ``women'' in the other instance (as in one-to-one applications the two sides are often in some sense exchangeable), but instead map the two sides to each other such that the resulting distance is minimized. 

\paragraph{Navigating the Space of SM Instances.}
Also for \textsc{SM} instances, it will prove useful to identify ``canonical" pairs of extreme mutual attraction matrices. 
The first three extreme matrices identified for \textsc{SR} instances are still clearly relevant here: 
Identity here corresponds to the situation where all women have the same preferences over the men and all men have the same preferences over the women. 
This results in the following pair of matrices:
$$\Big(\begin{bmatrix}
 1 & \dots & 1 \\
 2 & \dots & 2 \\
  & \vdots &  \\
n & \dots & n \\
\end{bmatrix}, 
\begin{bmatrix}
 1 & \dots & 1 \\
 2 & \dots & 2 \\
  & \vdots &  \\
n & \dots & n \\
\end{bmatrix}\Big)
$$

For mutual agreement, we still require that if agent $a$ ranks agent $b$ in position $i$, then $b$ also ranks $a$ in position $i$. 
This results in the following pair of matrices: 
$$\Big(\begin{bmatrix}
 1 & 2 & \dots & n-1 & n\\
 1 & 2 & \dots & n-1 & n\\
   &  & \vdots &  & \\
1 & 2 & \dots & n-1 & n\\
\end{bmatrix}, 
\begin{bmatrix}
 1 & 2 & \dots & n-1 & n\\
 1 & 2 & \dots & n-1 & n\\
   &  & \vdots &  & \\
1 & 2 & \dots & n-1 & n\\
\end{bmatrix}\Big)
$$

Notably, this pair of matrices is realizable. We can simply partition a complete bipartite graph with $n$ vertices on each side into $n$ perfect matchings $M_1,\dots, M_n$, where $M_i$ determines the position~$i$ in the agent's preferences. 
By Hall's theorem, such a partition into perfect matchings always exits. 

For mutual disagreement, we analogously require that if $a$ ranks agent $b$ in position $i$, then $b$ also ranks $a$ in position $n-i+1$. This results in the following pair of matrices: 
$$\Big(\begin{bmatrix}
 n & n-1 & \dots & 2 & 1\\
  n & n-1 & \dots & 2 & 1\\
   &  & \vdots &  & \\
 n & n-1 & \dots & 2 & 1\\
\end{bmatrix}, 
\begin{bmatrix}
  n & n-1 & \dots & 2 & 1\\
  n & n-1 & \dots & 2 & 1\\
   &  & \vdots &  & \\
 n & n-1 & \dots & 2 & 1\\
\end{bmatrix}\Big)
$$
One realization of this matrix pair is an \textsc{SM} instance where for $i\in [n]$ woman $w_i$ has preferences $m_{i+1}\succ m_{i+2} \succ \dots \succ m_{n} \succ m_1 \succ m_2 \succ \dots \succ m_{i}$ and man $m_i$ has preferences $w_{i}\succ w_{i+1} \succ \dots \succ w_{n} \succ w_1 \succ w_2 \succ \dots \succ w_{i-1}$. 

Our forth extreme matrix, which is the chaos matrix, has no naturally defined analogue for the \textsc{SM} setting which is why we omit it.\footnote{Note that intuitively taking one matrix from the mutual agreement pair and one matrix from the mutual disagreement pair could be a viable fourth extreme point. However, it is easy to see that the resulting (and similar) matrix pairs are not realizable.}
Determining the maximum distance between two realizable matrix pairs remains open (in our experiments, the mutual agreement and mutual disagreement pairs are furthest away).

\begin{figure}[t]
    \centering
    \includegraphics{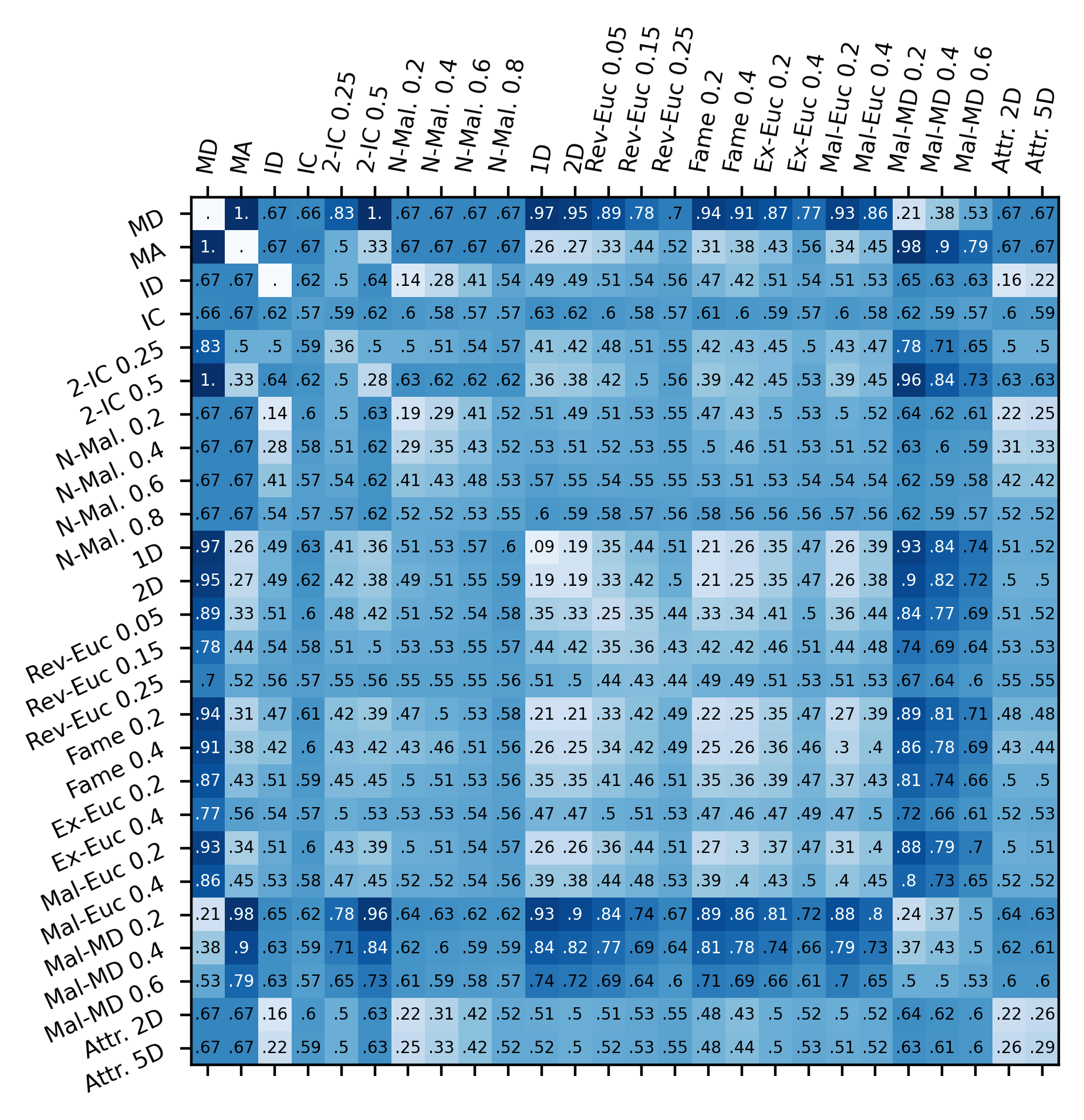}
    \caption{For each pair of statistical cultures for sampling \textsc{SM} instances, average mutual attraction distance between instances sampled from the two. The first three lines/columns contain, for each statistical culture, the average distance of instances sampled from the culture to our three extreme matrix pairs. The diagonal contains the average distance of two instances sampled from the same statistical culture.}
    \label{fig:distanceMatrix_marriage}
\end{figure}
\paragraph{Creating and Drawing the Map.}
To create our map of \textsc{SM} instances, we again sample $460$ instances from statistical cultures similar to the ones for \textsc{SR}. 
Notably, to maintain focus, we assume that the preferences of both women and men are generated using the same statistical culture (of course, in principle it would also be possible yet less clean to combine different cultures). 

We only describe how to adapt the cultures for \textsc{SR} instances to the bipartite SM setting and refer to \Cref{sub:creat} for the full descriptions. 
For the Impartial Culture and Mallows models, using the described procedure for SR instances, we sample for each woman $w\in W$ a preference order from $\mathcal{L}(U)$ and independently for each man $u\in U$ a preference order from $\mathcal{L}(W)$ using the described procedure. 
For the 2-IC model, given some $p\in [0,0.5]$, we partition $U$ into two sets $U_1\cupdot U_2$ with $|U_1|=\floor{p\cdot |U|}$ and we partition $W$ into two sets $W_1\cup W_2$ with $|W_1|=\floor{p\cdot |W|}$.  
    Each man~$u\in U$, respectively woman $w\in W$, samples one preference order $\succ$ from $\mathcal{L}(W_1)$, respectively $\mathcal{L}(U_1)$,  and one order $\succ'$ from $\mathcal{L}(W_2)$, respectively $\mathcal{L}(U_2)$. 
    If $a\in U_1\cup W_1$, then $a$'s preferences start with $\succ$ followed by $\succ'$. 
    If $a\in U_2\cup W_2$, then $a'$'s preferences start with $\succ'$ followed by $\succ$.
For the Euclidean, Mallows-Euclidean, Expectations-Euclidean, Fame-Euclidean, and Attributes models, we sample for each agent points and vectors as described for the respective model. 
Subsequently, a woman ranks all men according to their ``distance" (as described in the respective model) and a man ranks all women according to their distance.
For the Reverse-Euclidean model, given some $p\in [0,1]$, we partition $U$ into two sets $U_1\cupdot U_2$ with $|U_1|=\floor{p\cdot |U|}$ and $W$ into two sets $W_1\cup W_2$ with $|W_1|=\floor{p\cdot |W|}$.  
Subsequently, we sample the preferences as in the Euclidean model and subsequently reverse the preferences of all agents in $U_2\cup W_2$. 
Lastly, for Mallows-MD, we start with the \textsc{SM} instance realizing MD described previously, where for $i\in [n]$ woman $w_i$ has preferences $\succ_{w_i} : m_{i+1}\succ m_{i+2} \succ \dots \succ m_{n} \succ m_1 \succ m_2 \succ \dots \succ m_{i}$ and man $m_i$ has preferences $\succ_{m_i} :w_{i}\succ w_{i+1} \succ \dots \succ w_{n} \succ w_1 \succ w_2 \succ \dots \succ w_{i-1}$. 
As for SR, for each agent~$a\in U\cup W$, we obtain its final preferences by sampling a preference order from $\mathcal{D}_{\text{Mallows}}^{\succ_{a}, \normphi}$.

\begin{figure*}[t]
     \begin{subfigure}[b]{0.325\textwidth}
      \centering
    \includegraphics[trim={0.1cm 0.1cm 0.1cm 0.1cm}, clip,width=5.6cm]{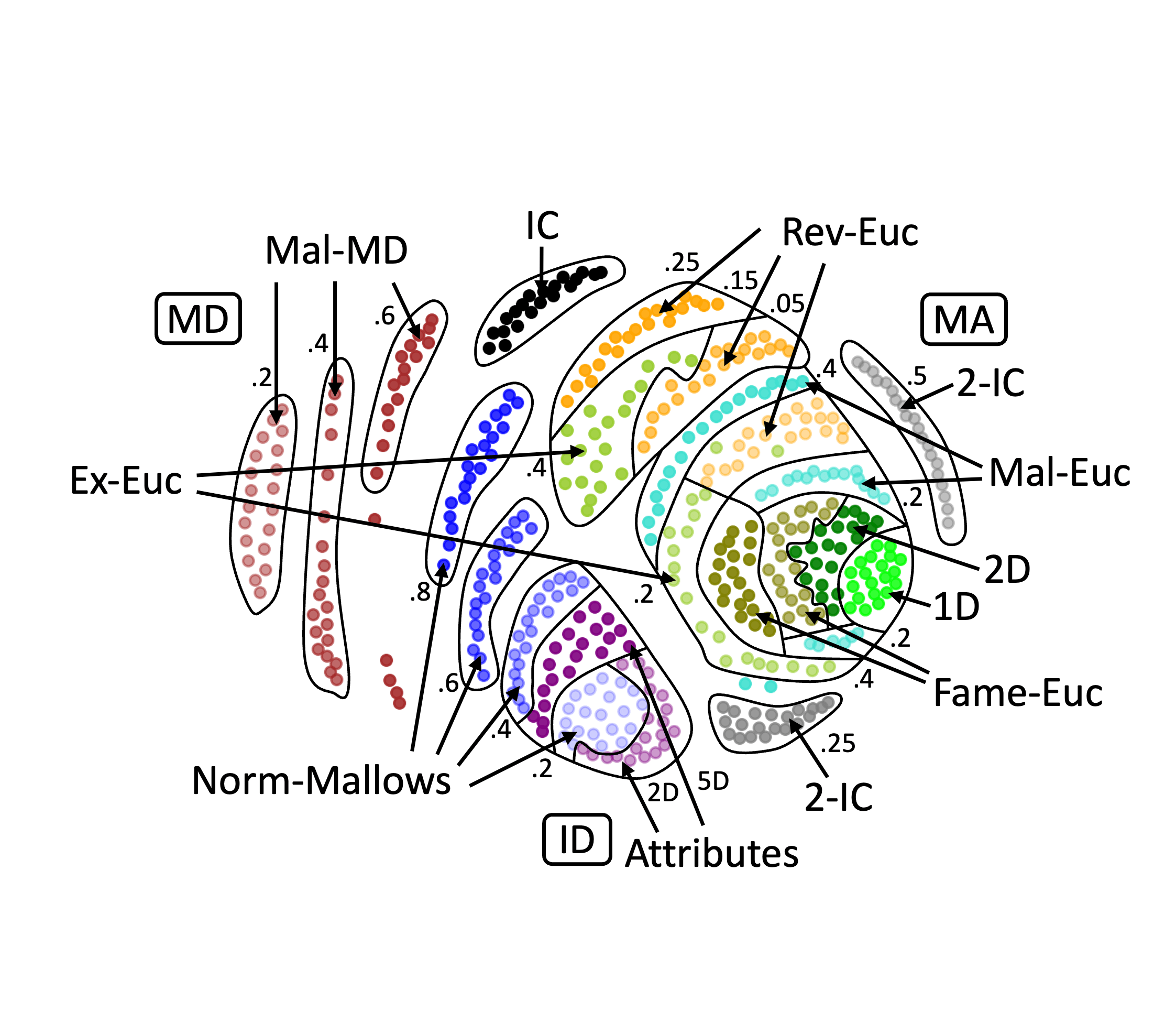}
    \caption{The color of a point indicates the statistical culture the respective instance was sampled from.} \label{map:SM}
     \end{subfigure}
      \begin{subfigure}[b]{0.325\textwidth}
      \centering
        \includegraphics[trim={0.1cm 0.1cm 0.1cm 0.1cm}, clip,width=5.6cm]{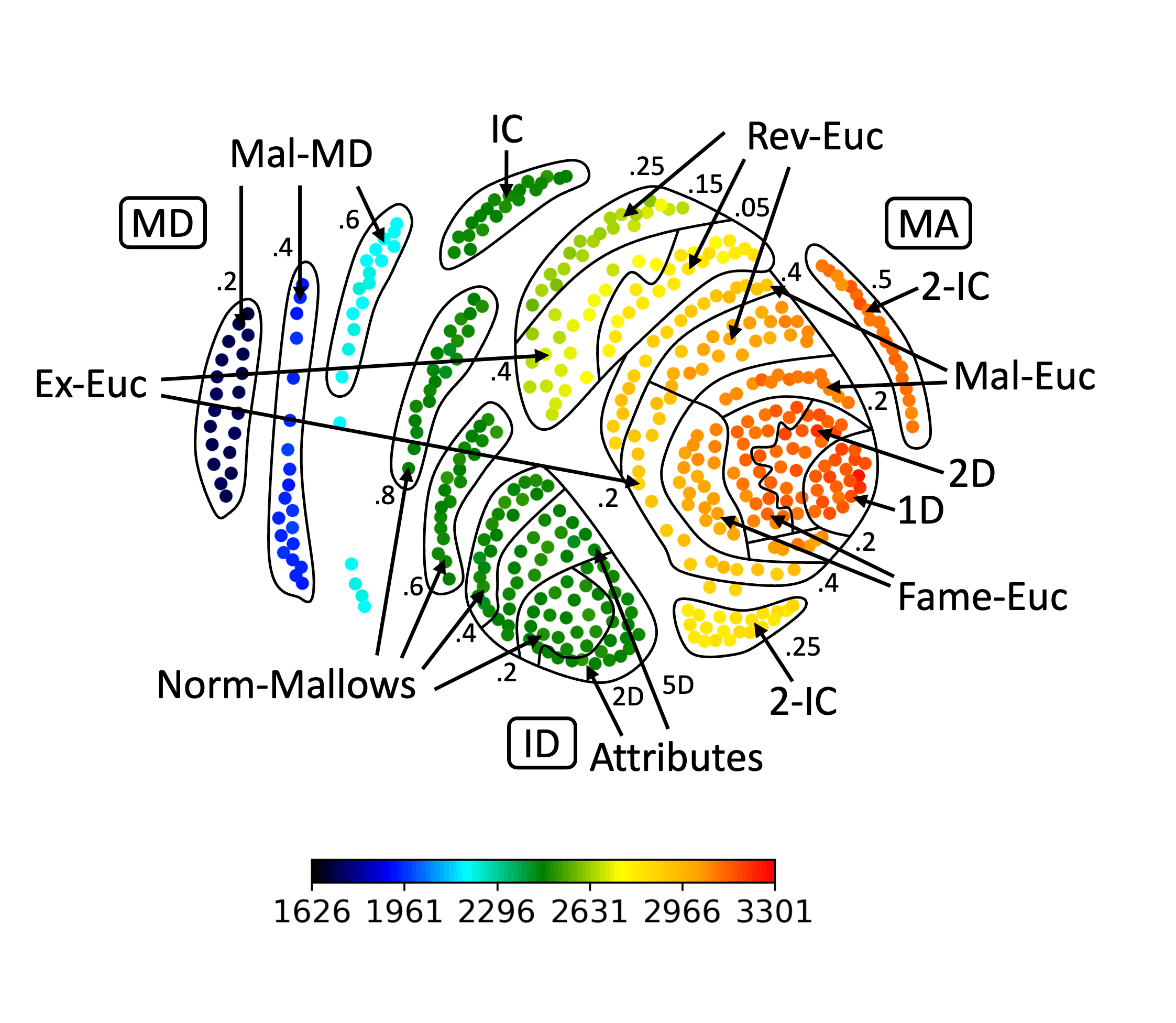}
         \caption{Average number of blocking pairs for a perfect matching.} \label{map:SM-averagebps}
     \end{subfigure}
     \begin{subfigure}[b]{0.325\textwidth}
         \centering
        \includegraphics[trim={0.1cm 0.1cm 0.1cm 0.1cm},clip,width=5.6cm]{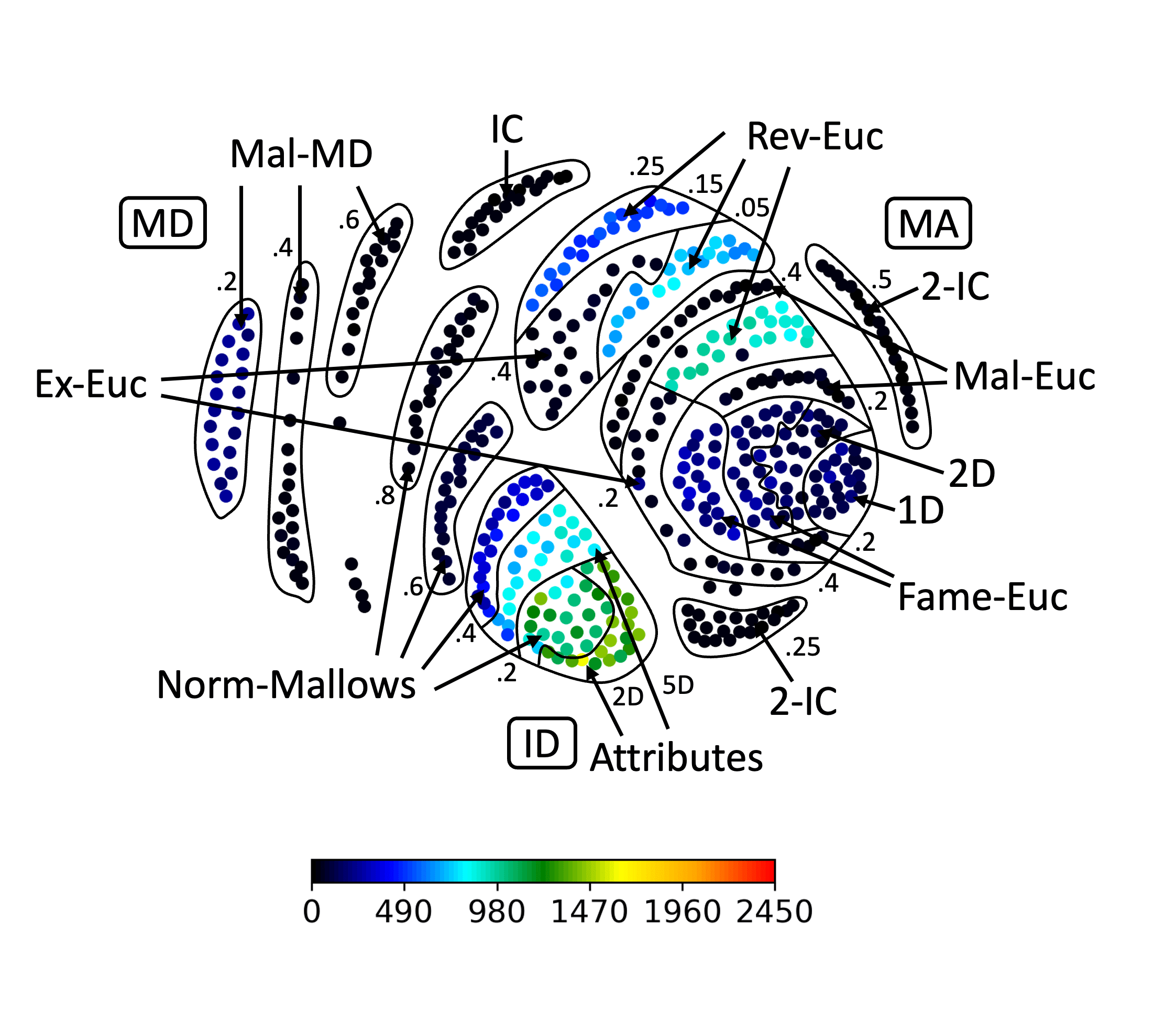}
  \caption{Number of blocking pairs for a minimum-weight matching. The scale is logarithmic.} \label{map:SM-minweightbps}
     \end{subfigure}
    \caption{Map of $460$ SM instances for $100$ men and $100$ women visualizing different quantities for each instance. Each instance is represented by a point. Roughly speaking, the closer two points are on the map, the more similar are the respective SR instances under the mutual attraction distance. }
\end{figure*}

As for \textsc{SR}, our dataset consists of $460$ \textsc{SM} instances sampled from the above described statistical cultures, where we use the same parameter configurations as for \textsc{SR} (as described in \Cref{sub:creat}).
In addition, on our maps, we include the three extreme matrix pairs described previously. 

As for \textsc{SR}, to draw the map, we first compute the mutual attraction distance of each pair of instances and subsequently embed them into the two-dimensional Euclidean space using a variant of the force-directed Kamada-Kawai algorithm \cite{DBLP:journals/ipl/KamadaK89}.
We depict the map visualizing our dataset of $460$ instances for $100$ men and $100$ women in \Cref{map:SM}.
Overall, the map of \textsc{SM} instances from \Cref{map:SM} is very similar to the map of \textsc{SR} instances from  \Cref{fig:mainMap}, where the only cultures that are placed slightly differently in the two maps are 2-IC and Expectations-Euclidean. 

Moreover, we depict in \Cref{fig:distanceMatrix_marriage} the average distance between the different statistical cultures. 
Unsurprisingly, the general picture in \Cref{fig:distanceMatrix_marriage} is very similar as in \Cref{fig:dis_matrix} for \textsc{SR}, with Expectations-Euclidean  being the culture that produces the largest differences between \textsc{SR} and \textsc{SM} (which is then also reflected on the map, as this culture is placed differently in the two maps). 

\begin{figure*}[t]
\begin{subfigure}[b]{0.325\textwidth}
         \centering
        \includegraphics[trim={0.1cm 0.1cm 0.1cm 0.1cm}, clip,width=5.6cm]{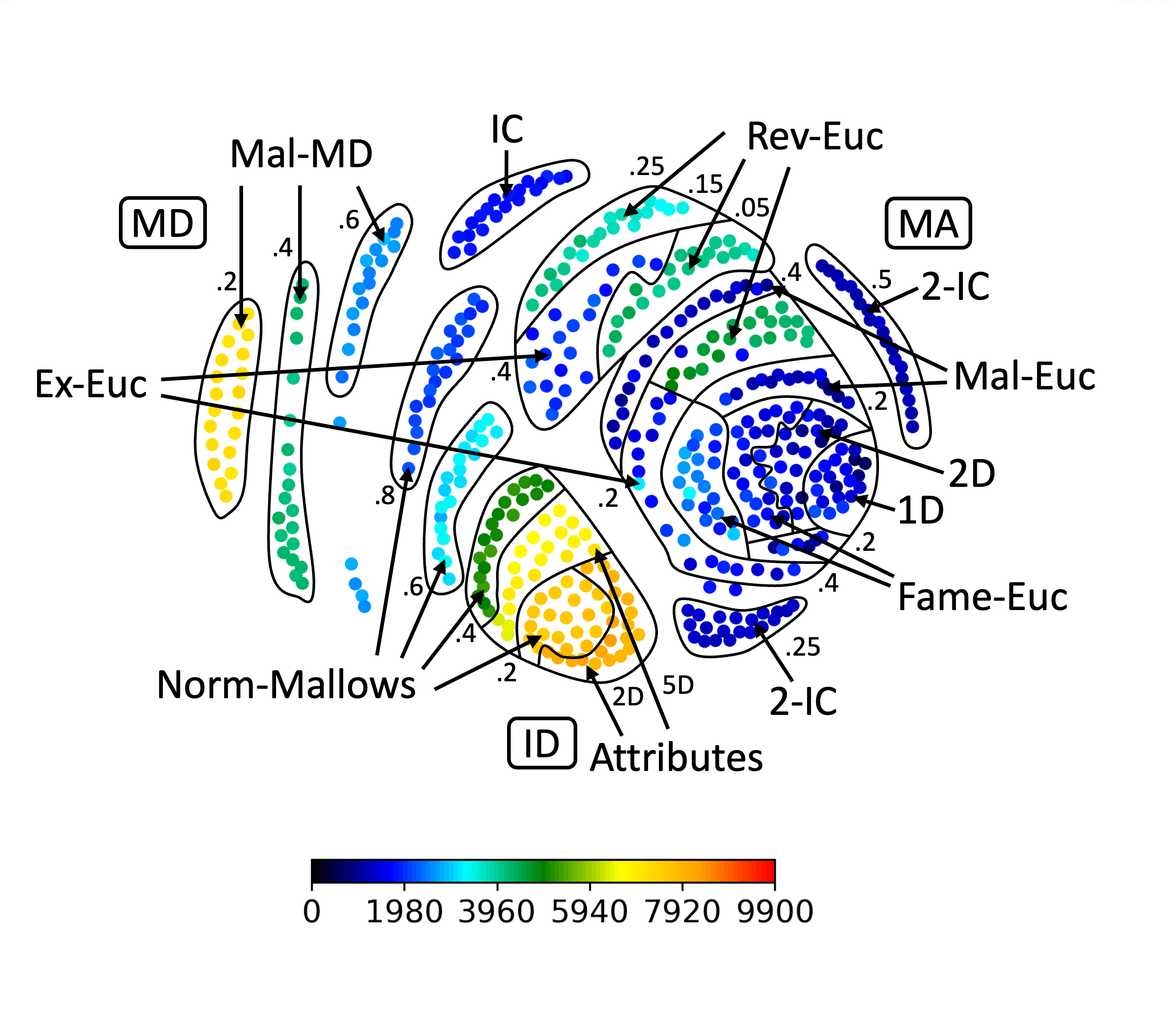}
         \caption{Minimum summed rank of agents for partner in any stable matching.} \label{map:SM-minimumSummed}
     \end{subfigure}
     \begin{subfigure}[b]{0.325\textwidth}
         \centering
        \includegraphics[trim={0.1cm 0.1cm 0.1cm 0.1cm}, clip,width=5.6cm]{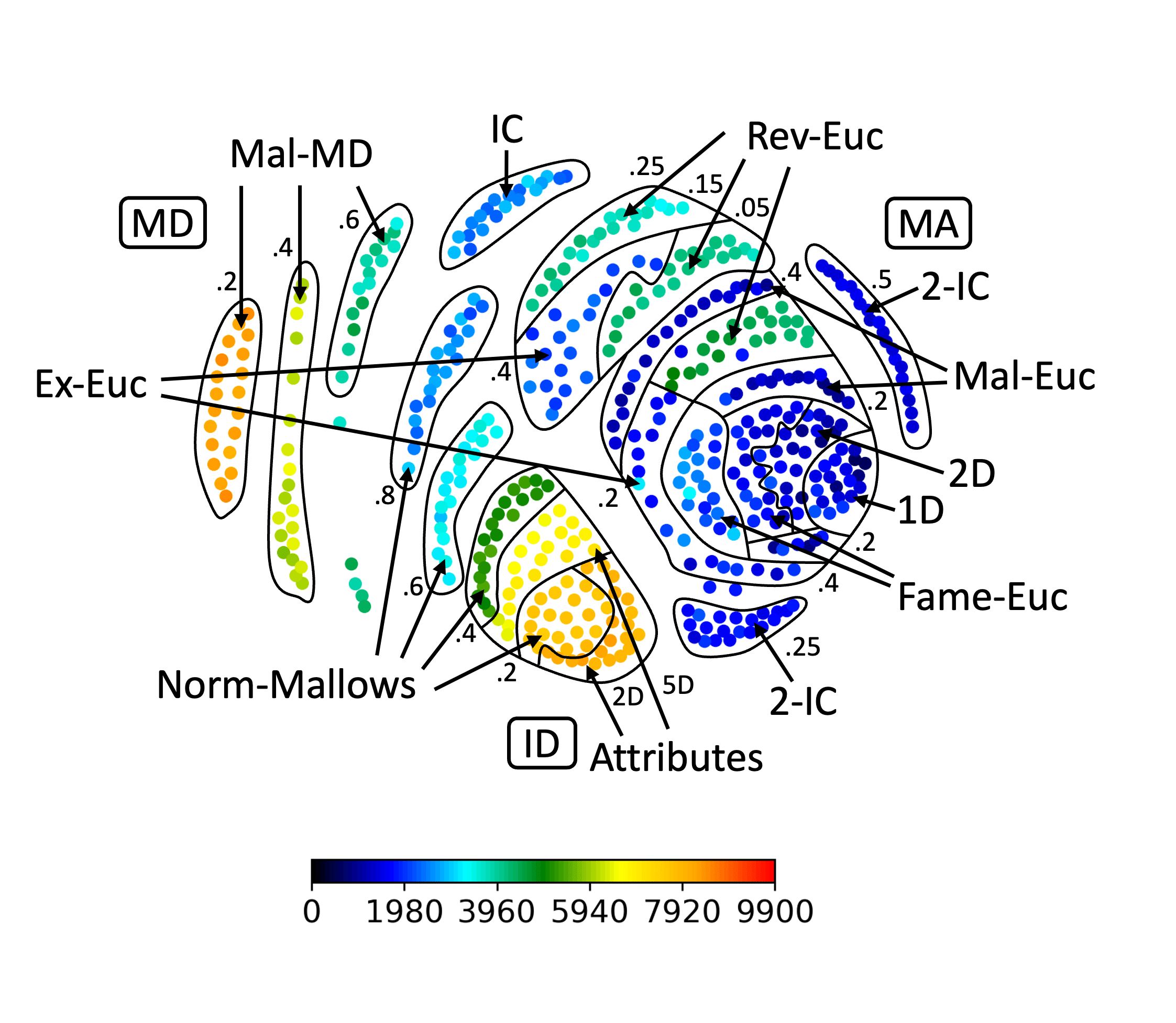}
    \caption{Maximum summed rank of agents for partner in any stable matching.} \label{map:SM-maximumSummed}
     \end{subfigure}
      \begin{subfigure}[b]{0.325\textwidth}
         \centering
        \includegraphics[trim={0.1cm 0.1cm 0.1cm 0.1cm}, clip,width=5.6cm]{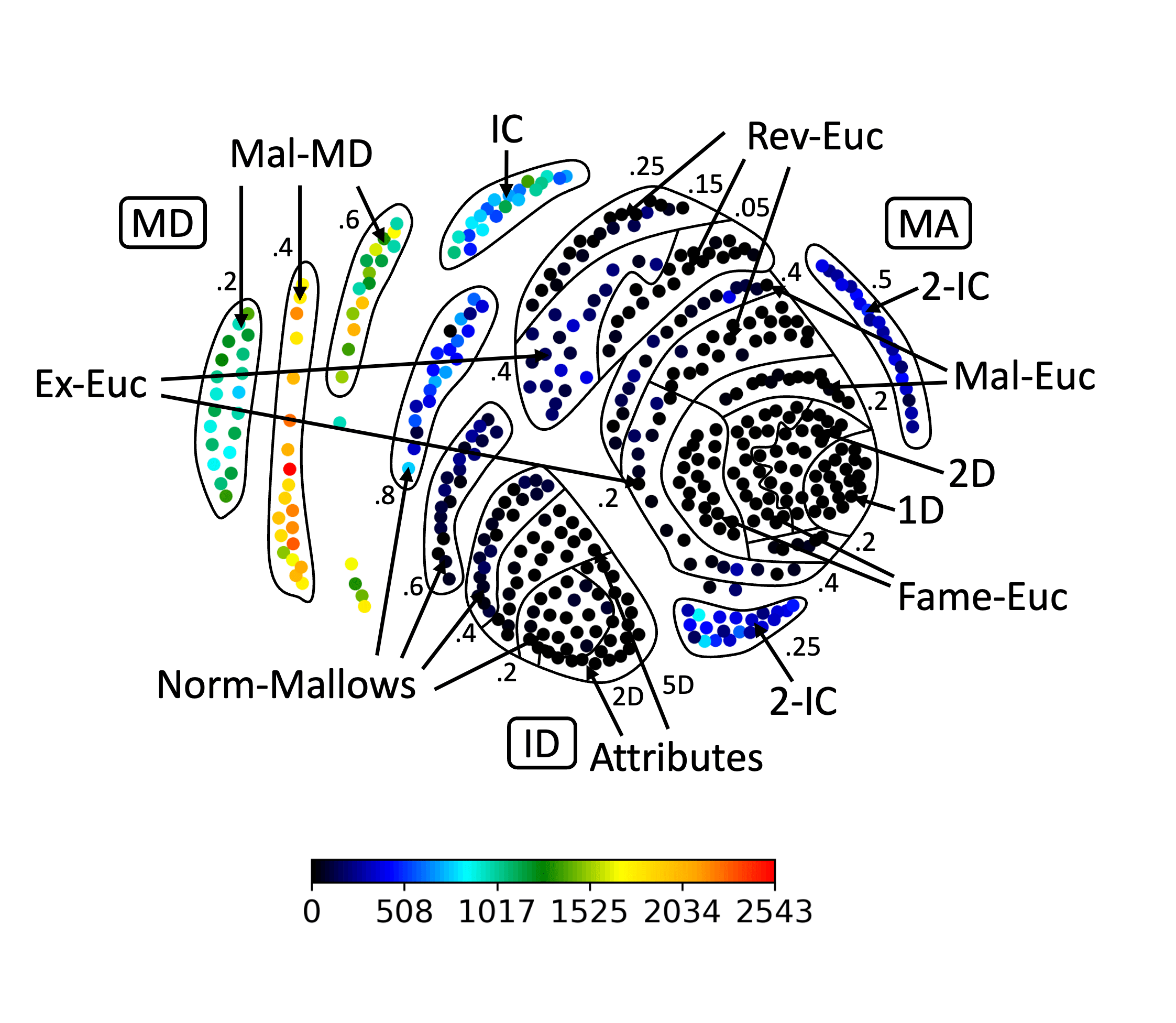}
    \caption{Difference between maximum and minimum summed rank of agents for partner in any stable matching.} \label{map:SM-maximumminimumRank}
     \end{subfigure}
    \caption{Map of $460$ SM instances for $100$ men and $100$ women visualizing different quantities for each instance.}
\end{figure*}

\begin{figure}
    \centering
    \includegraphics[trim={0.1cm 0.1cm 0.1cm 0.1cm}, clip,width=5.6cm]{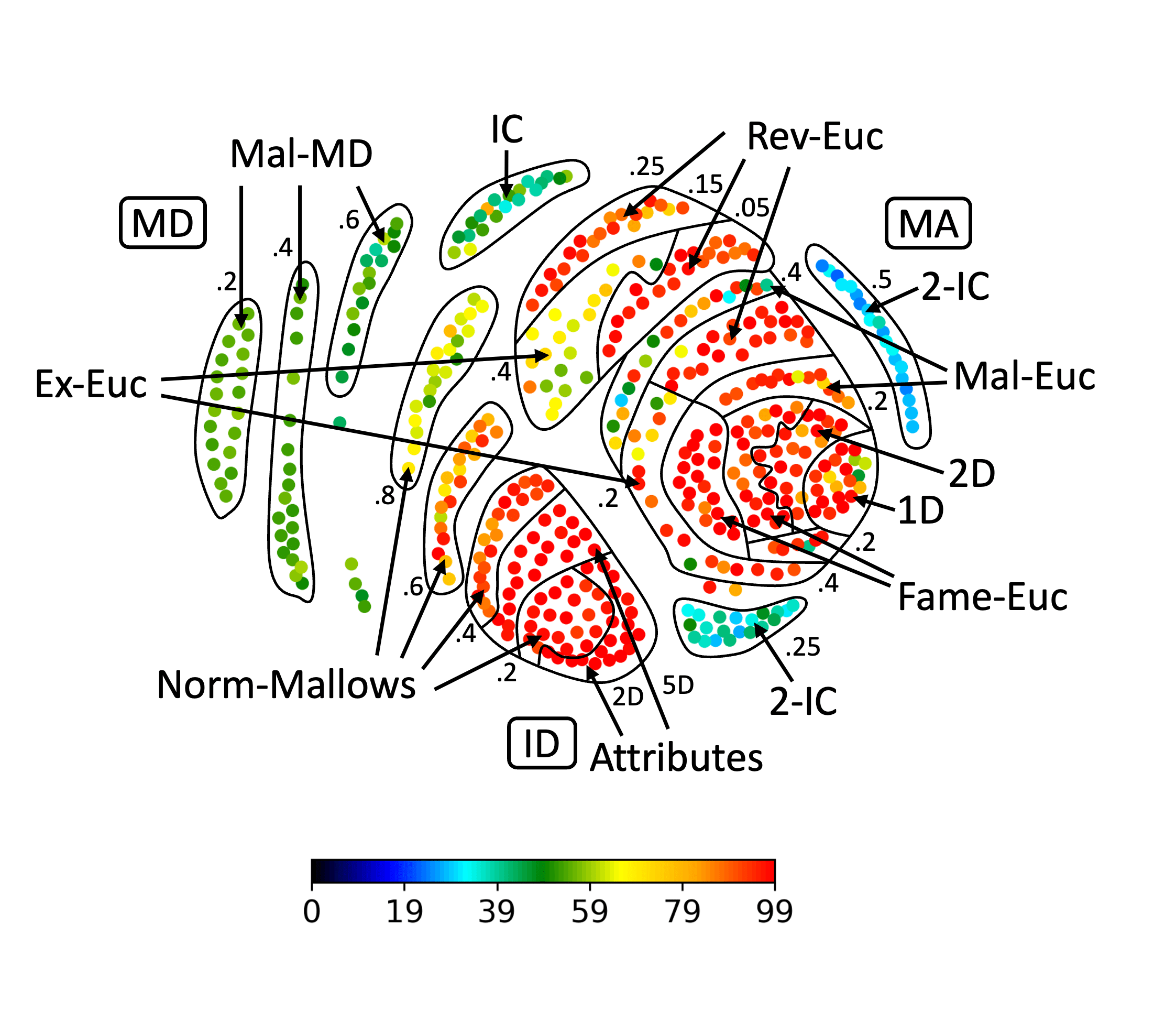}
    \caption{Map of $460$ SM instances for $100$ men and $100$ women visualizing the maximum rank an agent assigns to its partner in a stable matching minimizing this value.}
   \label{map:SM-maximumRank}
\end{figure}
\paragraph{Using the Map.}
To showcase possible use cases of our map of \textsc{SM} instances we repeated the experiments that we conducted for \textsc{SR} in \Cref{sec:using}. 

Specifically, analogous to \Cref{sub:expNumberBPs}, in \Cref{map:SM-averagebps}, we visualize the average number of blocking pairs for a perfect matching in our \textsc{SM} instances (by sampling $100$ perfect matchings and taking the average of the number of pairs blocking them). 
The results are very similar as for \textsc{SR} and in particular the average number of blocking pairs strongly correlates with the position of instances on the map. 

Moreover, analogous to \Cref{sub:min-weightBPs}, in \Cref{map:SM-minweightbps}, we show the number of blocking pairs for a matching minimizing the summed rank agents have for their partner. 
Remarkably, for a large majority of instances, this matching is quite close to being stable. 
The general picture here is again very similar as for \textsc{SR}; in particular, for both \textsc{SM} and \textsc{SR} the different statistical cultures produce instances with very similar properties.

Next, we consider the stable matching that minimizes/maximizes the summed rank that agents have for their partner (as in \Cref{sub:summed-rank}). 
We show the summed rank that agents have for their partner in \Cref{map:SM-minimumSummed} for the summed rank minimal matching and in \Cref{map:SM-maximumSummed} for the summed rank maximal matching. 
The general picture here looks again very similar as for \textsc{SR}. 
In particular, the instances sampled from one culture produce very similar results and for certain regions on the map instances falling in this region show a uniform behavior. 
Moreover, in \Cref{map:SM-maximumminimumRank}, we show the difference between the summed rank agents have for their partner in the stable matching maximizing and minimizing this value. 
Comparing this map to the respective map for \textsc{SR}, what stands out is that for \textsc{SM} for some instances there is a larger difference between the summed rank minimal and summed rank maximal matching than for \textsc{SR}, indicating that the space of stable matchings for some of the sampled \textsc{SM} instances is ``richer''.
Nevertheless, still for most of our \textsc{SM} instances there is only little difference between the summed rank minimal and maximal matching; this holds in particular for most of the instances from the bottom-right region of the map. 

Lastly, analogous to \Cref{sub:maxrank},  in \Cref{map:SM-maximumRank}, we depict the maximum rank an agent has for its partner in a stable matching minimizing this value. 
Notably, here, the results for SM differ from the results for SR. 
In particular, for SM, there are more instances where some agent is always matched to its almost least preferred agent than for SR (this contrast is most clear for 1D- and 2D-Euclidean instances).
Overall, ignoring 2-IC, in \Cref{map:SM-maximumRank}, a split of the map for SM instances is visible where in instances from the bottom right part some agent is matched to one of its least preferred agents in every stable matching, whereas in instances from the top left part of the map, in some stable matching even the worst off agent is matched to a partner that is not in the bottom $20\%$ of its preferences. 

\section{Conclusion}
The goal of this paper is to contribute to the toolbox for conducting experiments for stable matching problems. 
To achieve this, as a first step, we have introduced the polynomial-time computable mutual attraction distance and analyzed its properties as well as the space it induces.
As a second step, we have described a variety of statistical cultures for generating synthetic stable matching instances. 
The statistical cultures can be used to generate more diverse test data, where the diversity of test datasets may be assessed using our mutual attraction distance. 
One specific application of these two contributions is our map of synthetic stable matching instances, where we embed $460$ synthetically generated instances as points in the plane such that the euclidean distances between points roughly resembles the mutual attraction distance between the respective instances. 
We have verified that the produced map is meaningful in that it groups instances with similar properties together and have provided intuitive interpretations of the different regions on the map. 
Lastly, we have conducted various exemplary experiments, highlighting the value of the map to visualize and analyze experimental results. 
For instance, we were often able to identify certain regions on the map where instances share special properties. 
Overall, our experimental results underline the importance of using diverse test data and in particular to not only restrict oneself to uniformly at random sampled preferences, which produces instances covering only a small part on the map.

From a theoretical perspective, it would be interesting to extend our theoretical analysis of the space of SR instances to SM instances, and to analyze the theoretical properties of the space of SM and SR instances induced by other distance measures such as the swap or Spearman distance. 
To further verify the validity of the mutual attraction distance and the map, conducting further experiments and analyzing whether instances that are close on the map share similar properties is valuable
Lastly, it would be interesting to see where real-world instances lie on the map and to apply our framework to other types of stable matching problems. 

\subsubsection*{Acknowledgments}
NB was supported by the DFG project ComSoc-MPMS (NI 369/22). KH was supported by the DFG project
FPTinP (NI 369/16). This
project has received funding from the European Research Council (ERC) under the European Union’s
Horizon 2020 research and innovation programme (grant agreement No 101002854).
We are grateful to the anonymous \emph{MATCH-UP} 2022 reviewers for their thoughtful,
constructive, and helpful comments.

\begin{center}
	\includegraphics[width=3cm]{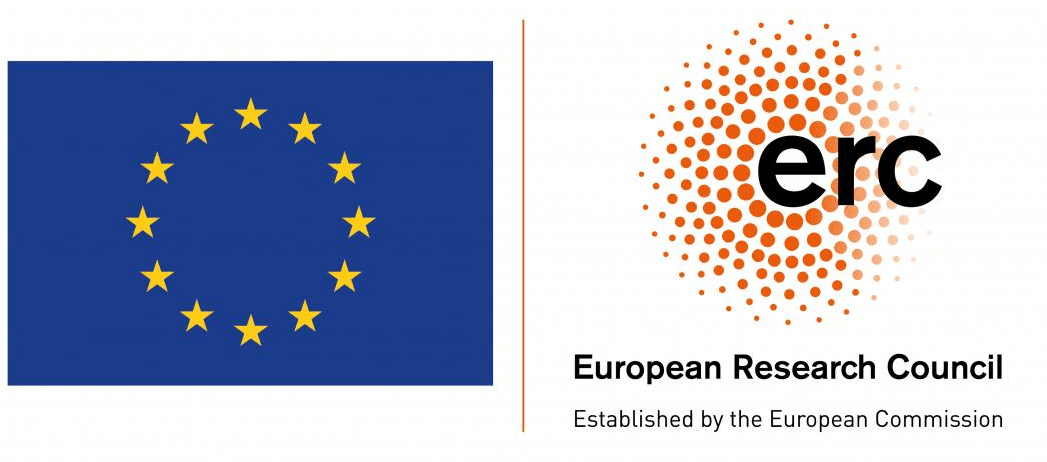}
\end{center}

\clearpage 
\clearpage
\bibliographystyle{plainnat}

\newpage 
\appendix
\section*{Appendix}

\section{Additional Material for \Cref{sec:understanding}}

\subsection{Identity}

Given $n \in \mathbb{N}$, the $2n \times (2n -1)$-matrix \IDn~can be written as follows:
\begin{equation*}
\begin{bmatrix}
 1 & 1 & 1 & 1 & \dots & 1 & 1 \\
 1 & 2 & 2 & 2 & \dots & 2 & 2 \\
 2 & 2 & 3 & 3 & \dots & 3 & 3 \\
 3 & 3 & 3 & 4 & \dots & 4 & 4 \\
 & & & \vdots & & &\\
 2n-2 & 2n-2 & 2n-2 & 2n-2 & \dots & 2n-2 & 2n-1 \\
 2n-1 & 2n-1 & 2n-1 & 2n-1 & \dots & 2n-1 & 2n-1 \\
\end{bmatrix}
\end{equation*}

\IDreal*
\begin{proof}
 $(\Rightarrow)$:
 Let $\mathcal{P}^*$ be a realization of \IDn.
 We call the agents from $\mathcal{P}^*$ by $a_1, a_2, \dots, a_{2n}$ and assume without loss of generality that the $i$-th row of \IDn\ belongs to~$a_i$ for every~$i \in [2n]$.
 We show by induction on $i$ that agent~$a_i$ is the $i$-th agent in the preferences of~$a_j$ for $j> i$ and the $(i-1)$-th agent in the preferences of~$a_j$ for $j < i$ (which is equivalent to $\mathcal{P}^*$ being derived from the master list $a_1 \succ a_2 \succ \dots \succ a_{2n}$).
 For~$i = 0$, there is nothing to show (as agent~$a_0$ does not exist).
 So fix~$i > 0$.
 The $i$-th row of~$\IDn$ contains the entry~``$i-1$" precisely $i-1$ times and the entry $i$ precisely~$n -i$ times.
 By the induction hypothesis, the preferences of~$a_j$ for $j > i$ start with~$a_1 \succ a_2 \succ \dots \succ a_{i-1}$, so $a_i$ cannot be the $(i-1)$-th agent of~$a_j$.
 Consequently, $a_i$ is the $(i-1)$-th agent in the preferences of~$a_j$ for every~$j < i$.
 It follows that $a_i$ is the $i$-th agent in the preferences of all remaining agents, that is, of agent~$a_j$ for every $j > i$.
 
 $(\Leftarrow)$:
 Let $\mathcal{P}$ be derived from the master list~$a_1 \succ a_2 \succ \dots \succ a_{2n}$.
 Then $a_i$ is the $(i-1)$-th agent in the preferences of~$a_j$ for $j< i$ and the $i$-th agent in the preferences of~$a_j$ for~$j> i$.
 It follows that the mutual attraction matrix~$R$ of~$\mathcal{P}$ fulfills that $R[i, j] = \begin{cases} i &  j \ge i\\
i - 1 & j < i
\end{cases}$, i.e., we have $R = \IDn$.
\end{proof}

\subsection{Mutual Agreement}
Given $n \in \mathbb{N}$, the $2n \times (2n -1)$-matrix \symn\ can be written as follows:
\begin{equation*}
\begin{bmatrix}
 1 & 2 & 3 & 4 & \dots & 2n- 2 & 2n-1 \\
 1 & 2 & 3 & 4 & \dots & 2n- 2 & 2n-1 \\
 1 & 2 & 3 & 4 & \dots & 2n- 2 & 2n-1 \\
 & & & & \vdots & &\\
 1 & 2 & 3 & 4 & \dots & 2n- 2 & 2n-1 \\
 1 & 2 & 3 & 4 & \dots & 2n- 2 & 2n-1 
\end{bmatrix}
\end{equation*}

\MAreal*
\begin{proof}
 We prove the statement by giving an injective function from the realizations of \symn\ to the set of Round-Robin tournaments as well as an injective function from Round-Robin tournaments to realizations of \symn.
 
 Formally, we interpret Round-Robin tournaments as colored 1-factorization of the complete graph~$K_{2n}$.
 \begin{definition}
 A \emph{colored 1-factorization} of the complete graph~$K_{2n}$ is a function $f:  E( K_{2n} ) \rightarrow [2n-1]$ such that $f^{-1} (j) $ is a perfect matching for every~$j \in [2n-1]$.
\end{definition}
Accordingly, in the following we speak of colored 1-factorization of~$K_{2n}$.
 
 We start by describing an injective function from the realizations of \symn\ to the set of colored 1-factorizations of~$K_{2n}$.
 Let $\mathcal{P}$ be a realization of \symn.
 We create a colored 1-factorization~$f$ of~$K_{2n}$ as follows (where we identify the vertices of~$K_{2n} $ with the agents of~$\mathcal{P}$).
 For every~$j \in [2n-1]$, we set $f (\{a, a'\}) = j$ if and only if $ v$ is the $j$-th agent in the preferences of~$w$.
 Since $\symn[i,j] =j$ for every~$j\in [2n-1]$, this is indeed a colored 1-factorization.
 The corresponding function is clearly injective.
 
 Next, we give an injective function from the set of colored 1-factorizations to the realizations of~\symn.
 Given a colored 1-factorization on $K_{2n} $, we create preferences for each vertex as follows:
 The $j$-th vertex in the preferences of some vertex~$v$ is the vertex~$w$ such that $f(\{v, w\}) = j$.
 Since we have a colored 1-factorization, these are well-defined and feasible preferences.
 Let $\mathcal{P}$ be the resulting preference profile.
 The function is clearly injective.
 It remains to show that the mutual attraction matrix~$R$ of~$\mathcal{P}$ is \symn.
 For every vertex~$v$ with vertex~$w$ at the $j$-th position in its preferences, we have $f(\{v, w\}) =j$, implying that also $w$ has $v$ in the $j$-th position in its preferences.
 Consequently, we have $R[i,j ] = j$ for every~$ i \in [2n]$ and $j \in [2n-1]$, i.e., $R = \symn$.
 
 We remark that in general, there may be multiple non-isomorphic 1-factorizations~\cite{MR1516644}
 and thus also multiple non-isomorphic realizations of~\symn.
\end{proof}

\subsection{Mutual Disagreement}
Given $n \in \mathbb{N}$, the $2n \times (2n -1)$-matrix \asymn\ is defined as follows:
\begin{equation*}
\begin{bmatrix}
 2n-1 & 2n - 2 & 2n - 3 & 2n- 4 & \dots & 2 & 1 \\
 2n-1 & 2n - 2 & 2n - 3 & 2n- 4 & \dots & 2 & 1 \\
 2n-1 & 2n - 2 & 2n - 3 & 2n- 4 & \dots & 2 & 1 \\
 & & \vdots & & & &\\
 2n-1 & 2n - 2 & 2n - 3 & 2n- 4 & \dots & 2 & 1 \\
 2n-1 & 2n - 2 & 2n - 3 & 2n- 4 & \dots & 2 & 1
\end{bmatrix}
\end{equation*}

    \MDreal*
\begin{proof}
We start by proving the first part.
 The following preferences on agents~$a_1, \dots, a_{2n}$ are a realization for \asymn:
 \[
  a_i : a_{i+1} \succ a_{i+2} \succ \dots \succ a_{2n} \succ a_1 \succ a_2 \succ \dots \succ a_{i - 1}\,.
 \]
 It remains to show that this is indeed a realization of~\asymn.
 So fix~$i \in [2n]$ and $j \in [2n-1]$.
 For the rest of the proof, all indices are taken modulo $2n$.
 The $j$-th position in the preferences of~$a_i$ is~$a_{i + j}$.
 In the preferences of~$a_{i+j}$, agent~$a_i$ is at the $(2n - j)$-th position.
 Consequently, the preference profile realizes \asymn.
 
 We now turn to the second part.   For $n = 3$, the following realization is not isomorphic to the one described above, but its mutual agreement matrix is \asymn:
 \begin{align*}
     a_1 & : a_2 \succ a_3 \succ a_4 \succ a_5 \succ a_6 \\
     a_2 & : a_4 \succ a_6 \succ a_5\succ a_3\succ a_1\\
     a_3&: a_5\succ a_2\succ a_6\succ a_1\succ a_4\\
     a_4 & : a_3\succ a_5\succ a_1\succ a_6\succ a_2\\
     a_5 & : a_6 \succ a_1\succ a_2\succ a_4\succ a_3\\
     a_6 &: a_1\succ a_4\succ a_3\succ a_2\succ a_5 \qedhere
 \end{align*}
\end{proof}

\subsection{Chaos}
Given $n \in \mathbb{N}$, the $2n \times (2n -1)$-matrix \unn\ can be written as follows:
\begin{equation*}
\begin{bmatrix}
 1 & 2 & 3 & 4 & \dots & 2n-2 & 2n- 1 \\
 1 & n+1 & 2 & n+2 & \dots & 2n-1 & n \\
 2 & n+2 & 3 & n+3 & \dots & 1 & n + 1 \\
 3 & n+3 & 4 & n+4 & \dots & 2 & n+2 \\
 4 & n+4 & 5 & n+5 & \dots & 3 & n+3 \\
 & & & & \vdots & &\\
 2n-2 & n - 1 & 2n-1 & n & \dots & 2n-3 & n-2 \\
 2n-1 & n & 1 & n+1 & \dots & 2n-2 & n-1
\end{bmatrix}
\end{equation*}

\CHreal*
\begin{proof}
We remark that for each $\ell \in [n]$ it holds that $\unn [i][2\ell -1] = i + \ell - 2 \mod 2n - 1$.
Further, for each $\ell \in [n-1]$ it holds that $\unn [i][2\ell] = i + \ell + n - 2\mod 2n - 1$.
First, we show that $n-1$, $n$, and $n+1 $ are coprime to~$2n-1$.
Any integer dividing $n-1$ and $2n-1$ also divides $2n-1 - 2\cdot (n-1) =1$.
Any integer dividing $n$ and $2n-1$ also divides $2\cdot n - (2n-1) =1$.
Any integer dividing $n+1$ and $2n-1$ also divides $2\cdot (n + 1) - (2n-1) =3$.
Since $3$ is prime and does not divide $2n-1$, it follows that $n+1$ and $2n-1$ are coprime.

Since $2n-1$ and $n$ are coprime, each row of \unn\ contains each number from~$[2n-1]$ exactly once.
Since the entries in the first column of rows 2 to~$2n$ are 1, 2, \dots, $2n-1$, it follows that each number from~$[2n-1]$ is contained exactly once in each row and once in each column of the submatrix arising from deleting the first row.
Consequently, for every $j \in [2n-1]$, there exists exactly one~$i \in \{2, 3, 4, \dots, 2n\}$ with $\unn[i, j] =j$.

Fix $i > 1$ and let~$j \in [2n-1]$.
Let $j^* := \unn[i][j]$ and assume that $j^* \neq j$.
We claim that there exists exactly one~$i^* \in \{2, 3, \dots, 2n\}$ such that $\unn[i^*][j^*] = j$ and that $i^* \neq i$.
Existence and uniqueness of~$i^*$ follows as $j$ is contained once in each column (so in particular also in column~$j^*$).
It remains to show that $i^* \neq i$.
So assume towards a contradiction that $i^* = i$.
We have $j^* = i + nj - n -1 \mod 2n-1$.
Since $j^* \neq j$, it follows that 
\begin{align}\label{ineq:unknown}
 i + (n -1 ) \cdot j - n -1 \neq 0 \mod 2n-1.
\end{align}
Since $\unn[i][j^*] = j$, we have $i + n (i + n j - n -1 ) - n-1 = j \mod 2n-1$ which is equivalent to
\begin{align}\label{eq:unknown}
    (n + 1) \cdot (i + (n-1)\cdot j -n - 1) = 0 \mod 2n-1
\end{align}
Since $n + 1 $ and $2n-1$ are coprime, \Cref{eq:unknown} implies that $i + (n-1)\cdot j-n-1 =0 \mod 2n-1$, contradicting \Cref{ineq:unknown}.

In the realization of \unn, agent~$a_i$ ranks at position~$j$ the agent~$i^*$ with~$\unn[i^*] [j^*] = j$.
The resulting preference profile clearly realizes \unn, so it remains to show that it is indeed a preference profile, i.e., each agent appears exactly once in the preferences of some other agent.
It suffices to show that no agent~$a_{i^*}$ appears twice in the preferences of an agent~$a_i$.
So assume towards a contradiction that $a_{i^*}$ appears at positions~$j_1$ and $j_2$ (where $j_1 \neq j_2$) in the preference of~$a_i$.
Thus, we have $j_1 = i^* + ni + n^2 j_1 - n^2 -2n -1 \mod 2n-1$ and $j_2 = i^* + ni + n^2 j_2 - n^2 -2n -1 \mod 2n-1$.
Consequently, we have $(n^2 -1) \cdot j_1 = i^* + ni - n^2 -2n -1 = (n^2 -1) j_2 \mod 2n-1$.
It follows that $(n-1 ) (n+ 1) \cdot (j_1 - j_2) = 0 \mod 2n-1$.
Since $n-1$ and $2n-1$ as well as $n+1$ and $2n-1$ are coprime, it follows that $j_1 = j_2 \mod 2n-1$, a contradiction.

The uniqueness of the realization is obvious.
\end{proof}

\subsection{Proof of \Cref{pr:calc}}
This section is devoted to proving the following statement:
\disss*

\begin{lemma}\label{lem:dist-sym-asym}
	$\retrodist (\symn, \asymn ) = 4\cdot (n-1)\cdot n^2$.
\end{lemma}

\begin{proof}
	Independent of the mapping of agents, we get the following distance per row:
	\begin{align*}
	\sum_{j=1}^{2n-1} | j - (2n-j)| & = \sum_{j = 1}^{2n-1} |2j - 2n|= \sum_{j=1}^{n-1} (2n - 2j) + \sum_{j = n + 1}^{2n-1} (2j -2n) \\
	& = 2n \cdot (n-1) -  2 \sum_{j =1}^{n-1} j + 2 \sum_{j= 1}^{n-1} j  = 2 n \cdot (n-1)\,.
	\end{align*}
	Summing over all $2n$~rows proves the lemma.
\end{proof}

Note that \Cref{lem:max-dist,lem:dist-sym-asym} imply that the distance between \symn\ and \asymn\ is the maximum possible distance between any two realizable matrices.

\begin{lemma}\label{lem:dist-id-sym}
	$\retrodist (\IDn , \symn) = \frac{8}{3} n^3 - 4n^2 + \frac{4}{3} n$.
\end{lemma}

\begin{proof}
	For the $i$-th row of~$\IDn$, the distance to any row from $\symn$ is
	\begin{align*}
	\sum_{j=1}^{i-1} ((i-1) - j) + \sum_{j=i}^{2n-1} (j -i) & = (i-1)^2 - \sum_{j=1}^{i-1} j + \sum_{j=1}^{2n-1 - i} j\\
	& = (i-1)^2 - 0.5 \cdot (i-1) \cdot i + 0.5 (2n-1-i) \cdot (2n -i)
	\end{align*}
    This is a polynomial of second degree in $i$ and~$n$.
    Thus, summing over all rows (i.e., from $i=1$ to $2n$) yields a polynomial of third degree in~$n$.
    This is uniquely determined by any four points, e.g., by~$\retrodist (\ID^{2\cdot 1}, \MA^{2\cdot 1}) = 0$, $\retrodist (\ID^{2\cdot 2}, \MA^{2\cdot 2}) =8$, $\retrodist (\ID^{2\cdot 3}, \MA^{2\cdot 3}) = 32$, and $\retrodist (\ID^{2\cdot 4}, \MA^{2\cdot 4}) =112$.
    Consequently, we have $\retrodist (\IDn , \symn) = \frac{8}{3} n^3 - 4n^2 + \frac{4}{3} n$.
\end{proof}

\begin{lemma}
	$\retrodist (\IDn , \asymn) =\frac{8}{3} n^3 - 2n^2 - \frac{2}{3} n$.
\end{lemma}

\begin{proof}
	For any~$i \le n$, the distance of the $i$-th row of~$\IDn$ to any row from $\symn$ is
	\begin{align*}
	\sum_{j=1}^{i-1} & (2n -j - (i-1)) + \sum_{j=i}^{2n-i} (2n -j -i) + \sum_{j = 2n-i+1}^{2n-1} i - (2n-j) \\
	\end{align*}
	This is a polynomial of second degree in~$n$ and~$i$.

	For any $i > n$, the distance of the $i$-th row of~$\IDn$ to any row from $\symn$ is
	\begin{align*}
	\sum_{j=1}^{2n -i} & (2n -j - (i-1)) + \sum_{j=2n - i+ 1}^{i-1} (i-1 - (2n-j)) + \sum_{j = i}^{2n-1} i - (2n-j) \\
	& = \sum_{j=1}^{2n -i} (2n -j - (i-1)) +  \sum_{j = 2n -i + 1}^{2n-1} i - (2n-j) - (i - (2n-i + 1))
	\end{align*}
	Again, this is a polynomial of second degree in~$n$ and~$i$.
	
	Consequently, summing up the distance over all rows (i.e., summing over~$i$ from 1 to~$2n$) yields a polynomial of third degree in~$n$.
	A polynomial of third degree is uniquely characterized by any four points on the polynomial.
	Using $\retrodist (\ID^1, \asym^1) = 0$, $\retrodist (\ID^2, \asym^2) = 12$, $\retrodist (\ID^{3}, \asym^{3}) = 52$, and $\retrodist (\ID^{4}, \asym^{4}) = 136$, we get that 
	$\retrodist (\ID, \asymn) =  \frac{8}{3} n^3 - 2  n^2 - \frac{2}{3} n$.
\end{proof}

\begin{lemma}
	$\retrodist (\symn, \unn) =  \frac{8}{3} n^3 - 4  n^2 +\frac{4}{3} n$.
\end{lemma}

\begin{proof}
	As every row from \symn\ is identical, the mapping between the rows is irrelevant.
	The first row of \unn\ is identical to any row of \symn\ and thus does not contribute to the mutual attraction distance.
	Fix some~$i > 1$.
	Let $j^*_{\text{odd}} := \min \{n, 2n + 1 - i\}$ and $j^*_{\text{even}} := \min \{ n-1, n + 1 - i\}$.
	For~${i > 1}$, the $i$-th row of \unn\ contributes
	\begin{align*}
	\sum_{j = 1} ^{2n -1} & |j - \unn[i][j]| = \sum_{\ell = 1}^{n} |\unn[i][2\ell -1] - (2\ell -1)| +  \sum_{\ell = 1}^{n-1} |\unn[i][2\ell] - 2\ell|\\
	& = \sum_{\ell =1}^{j^*_{\text{odd}}} | i + \ell - 2 - 2\ell + 1| + \sum_{\ell =j^*_{\text{odd}}}^n |i + \ell -2 - (2n-1) - 2\ell + 1| \\
	& + \sum_{\ell = 1}^{j^*_{\text{even}}} | i + \ell + n -2 - 2\ell| + \sum_{\ell = j^*_{\text{even}}} | i + \ell + n -2 -(2n-1) -2 \ell|\\
	& = \sum_{\ell =1}^{j^*_{\text{odd}}} | i - \ell -1| + \sum_{\ell =j^*_{\text{odd}}}^n |i - \ell  - 2n + 1| + \sum_{\ell = 1}^{j^*_{\text{even}}} | i - \ell + n -2 | + \sum_{\ell = j^*_{\text{even}}} | i - \ell - n -1|\\
	\end{align*}
	One easily verifies that this is a polynomial of second degree in $n$ and $i$.
	Consequently, summing up the distance over all rows (i.e., summing over~$i$ from 1 to~$2n$) yields a polynomial of third degree in~$n$.
	A polynomial of third degree is uniquely characterized by any four points on the polynomial.
	Using $\retrodist (\chaos^{2\cdot 1}, \sym^{2\cdot 1}) = 0$, $\retrodist (\chaos^{2\cdot 3}, \sym^{2\cdot 3}) = 40$, $\retrodist (\chaos^{{2\cdot 4}}, \sym^{{2\cdot }4}) = 112$, and $\retrodist (\chaos^{{2\cdot 6}}, \sym^{{2\cdot 6}}) = 448$, we get that 
	$\retrodist (\symn, \unn) =  \frac{8}{3} n^3 - 4  n^2 +\frac{4}{3} n$.
\end{proof}

\begin{lemma}
	$\retrodist (\asymn, \unn) =  \frac{8}{3} n^3 - 2 n^2 - \frac{2}{3}n$.
\end{lemma}

\begin{proof}
	As every row from \asymn\ is identical, the mapping between the rows is irrelevant.
	The first row of \unn\ is identical to any row of \asymn\ and thus does not contribute to the distance.
	Fix some~$i > 1$.
	Let $j^*_{\text{odd}} := \min \{n, 2n + 1 - i\}$ and $j^*_{\text{even}} := \min \{ n-1, n + 1 - i\}$.
	For $i > 1$, the $i$-th row of~\unn\ contributes
	\begin{align*}
	\sum_{j = 1} ^{2n -1} & |2n-j - \unn[i][j]| = \sum_{\ell = 1}^{n} |\unn[i][2\ell -1] - (2n - (2\ell -1))| +  \sum_{\ell = 1}^{n-1} |\unn[i][2\ell] - (2n - 2\ell)|\\
	& = \sum_{\ell =1}^{j^*_{\text{odd}}} | i + \ell - 2 - 2n + 2\ell - 1| + \sum_{\ell =j^*_{\text{odd}}}^n |i + \ell -2 - (2n-1) - 2n + 2\ell - 1| \\
	& + \sum_{\ell = 1}^{j^*_{\text{even}}} | i + \ell + n -2 - 2n + 2\ell| + \sum_{\ell = j^*_{\text{even}}} | i + \ell + n -2 -(2n-1) - 2n + 2 \ell|\\
	& = \sum_{\ell =1}^{j^*_{\text{odd}}} | i + 3 \ell -3 - 2n| + \sum_{\ell =j^*_{\text{odd}}}^n |i + 3 \ell  - 4n -2| + \sum_{\ell = 1}^{j^*_{\text{even}}} | i +3 \ell - n -2 | + \sum_{\ell = j^*_{\text{even}}} | i + 3 \ell - 3n -1|\\
	\end{align*}
	One easily verifies that this is a polynomial of second degree in $n$ and $i$.
	Consequently, summing up the distance over all rows yields a polynomial of third degree in~$n$.
	A polynomial of third degree is uniquely characterized by any four points on the polynomial.
	Using $\retrodist (\chaos^{2\cdot 1}, \asym^{2\cdot 1}) = 0$, $\retrodist (\chaos^{2\cdot 3}, \asym^{2\cdot 3}) = 52$, $\retrodist (\chaos^{{2\cdot 4}}, \asym^{{2\cdot }4}) = 136$, and $\retrodist (\chaos^{{2\cdot 6}}, \asym^{2\cdot {6}}) = 500$, we get that $\retrodist (\unn, \asymn) = \frac{8}{3} n^3 - 2 n^2 - \frac{2}{3}n$.
\end{proof}

\begin{lemma}
$\retrodist (\IDn, \unn) = \frac{8}{3} n^3 \pm O(n^2)$
\end{lemma}
\begin{proof}
	We only proof the leading term.
	We have $\retrodist (\IDn, \ID^*) = O(n^2)$.
	Since every row of $\ID^*$ contains only one number and each row of $\symn$ as well as $\unn$ contain each number from~$[2n-1]$ exactly once, it follows that $\retrodist (\ID^*, \symn) = \retrodist (\ID^*, \unn)$.
	Thus, we have
	\begin{align*}
	\retrodist (\IDn, \unn) & = \retrodist (\ID^*, \unn) \pm O(n^2)\\
	& = \retrodist (\ID^*, \symn) \pm O(n^2)\\
	& = \retrodist (\IDn, \symn) \pm O(n^2)\\
	& = \frac{8}{3} n^3 \pm O(n^2) \qedhere
	\end{align*}
\end{proof}
	We conjecture that $\retrodist (\IDn, \unn) = \frac{8}{3} n^3 -3n^2 -\frac{5}{3} n + 2$.
	
\section{Additional Material for \Cref{sec:map}}

\paragraph{Quality of the Embedding.}
We now want to analyze whether the two-dimensional visualization of our dataset as a map adequately reflects the mutual attraction distances between instances. 
We consider two different quality measures for the embedding. 
First we compute for each pair of instances $(\mathcal{I},\mathcal{I}')$ its distortion
which is defined as the maximum of (a) the normalized mutual attraction distance between $\mathcal{I}$ and $\mathcal{I}'$ divided by the normalized Euclidean distance between the points representing $\mathcal{I}$ and $\mathcal{I}'$ on the map and (b) the normalized Euclidean distance between the points representing $\mathcal{I}$ and $\mathcal{I}'$ on the map divided by the normalized mutual attraction distance between $\mathcal{I}$ and $\mathcal{I}'$.
The average distortion is $1.8$, indicating that distances between instances are certainly not represented perfectly on the map. 
Nevertheless, this also underlines that the map creates a roughly correct picture of the space of SR instances (distances on the map are typically only ``off'' by a factor of two).
However, we want to remark here that some error is to be expected because our space of SR instances is naturally too complex to be perfectly embedded into two-dimensional space. 
In \Cref{fig:dist}, we analyze which of the instances on the map are particularly challenging to embed and are thus misplaced. 
We do so by coloring each point on the map according to the average distortion of all pairs involving this instance.
What we see here is that the instances that fall into the middle of the map are particularly problematic and that instances close to $\sym$ and $\asym$ are embedded nearly perfectly.

Moreover, as a slightly simpler measure we also consider for each pair of instances their normalized Euclidean distance on the map divided by their mutual attraction distance. 
We visualize the results as a histogram in \Cref{fig:hist}. 
What we see here is that instances are mostly placed ``too close to each other'' and that for a majority of instances the normalized Euclidean distance on the map is more than half of their mutual attraction distance. 

\begin{figure*}
\centering
\begin{minipage}{.48\textwidth}
   \centering
    \includegraphics[trim={0.1cm 0.1cm 0.1cm 0.1cm}, clip,width=8cm]{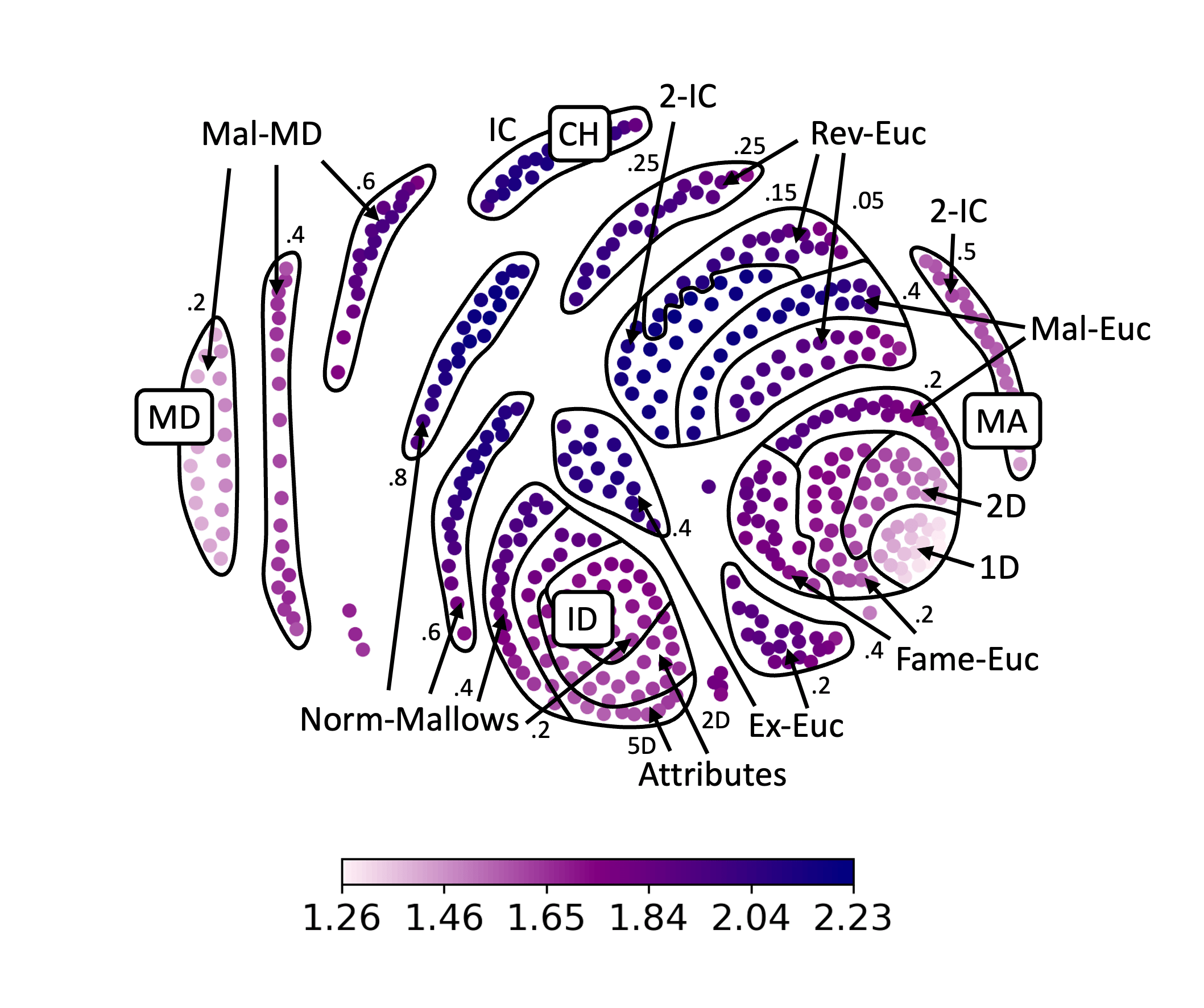}
    \caption{Average distortion for each instance on our map of SR instances for $200$ agents.}
  \label{fig:dist}
\end{minipage}\hfill
\begin{minipage}{.48\textwidth}
  \centering
  \includegraphics[width=8cm]{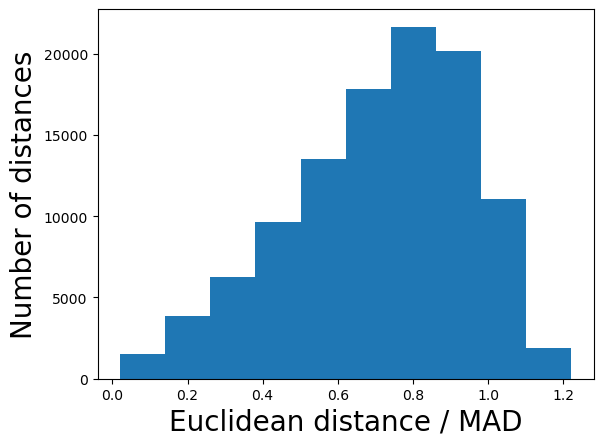}
  \caption{Histogram visualizing for instance pairs on the map their normalized Euclidean distance divided by the mutual attraction distance of the instances.}
  \label{fig:hist}
\end{minipage}
\end{figure*}

\paragraph{Map for Different Numbers of Agents.}
In \Cref{sub:creat}, we created a diverse test dataset for $200$ agents and visualized it as a map. In addition to this, we also created similar datasets for $102$\footnote{The reason  we consider $102$ instead of $100$ agents is because for $100$ agents the chaos matrix is not realizable.} and $50$ agents. 
The composition of the dataset is as in \Cref{sub:creat} so the same statistical cultures with the same parameters are used. 
Overall, all maps are very similar to each other. 
The only substantial difference between the map for $102$ agents (\Cref{fig:map_102}) and $200$ agents (\Cref{fig:mainMap}) is that instances sampled from  Mallows-Euclidean with $\normphi=0.2$ can be found both below and above the Euclidean instances in the map for $102$ agents. 
The map for $50$ agents (\Cref{fig:map_50}) differs a bit more from the map for $200$ agents in that the ``islands'' for the different cultures are a bit more scattered for $50$ agents. 
Moreover, instances sampled from $2$-IC with $p=0.25$ are placed on the bottom of the map.

\begin{figure*}
    \centering
\begin{subfigure}[b]{0.49\textwidth}
         \centering
         \includegraphics[width=7cm]{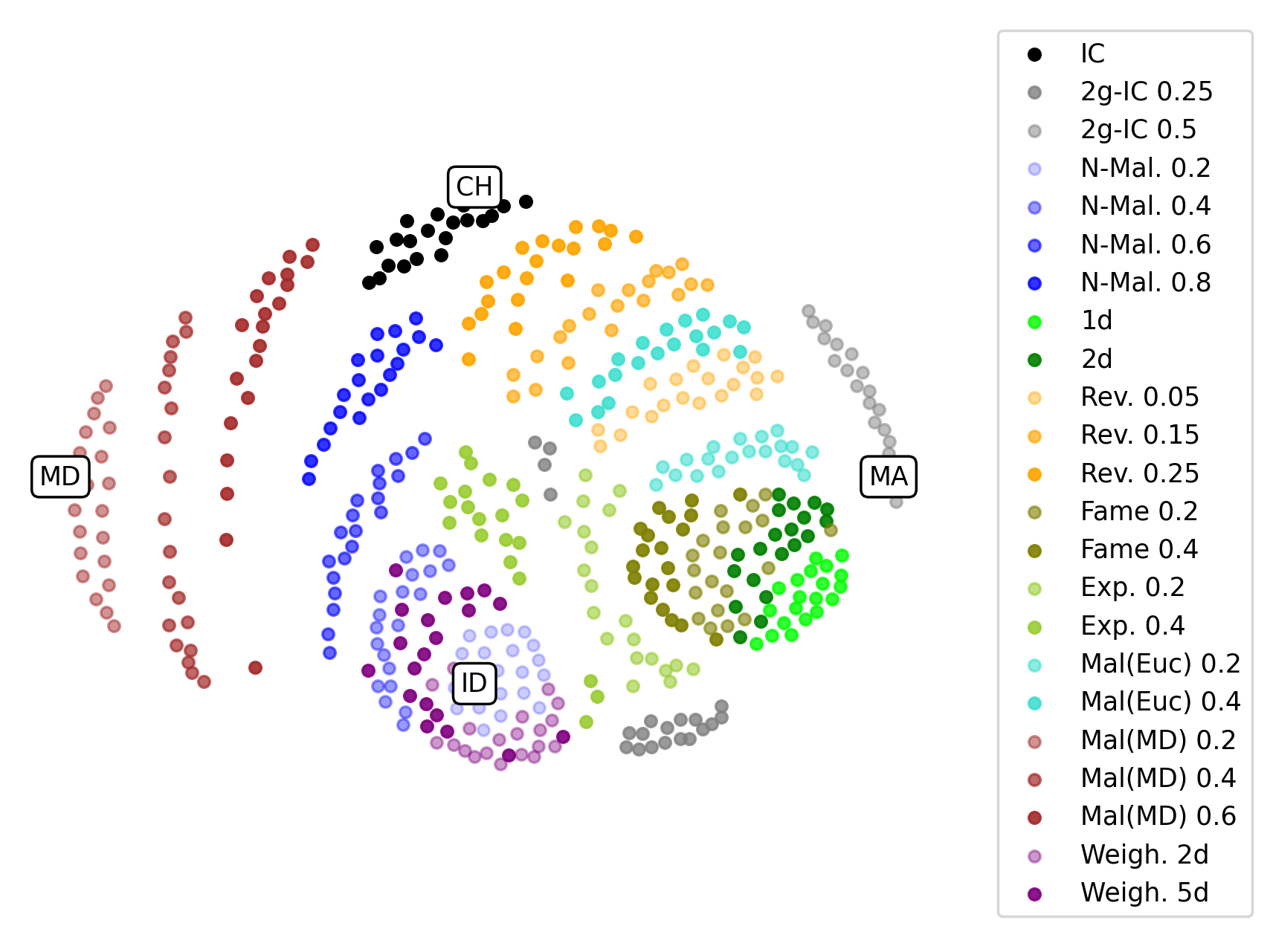}
         \caption{$50$ agents}
         \label{fig:map_50}
     \end{subfigure}
      \begin{subfigure}[b]{0.49\textwidth}
         \centering
         \includegraphics[width=7cm]{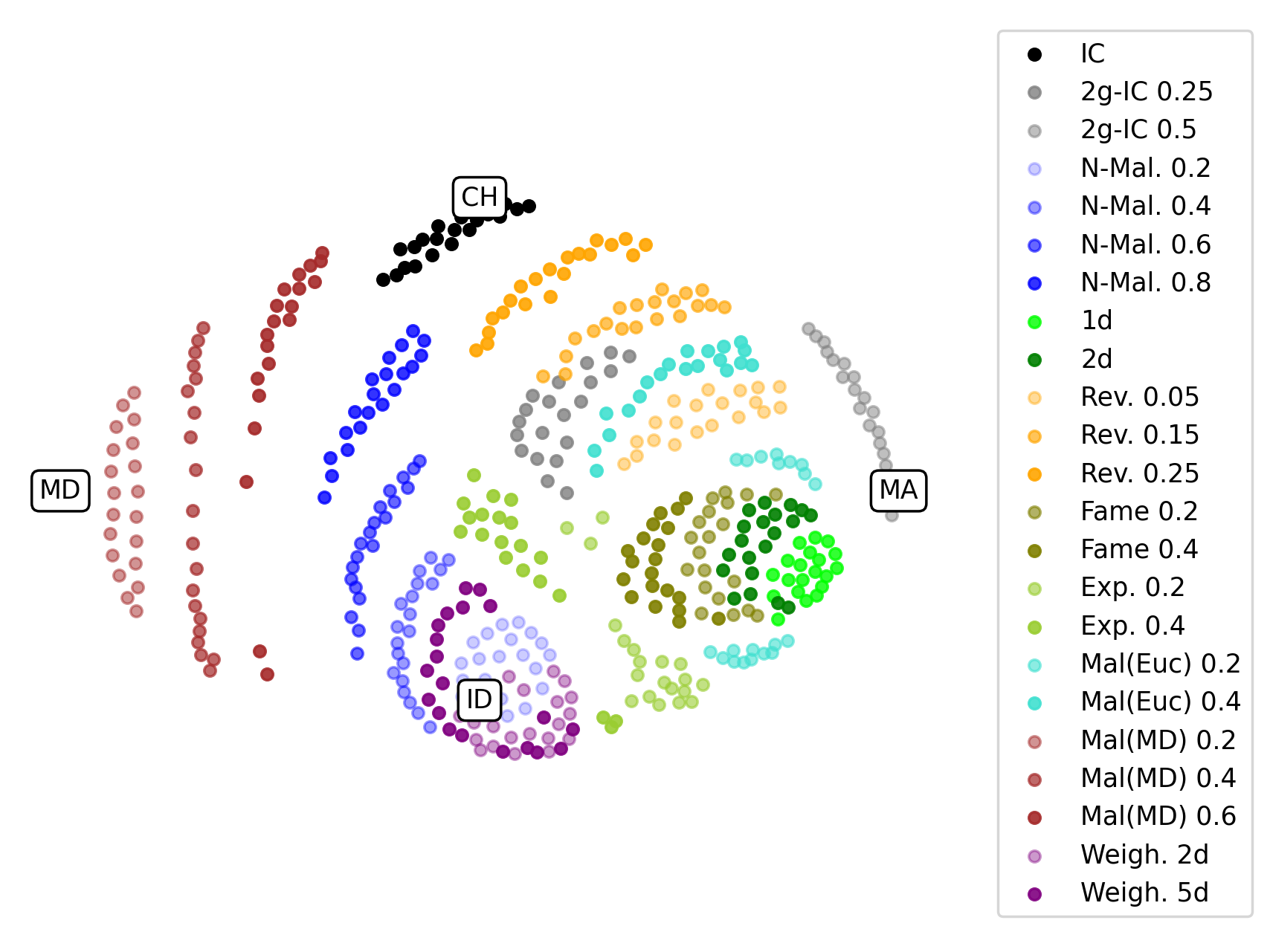}
         \caption{$102$ agents}
         \label{fig:map_102}
     \end{subfigure}
    \label{fig:other_maps}
    \caption{Map of SR instances for different numbers of agents.}
\end{figure*}

\end{document}